\documentclass[11pt]{amsart}

\usepackage{color}
\usepackage{comment}
\usepackage[margin=1in]{geometry}
\usepackage{graphicx}
\usepackage{float}
\usepackage{yfonts}

\newtheorem{rhp}{Riemann-Hilbert Problem}
\newtheorem{theorem}{Theorem}
\newtheorem{lemma}{Lemma}
\newtheorem{proposition}{Proposition}
\theoremstyle{remark}
\newtheorem{rmk}{Remark}

\def\eq{\begin{equation}}
\def\endeq{\end{equation}}
\def\bbm{\begin{bmatrix}}
\def\ebm{\end{bmatrix}}
\def\bpm{\begin{pmatrix}}
\def\epm{\end{pmatrix}}

\numberwithin{equation}{section}

\title[NLS $n$-solitons]{Large-order asymptotics for multiple-pole solitons \\ of the focusing nonlinear Schr\"odinger equation}

\author{Deniz Bilman}

\author{Robert Buckingham}

\thanks{D. Bilman was supported by an AMS-Simons travel grant.  
R. Buckingham was supported by the National Science Foundation through grant 
DMS-1615718.}

\begin{document}

\begin{abstract}
We analyze the large-$n$ behavior of soliton solutions of the integrable 
focusing nonlinear Schr\"odinger equation with associated spectral data 
consisting of a single pair of conjugate poles of order $2n$.  Starting from 
the zero background, we generate multiple-pole solitons by $n$-fold 
application of Darboux transformations.   The resulting functions are encoded 
in a Riemann-Hilbert problem using the robust inverse-scattering transform 
method recently introduced by Bilman and Miller.  For moderate values of $n$ 
we solve the Riemann-Hilbert problem exactly.  With appropriate scaling, the 
resulting plots of exact solutions reveal semiclassical-type behavior, 
including regions with high-frequency modulated waves and quiescent regions.  
We compute the boundary of the quiescent regions exactly and use the 
nonlinear steepest-descent method to prove the asymptotic limit of the 
solitons is zero in these regions.  Finally, we study the behavior of the 
solitons in a scaled neighborhood of the central peak with amplitude 
proportional to $n$.  We prove that in a local scaling the solitons converge 
to functions satisfying the second member of the Painlev\'e-III hierarchy 
in the sense of Sakka.  This function is a generalization of a function 
recently identified by Suleimanov in the context of geometric optics and by 
Bilman, Ling, and Miller in the context of rogue wave solutions to the 
focusing nonlinear Schr\"odinger equation.
\end{abstract}

\maketitle


\section{Introduction}

The one-dimensional focusing cubic nonlinear Schr\"odinger equation 
\eq
\label{nls}
i\psi_t + \frac{1}{2}\psi_{xx} + |\psi|^2\psi = 0, \quad x,t\in\mathbb{R}
\endeq
is a standard model for the evolution of quasi-monochromatic waves in 
weakly nonlinear dispersive media \cite{BenneyN:1967} with applications 
including fluid dynamics \cite{Zakharov:1968} and nonlinear optics 
\cite{ChiaoGT:1964}.  Solutions of \eqref{nls} are well known to exhibit 
highly structured multiscale coherent wave patterns that may be viewed as 
the focusing counterpart to dispersive shock waves occuring in systems with 
hyperbolic modulation equations such as the defocusing nonlinear Schr\"odinger 
and Korteweg-de Vries equations \cite{ElH:2016}.  A standard mechanism 
for the generation of rapid oscillations of $|\psi(x,t)|$ from smooth Cauchy 
data is the dispersive regularization of a so-called gradient catastrophe 
\cite{BertolaT:2013}.  
These phenomena have been extensively studied using the semiclassically 
scaled problem $i\epsilon\psi_t+\frac{\epsilon}{2}\psi_{xx}+|\psi|^2\psi=0$ 
with $\epsilon$-independent Cauchy data $\psi_0(x,0)$.  As $\epsilon\to 0$, 
the initial condition can be better approximated by reflectionless initial 
data consisting of $\mathcal{O}(\epsilon^{-1})$ solitons.  These 
so-called semiclassical soliton ensembles can be computed explicitly for 
moderately small values of $\epsilon$, and studied asymptotically via 
Riemann-Hilbert techniques in the limit $\epsilon\to 0$ 
\cite{KamvissisMM:2003,LyngM:2007}.  Solutions asymptotically display rapid 
oscillations of period $\mathcal{O}(\epsilon)$ in fixed 
($\epsilon$-independent) regions of the space-time plane.  

The nonlinear Schr\"odinger equation \eqref{nls} is completely integrable 
\cite{ZakharovS:1972}, and each initial condition with sufficient smoothness 
and decay is associated to scattering data consisting of poles and connection 
coefficients (encoding solitons) and a reflection coefficient (encoding 
radiation).  The scattering data for a semiclassical soliton ensemble 
consists of $\mathcal{O}(\epsilon^{-1})$ \emph{simple} poles (and associated 
connection coefficients, but the reflection coefficient is zero).  On the other 
hand, it has been known since the original work of 
Zakharov and Shabat \cite{ZakharovS:1972} that \eqref{nls} has soliton 
solutions with spectral data consisting of \emph{higher-order} poles.  We refer 
to a reflectionless solution of \eqref{nls} with scattering data consisting 
of a pair of poles of order $m$ as an \emph{$m^\text{th}$-order pole soliton} 
or, more generally, a \emph{multiple-pole soliton}.  In light of 
the rich mathematical structure displayed by solutions with scattering data 
consisting of a large number of simple poles, along with the fact that 
multiple-pole solitons can be generated by the coalescence of simple poles, 
it is natural to study the asymptotic behavior of $m^\text{th}$-order pole 
solitons as $m\to\infty$.  We show that multiple-pole solitons provide 
an alternate mechanism for generating behavior qualitatively similar to 
dispersive shock waves (see Figures 
\ref{fig-far-field-2d}--\ref{fig-zero-region-boundary}).

Previous studies of muliple-pole soliton solutions of \eqref{nls} have 
primarily focused on algebraic as opposed to asymptotic aspects 
\cite{AktosunDv:2007,GagnonS:1994,Martines:2017,Tanaka:1975}.  Olmedilla 
\cite{Olmedilla:1987} and Schiebold \cite{Schiebold:2017} studied the 
long-time asymptotic behavior (while keeping the pole order fixed).  Solitons 
associated to higher-order poles have also been studied for the modified 
Korteweg-de Vries equation \cite{WadatiO:1982}, the sine-Gordon equation 
\cite{AktosunDv:2010,Poppe:1983,TsuruW:1984}, the 
Caudrey-Dodd-Gibbon-Sawada-Kotera equation \cite{FuchssteinerA:1987}, the 
Kadomtsev-Petviashvili I equation \cite{AblowitzCTV:2000,VillarroelA:1999}, 
the $N$-wave system \cite{Shchesnovich:2003}, the complex short-pulse 
equation \cite{LingFZ:2016}, and the coupled Sasa-Satsuma system 
\cite{KuangZ:2017}.  Vinh \cite{Vinh:2017} has recently studied analogues of 
higher-order solitons for nonintegrable generalized Korteweg-de Vries 
equations.  

The recently introduced robust inverse-scattering transform 
\cite{BilmanM:2017} (see \S\ref{subsec-robust} for more details) provides the 
tools necessary to analyze the large-order behavior of multiple-pole 
solitons.  Bilman, Ling, and Miller \cite{BilmanLM:2018} used the robust 
inverse-scattering transform to study the large-order asymptotic behavior of 
multiple-pole soliton solutions of 
\eq
i\psi_t + \frac{1}{2}\psi_{xx} + (|\psi|^2-1)\psi = 0, \quad x,t\in\mathbb{R}
\label{nls2}
\endeq 
with non-decaying initial conditions (i.e. the nonlinear Schr\"odinger 
equation expressed in a rotating frame).  Here we adapt the robust 
inverse-scattering transform to analyze multiple-pole soliton solutions of 
\eqref{nls} generated by repeated Darboux transformations.  
Specifically, we fix a 
Darboux transformation that takes a given solution $\psi_0(x,t)$ to 
\eqref{nls} and generates a new solution $\widetilde{\psi}_0(x,t)$ with the 
same Beals-Coifman scattering data except for the addition of double poles at 
points $\xi$ and $\xi^*$.  If the Beals-Coifman scattering data for 
$\psi_0(x,t)$ already has poles of order $m$ at $\xi$ and $\xi^*$, then the 
Beals-Coifman scattering data for $\widetilde{\psi}_0(x,t)$ will have poles of 
order $m+2$ at these points.  We start with the trivial initial condition 
$\psi^{[0]}(x,t)\equiv 0$ and repeatedly apply the same Darboux 
transformation $n$ times to obtain a solution $\psi^{[2n]}(x,t)$ with order 
$2n$ poles at $\xi$ and $\xi^*$.  

As one might expect, the 
global behavior (the \emph{far field}) is markedly different for the 
multiple-pole soliton solutions of \eqref{nls} studied here and the 
multiple-pole soliton solutions of \eqref{nls2} studied in 
\cite{BilmanLM:2018}.  Nevertheless, the multiple-pole soliton solutions of 
\eqref{nls} and \eqref{nls2} both have a peak 
of amplitude proportional to the pole order $m$.  In 
\cite{BilmanLM:2018}, it was shown for \eqref{nls2} that the local behavior 
in a scaled neighborhood of this peak (the \emph{near field}) is given by a 
certain Painlev\'e function.  We show that for \eqref{nls} the near-field 
behavior is described by a new family of Painlev\'e solutions that agree with 
the Painlev\'e function in \cite{BilmanLM:2018} for special parameter 
values.  We now summarize our results.

\subsection{Far-field results}
Fix a pole location $\xi\in\mathbb{C}^+$, a vector of connection coefficients 
${\bf c}\equiv(c_1,c_2)\in(\mathbb{C}^*)^2$, and a non-negative integer $n$.  
Define $D_0\subset\mathbb{C}$ to be a circular disk centered at the origin 
containing $\xi$ in its interior.  Let 
${\bf M}^{[n]}(\lambda;x,t;{\bf c}) \equiv {\bf M}^{[n]}(\lambda;x,t)$ be the 
unique solution of the following Riemann-Hilbert problem (for uniqueness see 
the argument in, for example, \cite[Theorem 2.4]{BilmanM:2017}).  
\begin{rhp}
Let $(x,t)\in\mathbb{R}^2$ be arbitrary parameters, and let $n\in\mathbb{Z}_{\geq 0}$. Find the unique $2\times 2$ matrix-valued function $\mathbf{M}^{[n]}(\lambda;x,t)$ with the following properties:
\begin{itemize}
\item[]\textbf{Analyticity:} $\mathbf{M}^{[n]}(\lambda;x,t)$ is analytic for $\lambda\in\mathbb{C}\setminus \partial D_0$, and it takes continuous boundary values from the interior and exterior of $\partial D_0$.
\item[]\textbf{Jump condition:} The boundary values on the jump contour $\partial D_0$ (oriented clockwise) are related as 
\begin{equation}
\mathbf{M}_{+}^{[n]}(\lambda;x,t) = \mathbf{M}_{-}^{[n]}(\lambda;x,t)e^{-i(\lambda x+\lambda^2 t)\sigma_3}\mathbf{S} \left( \frac{\lambda-\xi}{\lambda-\xi^*}\right)^{n\sigma_3} \mathbf{S}^{-1}e^{i(\lambda x+\lambda^2 t)\sigma_3},\quad\lambda\in\partial D_0,
\label{eq:jump-rhp-M}
\end{equation}
where
\eq
{\bf S}:=\frac{1}{|{\bf c}|}\bpm c_1 & -c_2^* \\ c_2 & c_1^* \epm.
\label{S-def}
\endeq
\item[]\textbf{Normalization:} $\mathbf{M}^{[n]}(\lambda;x,t)\to\mathbb{I}$ as $\lambda\to\infty$.
\end{itemize}
\label{rhp:psi-n}
\end{rhp}
\noindent
The $2n^\text{th}$-order pole solitons we study are defined by 
\eq
\label{psi-from-M}
\psi^{[2n]}(x,t;{\bf c}) \equiv \psi^{[2n]}(x,t) := 2i\lim_{\lambda\to\infty}\lambda[{\bf M}^{[n]}(\lambda;x,t;{\bf c})]_{12}
\endeq
(see Remark \ref{rmk-pole-order} in \S\ref{subsec-iteration} for the 
explanation of why the pole order is necessarily even).  
Introduce the scaled space and time variables $\chi$ and $\tau$ by 
\eq
\label{chi-tau}
\chi:=\frac{x}{n}, \quad \tau:=\frac{t}{n}.
\endeq
As illustrated in Figures \ref{fig-far-field-2d} and 
\ref{fig-far-field-2d-c5}, as $n\to\infty$ the $\chi\tau$-plane is 
partitioned into well-defined regions in which the leading-order behavior of 
$\psi^{[2n]}(n\chi,n\tau)$ is different.  Figure 
\ref{fig-far-field-2d} illustrates this convergence for $\xi=i$ and 
${\bf c}=(1,1)$.  
\begin{figure}[H]
\begin{center}
\includegraphics[height=2in]{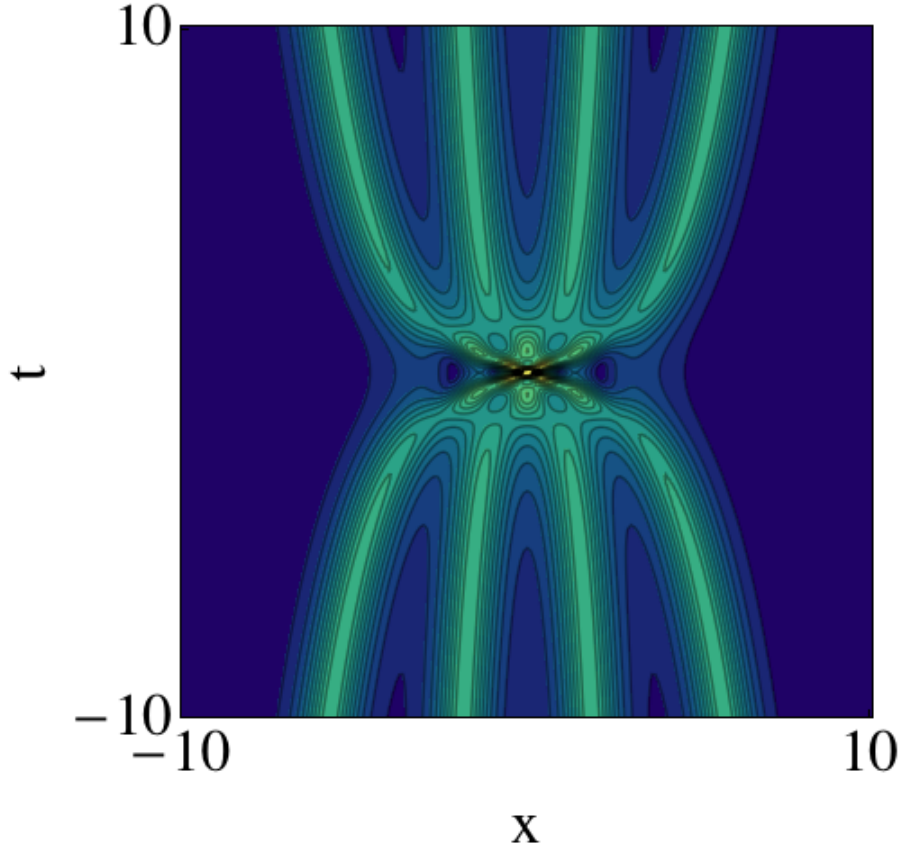}
\includegraphics[height=2in]{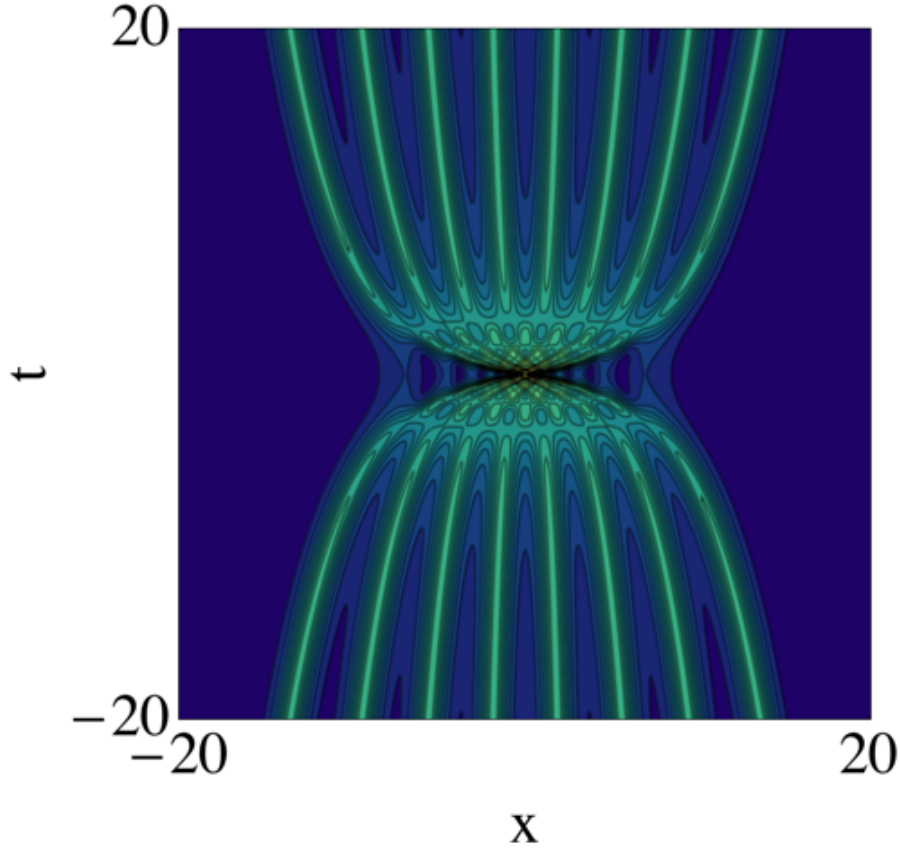}
\includegraphics[height=2in]{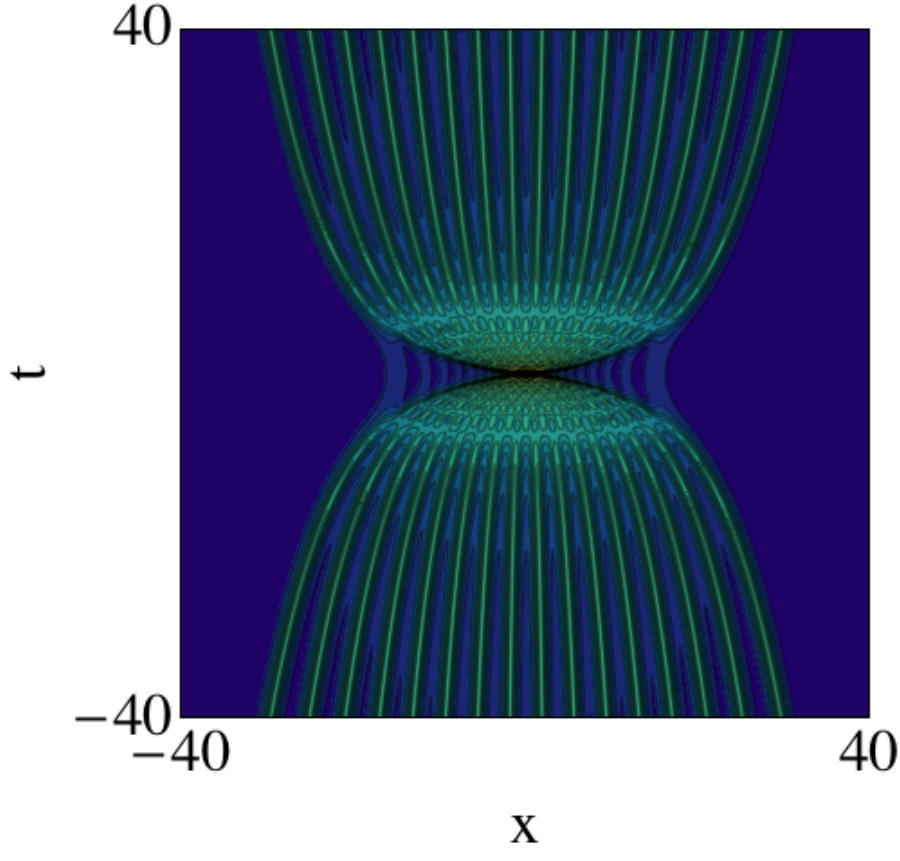}
\caption{Plots of $|\psi^{[2n]}(x,t;(1,1))|$ for $-5n\leq t\leq 5n$ and $-5n\leq x\leq 5n$ (i.e. $-5\leq\tau\leq 5$ and $-5\leq\chi\leq 5$), where $\psi^{[2n]}(x,t;(1,1))$ is a multiple-pole soliton solution of the nonlinear Schr\"odinger equation \eqref{nls}.  In each plot $c_1=c_2=1$ and $\xi=i$.  \emph{Left:} $n=2$, \emph{Center:} $n=4$. \emph{Right:} $n=8$.}
\label{fig-far-field-2d}
\end{center}
\end{figure}
Figures \ref{fig-far-field2-1d} and \ref{fig-far-field-1d} show time slices 
of $|\psi^{[2n]}(n\chi,n\tau)|$ at $\tau=\frac{3}{8}$ and $\tau=5$, 
respectively.  Although it is beyond the scope of this paper, these plots 
suggest there are at least three different nonzero leading-order behaviors.  
\begin{figure}[H]
\begin{center}
\includegraphics[height=1.35in]{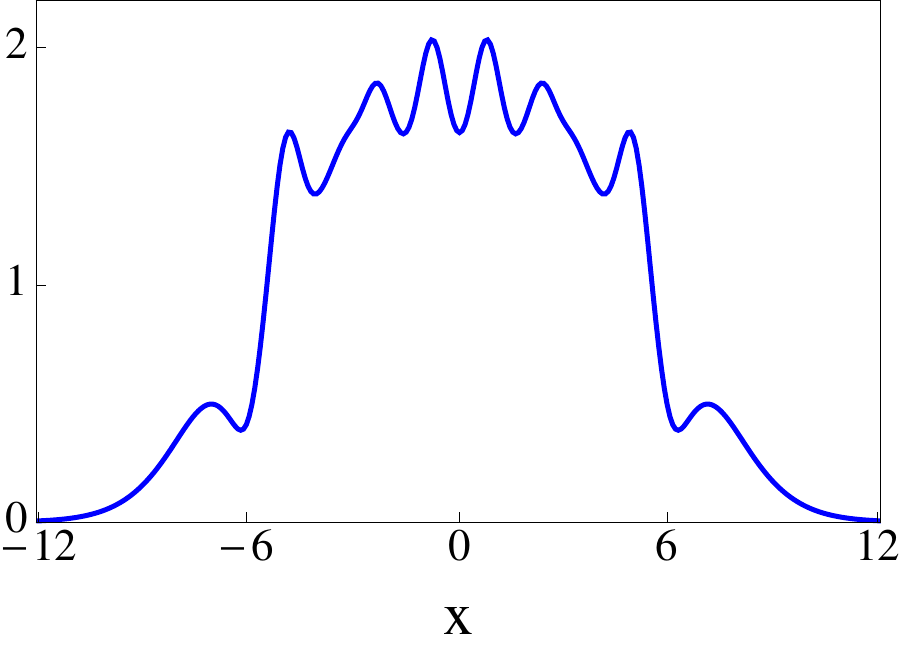}
\includegraphics[height=1.35in]{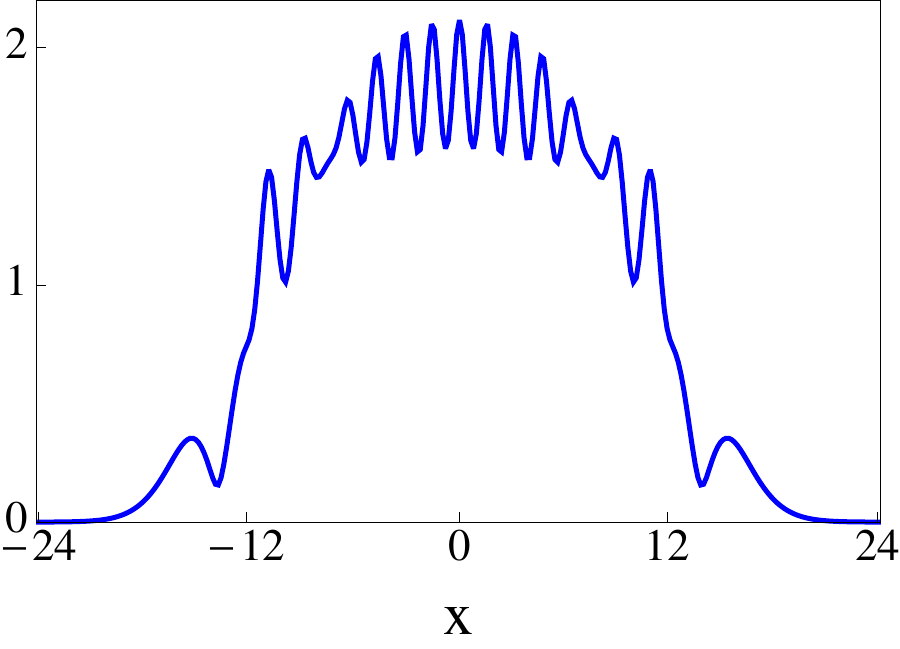}
\includegraphics[height=1.35in]{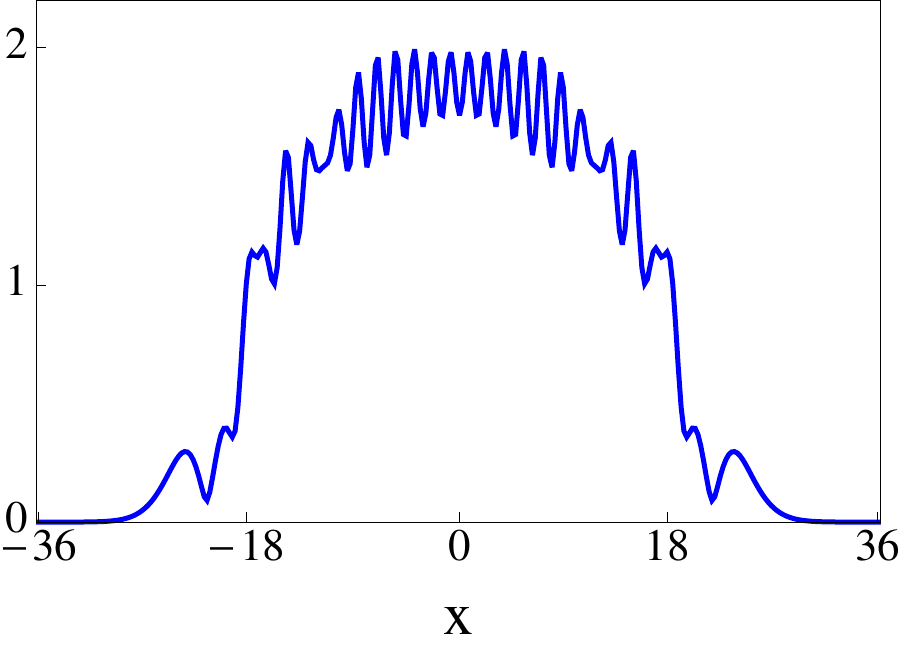}
\caption{Plots of $|\psi^{[2n]}(x,\frac{3}{8}n;(1,1))|$ for $t=\frac{3}{8}n$ and $-3n\leq x\leq 3n$ (i.e. $\tau=\frac{3}{8}$ and $-3\leq\chi\leq 3$), where $\psi^{[2n]}(x,t;(1,1))$ is a multiple-pole soliton solution of the nonlinear Schr\"odinger equation \eqref{nls}.  In each plot $c_1=c_2=1$ and $\xi=i$.  \emph{Left:} $n=4$. \emph{Center:} $n=8$. \emph{Right:} $n=12$.}
\label{fig-far-field2-1d}
\end{center}
\end{figure}
\begin{figure}[H]
\begin{center}
\includegraphics[height=1.35in]{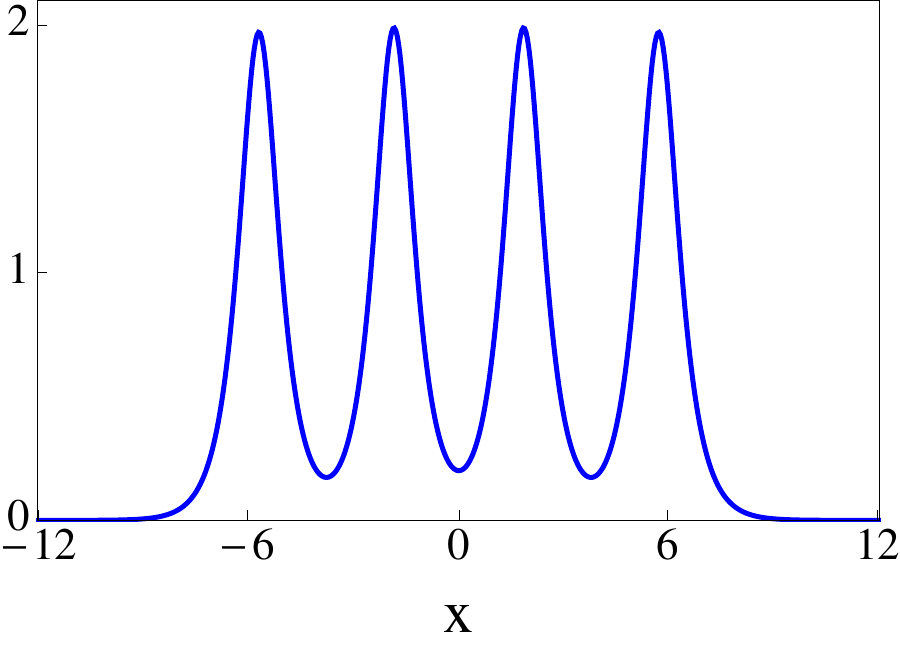}
\includegraphics[height=1.35in]{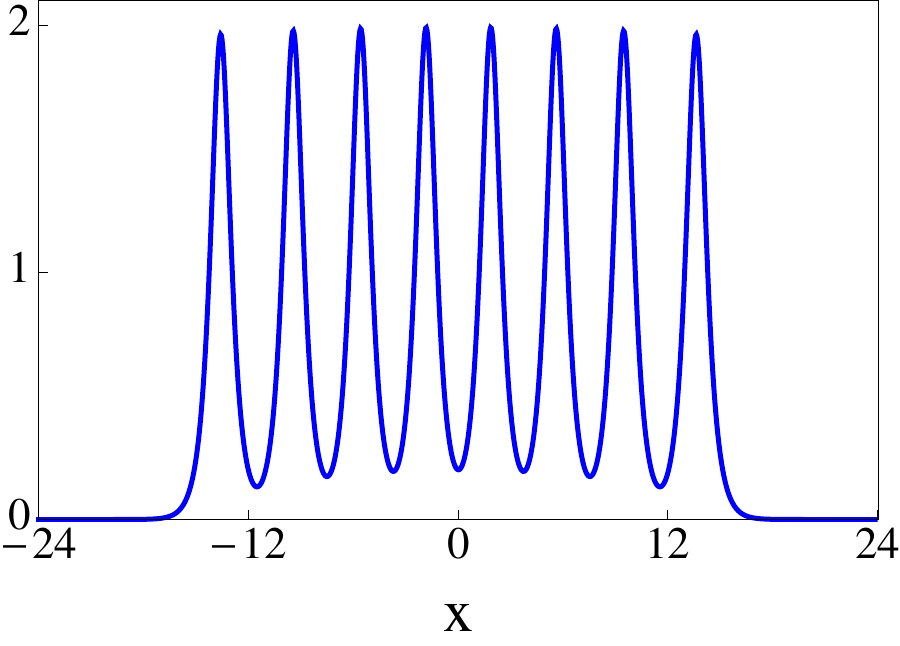}
\includegraphics[height=1.35in]{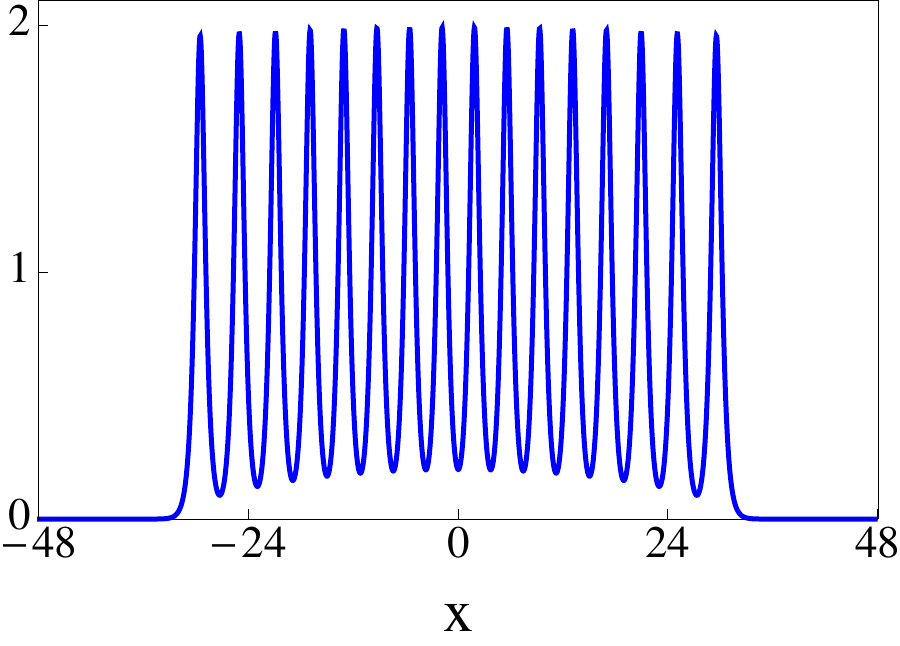}
\caption{Plots of $|\psi^{[2n]}(x,5n;(1,1))|$ for $t=5n$ and $-6n\leq x\leq 6n$ (i.e. $\tau=5$ and $-6\leq\chi\leq 6$), where $\psi^{[2n]}(x,t;(1,1))$ is a multiple-pole soliton solution of the nonlinear Schr\"odinger equation \eqref{nls}.  In each plot $c_1=c_2=1$ and $\xi=i$.   \emph{Left:} $n=2$. 
\emph{Center:} $n=4$. \emph{Right:} $n=8$.}
\label{fig-far-field-1d}
\end{center}
\end{figure}

In Figure \ref{fig-far-field-2d-c5} we illustrate the effect of changing 
${\bf c}$ by plotting $\psi^{[2n]}(n\chi,n\tau)$ with ${\bf c}=(1,5)$.  Much 
of the far-field structure remains the same.  The 
solution appears to still be converging to zero at all 
$(\chi,\tau)$ points at which the solution converged to zero with 
${\bf c}=(1,1)$ (this is made precise in Theorem 
\ref{thm-zero-reg}). Furthermore, if $|\tau|$ is sufficiently large, then the 
oscillatory structure appears unchanged.  However, there are noticable 
qualitative differences for $\chi$ and $\tau$ near the 
origin.  
\begin{figure}[H]
\begin{center}
\includegraphics[height=2in]{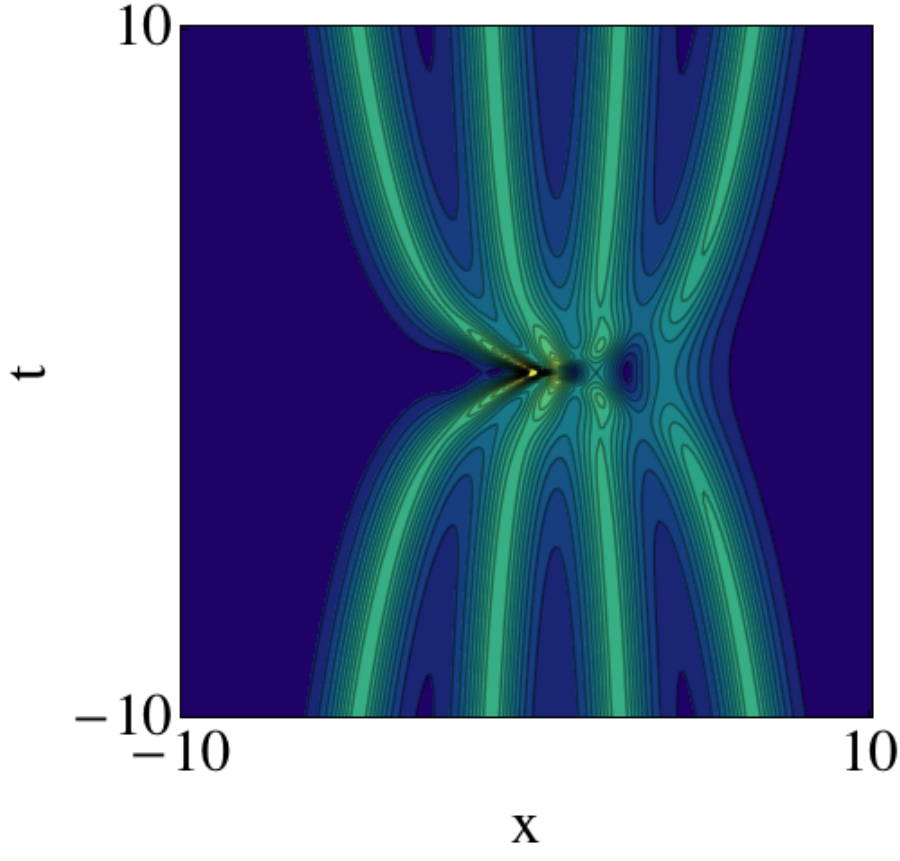}
\includegraphics[height=2in]{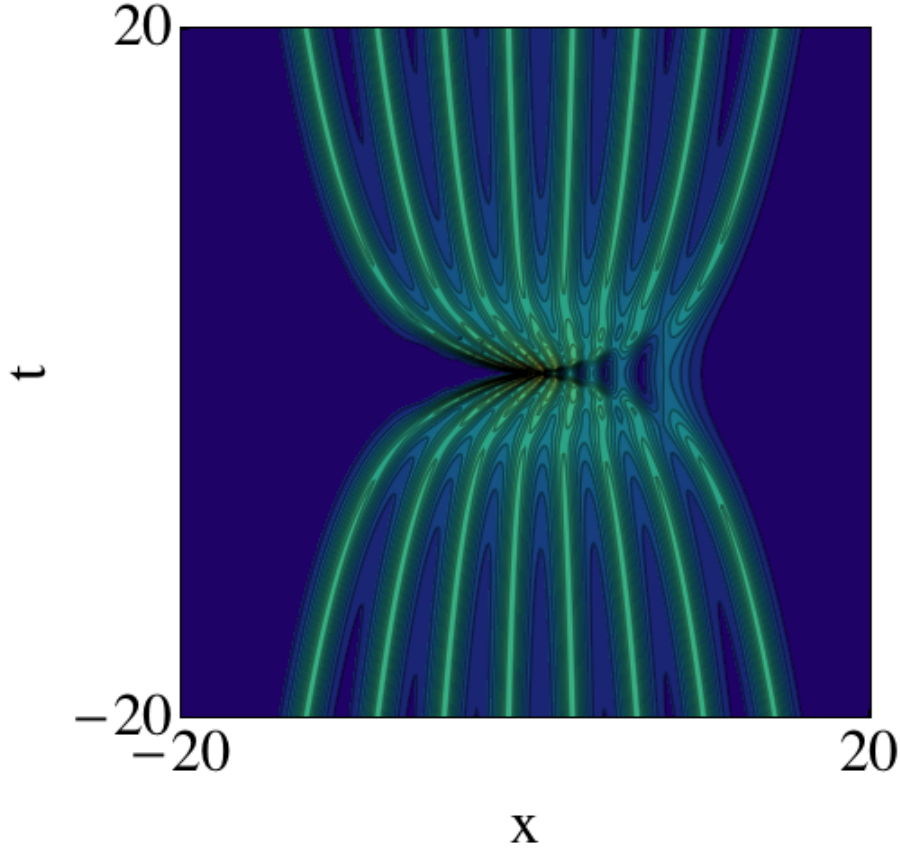}
\includegraphics[height=2in]{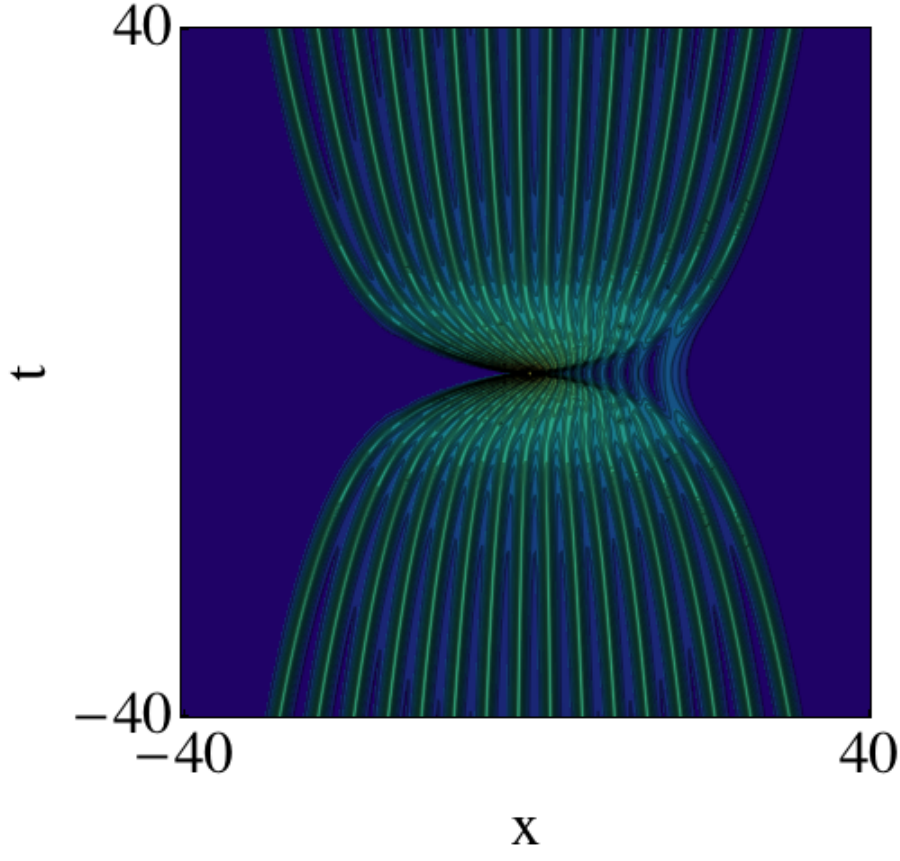}
\caption{Plots of $|\psi^{[2n]}(x,t;(1,5))|$ for $-5n\leq t\leq 5n$ and $-5n\leq x\leq 5n$ (i.e. $-5\leq\tau\leq 5$ and $-5\leq\chi\leq 5$), where $\psi^{[2n]}(x,t;(1,5))$ is a multiple-pole soliton solution of the nonlinear Schr\"odinger equation \eqref{nls}.  In each plot $c_1=1$, $c_2=5$, and $\xi=i$.   \emph{Left:} $n=2$. \emph{Center:} $n=4$. \emph{Right:} $n=8$.}
\label{fig-far-field-2d-c5}
\end{center}
\end{figure}
These differences are further illustrated in Figure 
\ref{fig-far-field2-1d-c5}, which shows a time slice at $\tau=\frac{3}{8}$ 
for ${\bf c}=(1,5)$.  We partially quantify these differences by studying the 
dependence on ${\bf c}$ in the near-field limit in Theorem 
\ref{thm-near-field}.
\begin{figure}[H]
\begin{center}
\includegraphics[height=1.35in]{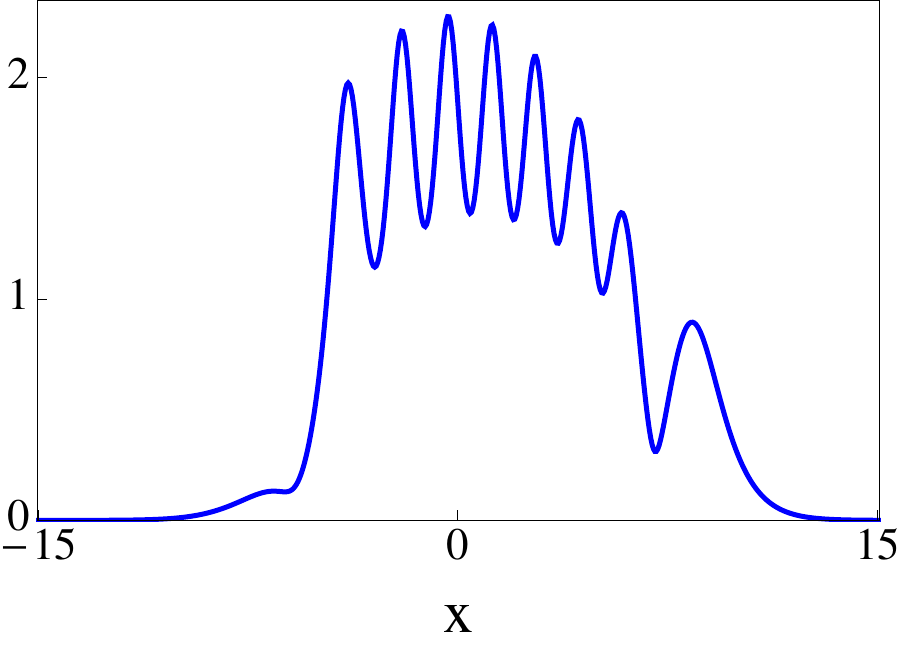}
\includegraphics[height=1.35in]{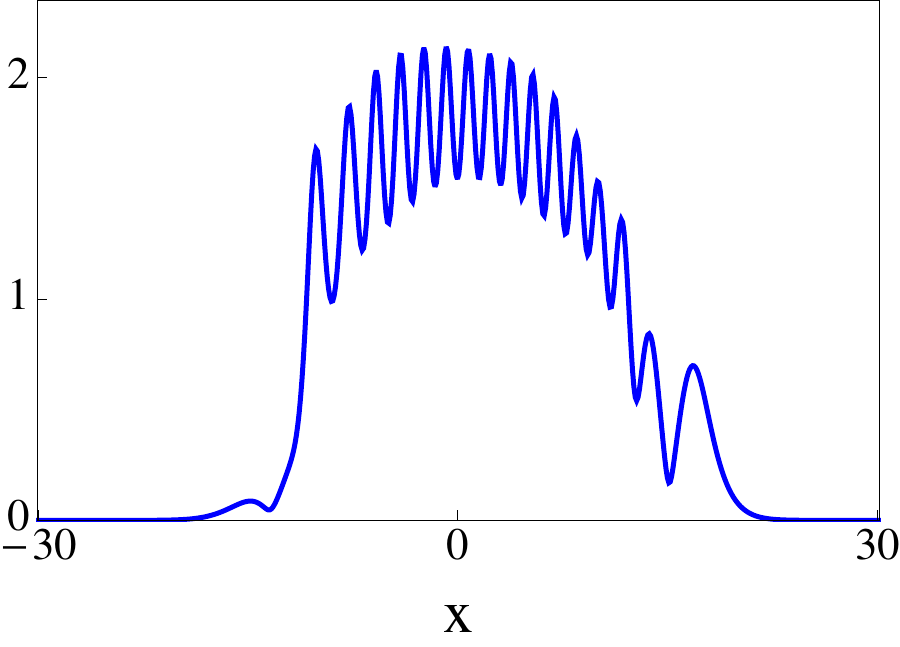}
\includegraphics[height=1.35in]{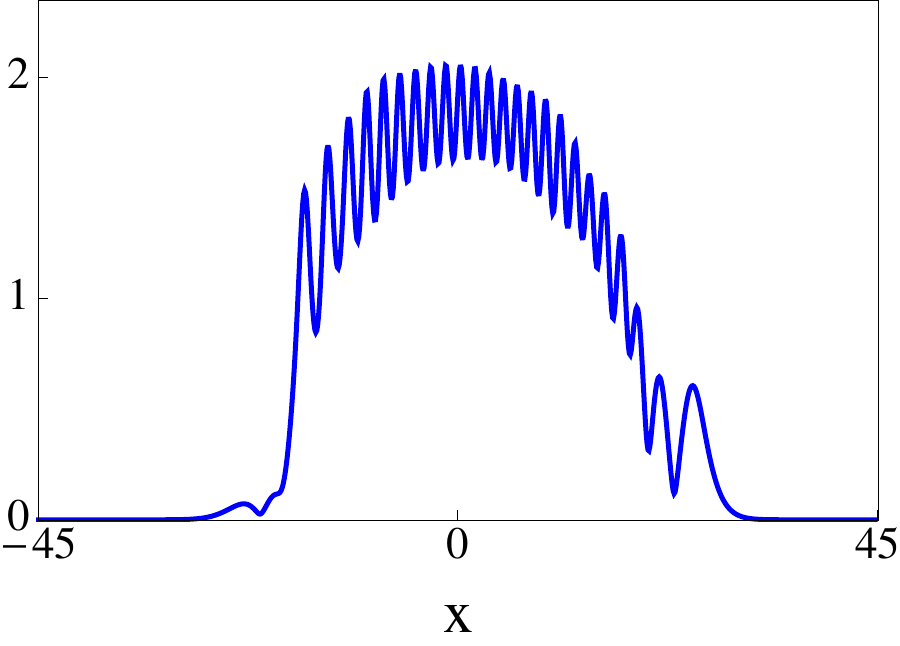}
\caption{Plots of $|\psi^{[2n]}(x,\frac{3}{8}n;(1,5))|$ for $t=\frac{3}{8}n$ and $-\frac{15}{4}n\leq x\leq \frac{15}{4}n$ (i.e. $\tau=\frac{3}{8}$ and $-\frac{15}{4}\leq\chi\leq\frac{15}{4}$), where $\psi^{[2n]}(x,t;(1,5))$ is a multiple-pole soliton solution of the nonlinear Schr\"odinger equation \eqref{nls}.  In each plot $c_1=1$, $c_2=5$, and $\xi=i$.   \emph{Left:} $n=4$.  \emph{Center:} $n=8$. \emph{Right:} $n=12$.}
\label{fig-far-field2-1d-c5}
\end{center}
\end{figure}

Our first result is Theorem \ref{thm-zero-reg}, showing there is a specific 
region in the $\chi\tau$-plane in which $\psi^{[2n]}(n\chi,n\tau)$ decays 
exponentially fast to zero as $n\to\infty$.  This region depends on $\xi$ but 
is independent of ${\bf c}$.  We provide an explicit (though transcendental) 
characterization of the boundary of the zero region.  

We state our results for $\xi=i$, which is sufficient for general $\xi$ by 
symmetry properties of \eqref{nls} since there is only one conjugate pair of 
poles in the scattering data.  Formulas for general $\xi$ are provided in 
\S\ref{sec-zero-reg}.  The boundary curve consists of two different types of 
points.  First, consider the locus of points $(\chi,\tau)\in\mathbb{R}^2$ 
satisfying 
\eq
\label{xi-i-boundary-quadratic}
16 \tau^4 + (8 \chi^2 - 72 \chi + 108)\tau^2 + (\chi^4 - 2\chi^3) = 0.
\endeq
Part of this locus is a smooth arc with endpoints 
$\left(\frac{9}{4},\pm\frac{3\sqrt{3}}{8}\right)$ and containing the point 
$(2,0)$.  Call this arc $\mathcal{L}_1$.  This arc appears to separate the 
zero region from a region in which the leading-order behavior of 
$\psi^{[2n]}(n\chi,n\tau)$ is specified by a model Riemann-Hilbert problem 
with a single band (suggesting non-oscillatory behavior).  

Next, given $\xi\in\mathbb{C}$, define the phase function 
\eq
\label{phi}
\varphi(\lambda;\chi,\tau) :=i(\lambda \chi+\lambda^2 \tau) + \log\left(\frac{\lambda-\xi^*}{\lambda-\xi}\right).
\endeq
The critical points of $\varphi(\lambda;\chi,\tau)$ are those values of 
$\lambda$ satisfying 
\eq
\label{phi-crit-eqn}
2\tau\lambda^3+\chi\lambda^2+2\tau\lambda+(\chi-2)=0.
\endeq
For $\tau=0$ and $\chi>2$, define 
\eq
\label{lambda-plus-tau-zero}
\lambda_+(\chi,0):=\left(\frac{2-\chi}{\chi}\right)^{1/2}
\endeq
so that $\Im(\lambda_+(\chi,0))>0$.  Note that if $\tau=0$, then 
$\lambda_+(\chi,0)$ satisfies \eqref{phi-crit-eqn}.  For $\tau\geq 0$, let 
$\lambda_+(\chi,\tau)$ be the solution of \eqref{phi-crit-eqn} that is the 
analytic continuation in $\tau$ of \eqref{lambda-plus-tau-zero}.  We restrict 
this definition to $(\chi,\tau)$ values that can be reached by a vertical path 
in the $\chi\tau$-plane starting at $(\chi,0)$ along which no two solutions of 
\eqref{phi-crit-eqn} coincide.  Then, the condition 
\eq
\label{zero-boundary2}
\Re(\varphi(\lambda_+(\chi,\tau),\chi,\tau))=0
\endeq
defines a simple unbounded curve in the first quadrant of the 
$\chi\tau$-plane 
with one endpoint at $\left(\frac{9}{4},\frac{3\sqrt{3}}{8}\right)$.  Denote 
this curve by $\mathcal{L}_2^+$ and define 
$\mathcal{L}_2^-:=\{(\chi,\tau):(\chi,-\tau)\in\mathcal{L}_2^+\}$.  
The curves $\mathcal{L}_2^+$ and $\mathcal{L}_2^-$ appear to separate the 
zero region from regions in which the leading-order behavior of 
$\psi^{[2n]}(n\chi,n\tau)$ is specified by a model Riemann-Hilbert problem 
with two bands (suggesting oscillatory behavior).  
Let $\mathcal{Z}_+$ denote the unbounded region in the 
$\chi\tau$-plane containing the ray $\chi>2$ and bounded by 
$\mathcal{L}_1\cup\mathcal{L}_2^+\cup\mathcal{L}_2^-$.  Also define 
$\mathcal{Z}_-:=\{(\chi,\tau):(-\chi,-\tau)\in\mathcal{Z}_+\}$.  Define the 
zero region $\mathcal{Z}:=\mathcal{Z}_+\cup\mathcal{Z}_-$.  
\begin{theorem}
\label{thm-zero-reg}
If $(\chi,\tau)\in\mathcal{Z}$,  
\eq
\label{zero-region-result}
\psi^{[2n]}(n\chi,n\tau;{\bf c}) = \mathcal{O}(e^{-dn})
\endeq
holds for some constant $d>0$.  
\end{theorem}
Theorem \ref{thm-zero-reg} holds for general $\xi\in\mathbb{C}^+$ with 
$\mathcal{Z}$ defined 
as in \S\ref{sec-zero-reg}.  The boundary curve is independent of both 
${\bf c}$ and $n$, although it does depend on $\xi$ (see \S\ref{sec-zero-reg} 
for the $\xi$-dependent formulas).  Figure \ref{fig-zero-region-boundary} 
illustrates the boundary curve for two choices of $\xi$.  
\begin{figure}[H]
\begin{center}
\includegraphics[height=2.5in]{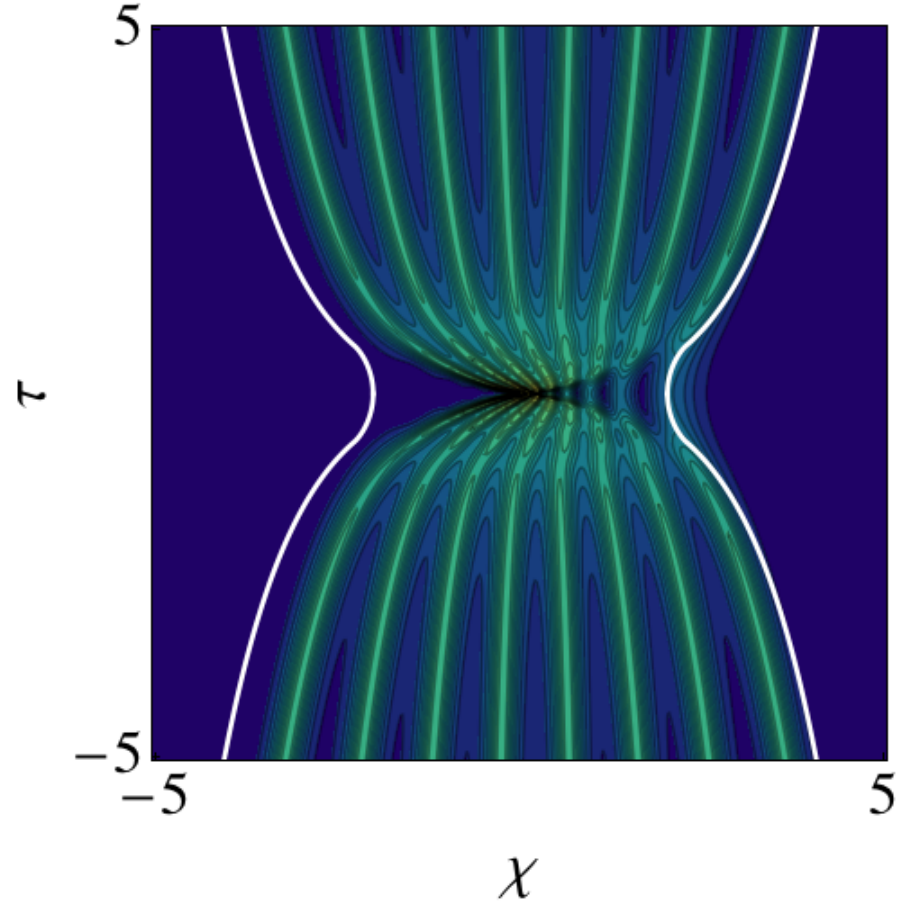}
\includegraphics[height=2.5in]{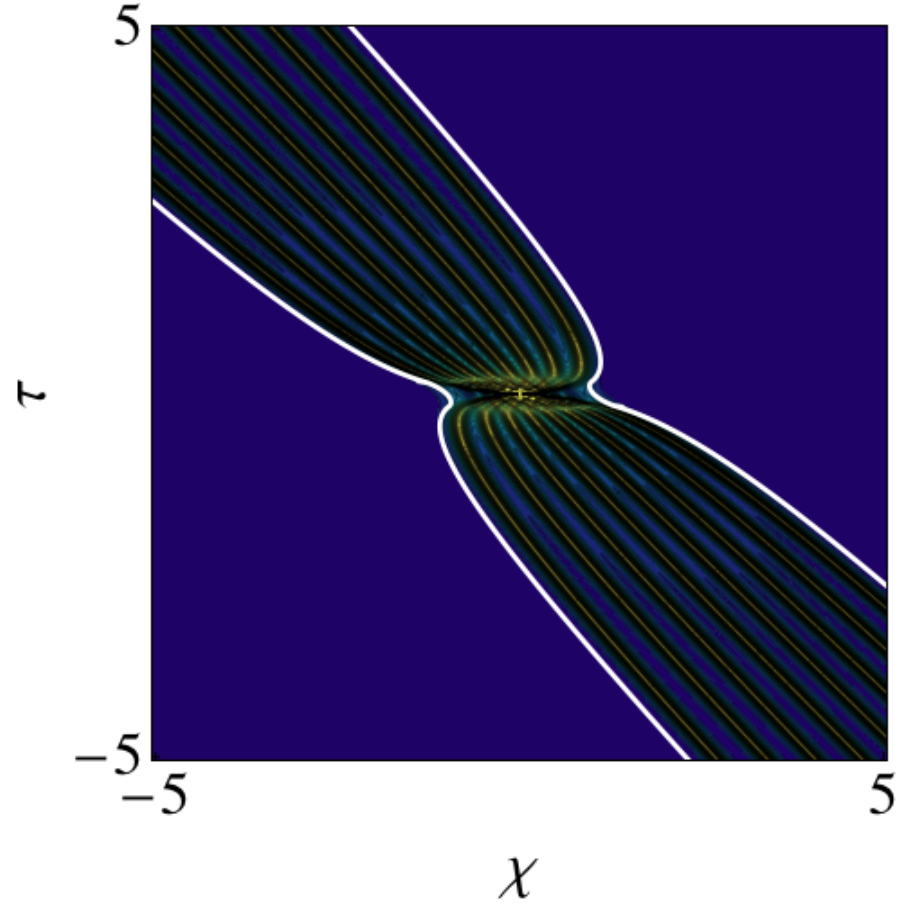}
\caption{The boundary of the zero region for two choices of $\xi$.  
\emph{Left:}  Plot of $|\psi^{[2n]}(n\chi,n\tau;(1,5))|$ with $c_1=1$, $c_2=5$, 
and $\xi=i$.  \emph{Right:}  Plot of $|\psi^{[2n]}(n\chi,n\tau;(1,1))|$ with 
$c_1=1$, $c_2=1$, and $\xi=\frac{1}{2}+2i$.  In both plots $n=4$, 
$-5\leq \chi\leq 5$, $-5\leq \tau\leq 5$, and $\psi^{[2n]}(x,t;{\bf c})$ 
is a multiple-pole soliton solution of the nonlinear Schr\"odinger equation 
\eqref{nls}.}
\label{fig-zero-region-boundary}
\end{center}
\end{figure}

\subsection{Near-field results}
\label{sec:near-field}
From plots such as Figure \ref{fig-origin-3d}, it is evident that the 
qualitative behavior of $\psi^{[2n]}(x,t;{\bf c})$ near $(x,t)=(0,0)$ is 
distinctly different from elsewhere in the space-time plane and is dominated 
by a single peak (with shape dependent on ${\bf c}$) with amplitude of 
$\mathcal{O}(n)$ for $n$ large.  
\begin{figure}[H]
\begin{center}
\includegraphics[height=2.4in]{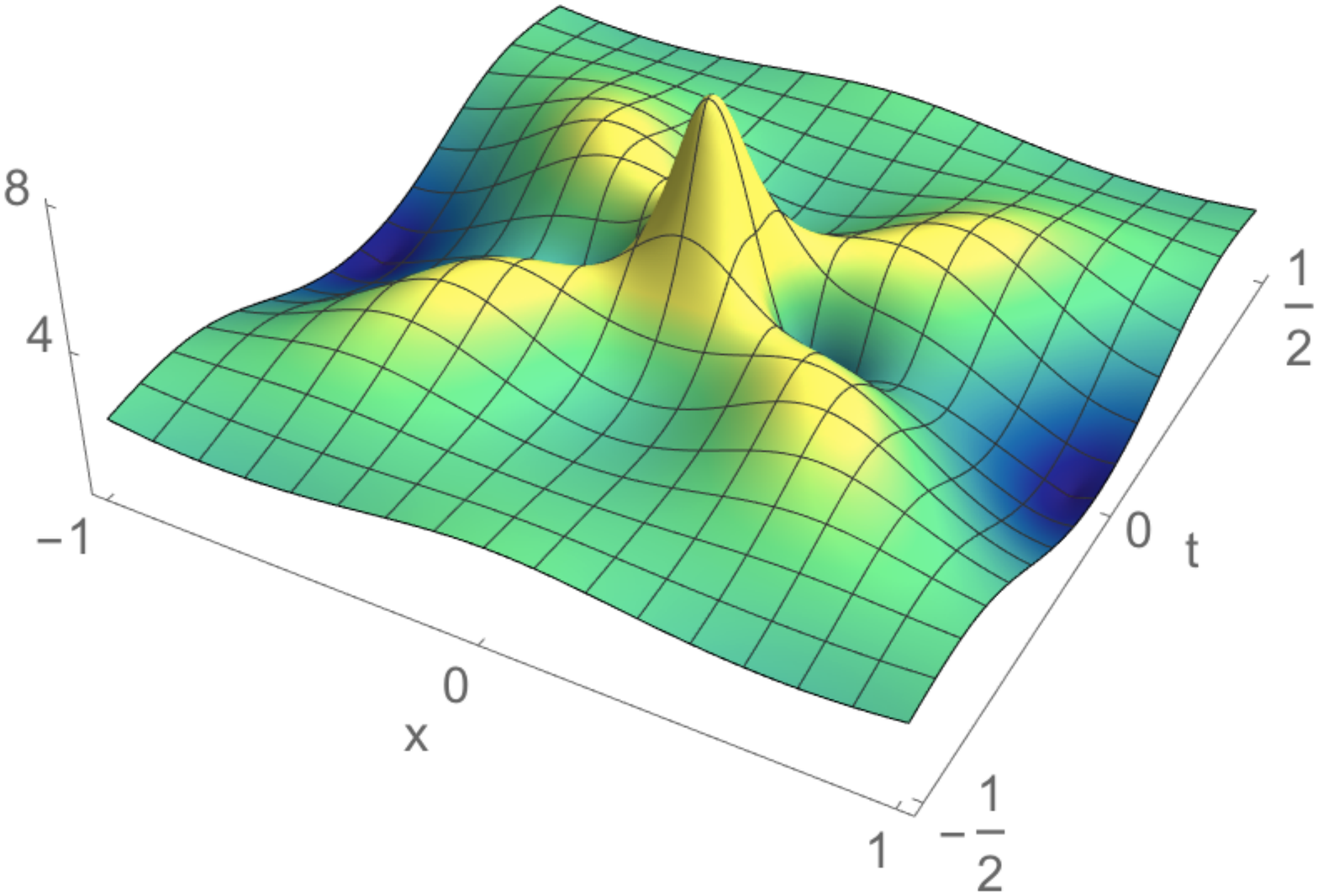}
\includegraphics[height=2.4in]{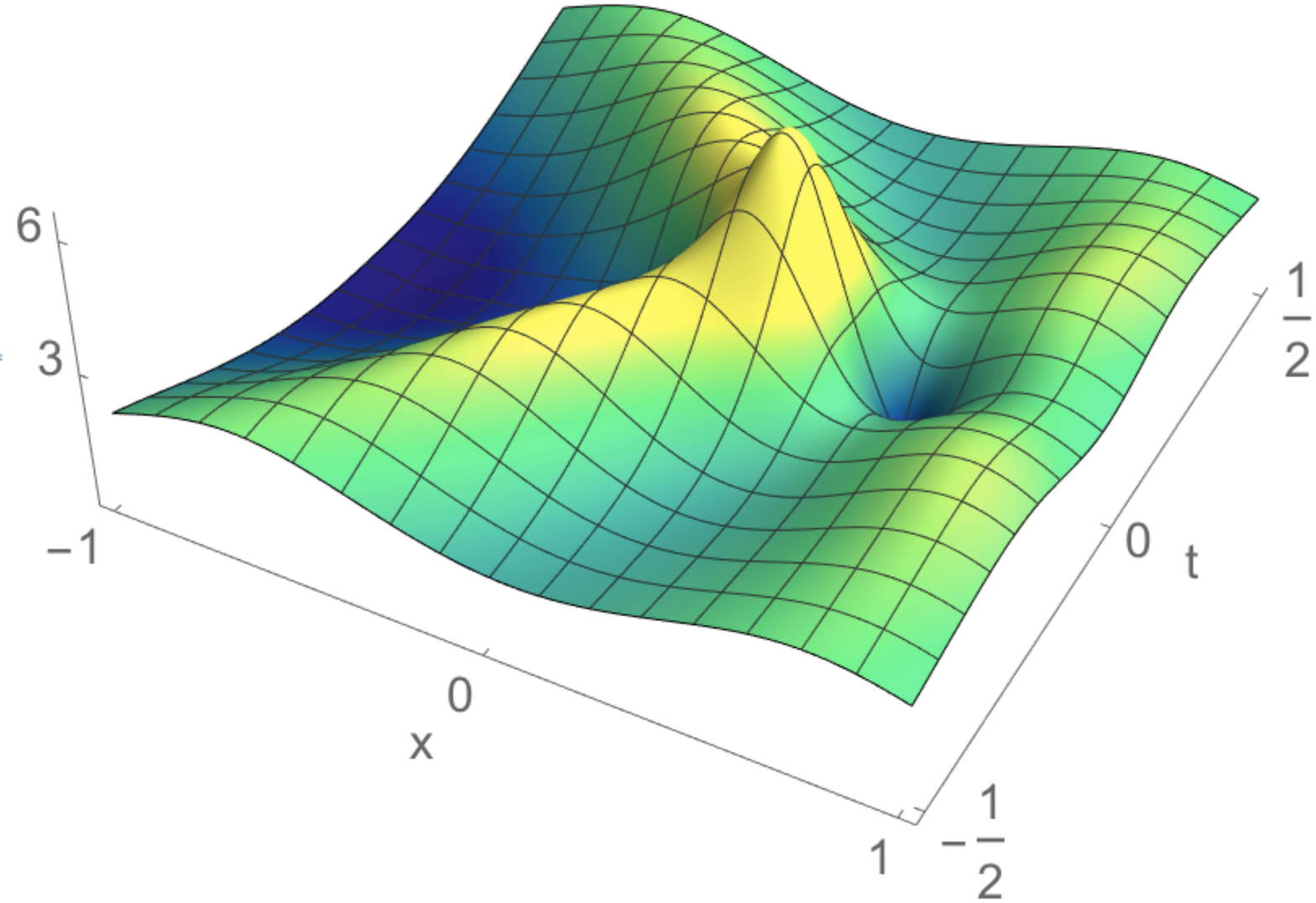}
\caption{Plots of $|\psi^{[2]}(x,t;{\bf c})|$ illustrating the behavior near 
$(x,t)=(0,0)$, where $\psi^{[2]}(x,t;{\bf c})$ is a $2^\text{nd}$-order pole 
soliton solution of the nonlinear Schr\"odinger equation \eqref{nls}. 
\emph{Left:} ${\bf c}=(1,1)$.  \emph{Right:} ${\bf c}=(1,5)$.}
\label{fig-origin-3d}
\end{center}
\end{figure}
We study the behavior in this region by defining the rescaled variables 
\eq
\label{near-field-scaling}
X:=nx, \quad T:=n^2t.
\endeq
Our main result is that, in the large-$n$ limit, the function 
$\frac{1}{n}\psi^{[2n]}(\frac{X}{n},\frac{T}{n^2};{\bf c})$ is 
well-approximated by a 
function $\Psi(X,T;{\bf c})\equiv\Psi(X,T)$ that satisfies Painlev\'e-III 
hierarchy equations in the sense of Sakka \cite{Sakka:2009}.  The functions 
$\Psi(X,T;(1,1))$ and $\Psi(X,T;(1,-1))$ (with $\xi=i$) were first identified 
recently in two different but related contexts,  self-focusing in nonlinear 
geometric optics \cite{Suleimanov:2017} and rogue waves \cite{BilmanLM:2018}. 
The work \cite{BilmanLM:2018}, which is more closely related to the current 
study, analyzes solutions of the focusing nonlinear Schr\"odinger equation 
\eqref{nls2} generated from repeated Darboux transformations applied to the 
constant solution $\psi(x,t)\equiv 1$.  The resulting solutions, referred to 
in \cite{BilmanLM:2018} as fundamental rogue waves, can be viewed as 
higher-order analogues of the Peregrine breather.  The spectral data encoding 
the fundamental rogue waves includes a conjugate pair of singularities of 
fractional order (as opposed to a pair of poles of integer order in our 
situation).  The far-field behavior is completely different behavior for the 
two problems (compare, say, Figure \ref{fig-far-field-2d} above with Figure 2 
in \cite{BilmanLM:2018}), and it is only in an appropriate scaling near the 
origin that the functions $\Psi(X,T)$ emerge.  This situation in which the 
local behavior of more than one solution is described by the same Painlev\'e 
transcendent or other special function is a hallmark of integrable wave 
equations.  

Some of the main results concerning these functions in 
\cite{BilmanLM:2018} are as follows:
\begin{enumerate}
\item[(a$^\prime$)] Fundamental rogue-wave solutions of \eqref{nls2}
near the origin are, after appropriate scaling, well approximated by certain 
functions $\Psi(X,T;(1,\pm 1))$ that are solutions of \eqref{nls}.
\item[(b$^\prime$)] For fixed $T$, the functions $\Psi(X,T;(1,\pm 1))$ satisfy
the second member of Sakka's Painlev\'e-III hierarchy \cite{Sakka:2009} with 
certain parameters.  These functions satisfy $\Psi(0,0;(1,\pm 1))=\pm 4$, 
$\Psi(-X,T;(1,\pm 1))=\Psi(X,T;(1,\pm 1))$, and 
$\Psi(X,-T;(1,\pm 1))=\Psi(X,T;(1,\pm 1))^*$.
\item[(c$^\prime$)] For $T=0$, the functions $\Psi(X,0;(1,\pm 1))$ generate 
certain solutions of the Painlev\'e-III equation.  
\end{enumerate}

In this work we prove analogues of (a$^\prime$)--(c$^\prime$) for the 
higher-order pole soliton solutions of \eqref{nls}.  We start with a more 
general Darboux transformation than in \cite{BilmanLM:2018,BilmanM:2017}, 
allowing both general $\xi\in\mathbb{C}^+$ and general 
${\bf c}\in(\mathbb{C}^*)^2$.  The choice of $\xi$ can effectively be scaled 
out and does not introduce new phenomena, so we will fix $\xi=i$.  However, 
changing ${\bf c}$ amounts to changing the boundary 
conditions of $\Psi(X,T;{\bf c})$, leading to a new family of distinguished 
Painlev\'e solutions.  We emphasize that $\Psi(X,T;{\bf c})$ depends only on 
the ratio $c_1/c_2$, so in effect it is a one-(complex)-parameter family of 
solutions (see the discussion at the end of \S\ref{subsec-Darboux}).  We 
illustrate the convergence of $\frac{1}{n}\psi^{[2n]}(\frac{X}{n},0;{\bf c})$ 
to $\Psi(X,0;{\bf c})$ for ${\bf c}=(1,1)$ and ${\bf c}=(1,5)$ in Figures 
\ref{fig-origin} and \ref{fig-origin-c5}, respectively.  The plots of 
$\frac{1}{n}\psi^{[2n]}(\frac{X}{n},0;{\bf c})$ ($n=2,4,8$) were generated by 
solving the associated Riemann-Hilbert problem recast as a linear system (see 
Appendix \ref{sec-app-linear}).  The functions $\Psi(X,0;{\bf c})$ were 
computed using the methodology introduced in \cite{TrogdonO:2016} and 
\texttt{RHPackage} \cite{Olver:website} (see Appendix \ref{sec-app-Psi}).  
Note that the even $X$-symmetry enjoyed by $\Psi(X,0;(1,\pm 1))$ is broken 
for general ${\bf c}$.  
\begin{figure}[H]
\begin{center}
\includegraphics[height=1.6in]{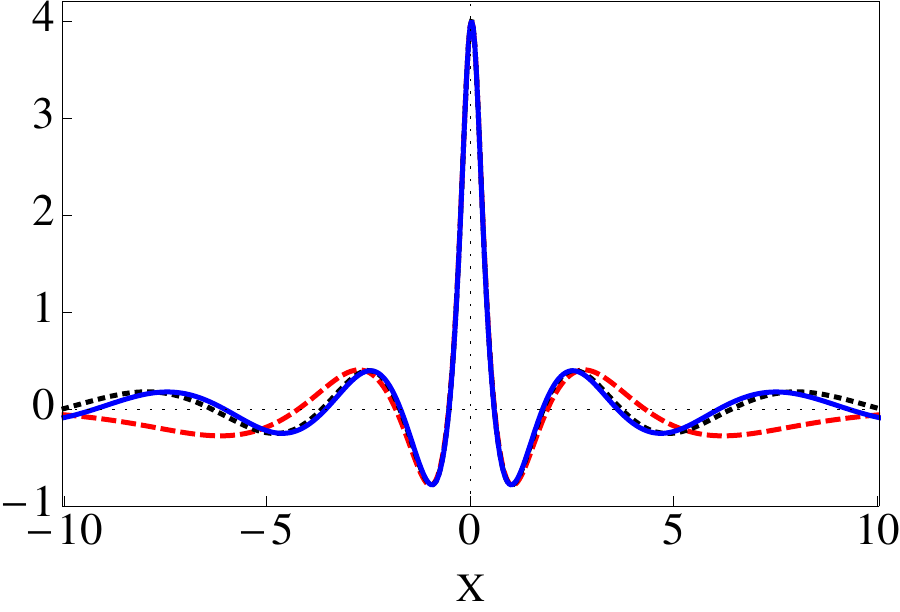}
\includegraphics[height=1.6in]{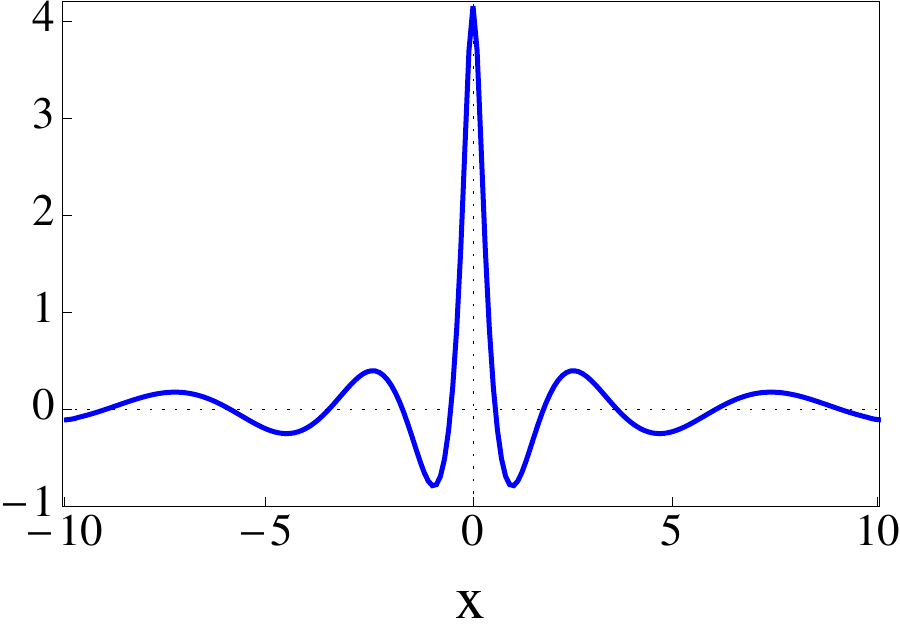}
\caption{\emph{Left:}  Scaled multiple-pole soliton solutions $\frac{1}{n}\psi^{[2n]}(\frac{X}{n},0;(1,1))$ of the nonlinear Schr\"odinger equation \eqref{nls} for $T=0$ and $-10\leq X\leq 10$ (i.e. $t=0$ and $-\frac{10}{n}\leq x\leq\frac{10}{n}$) for $n=2$ (red and dashed), $n=4$ (black and dotted), and $n=8$ (blue and solid) with ${\bf c}=(1,1)$.  \emph{Right:} The limiting function $\Psi(X,0;(1,1))$ with ${\bf c}=(1,1)$.}
\label{fig-origin}
\end{center}
\end{figure}
\begin{figure}[H]
\begin{center}
\includegraphics[height=1.4in]{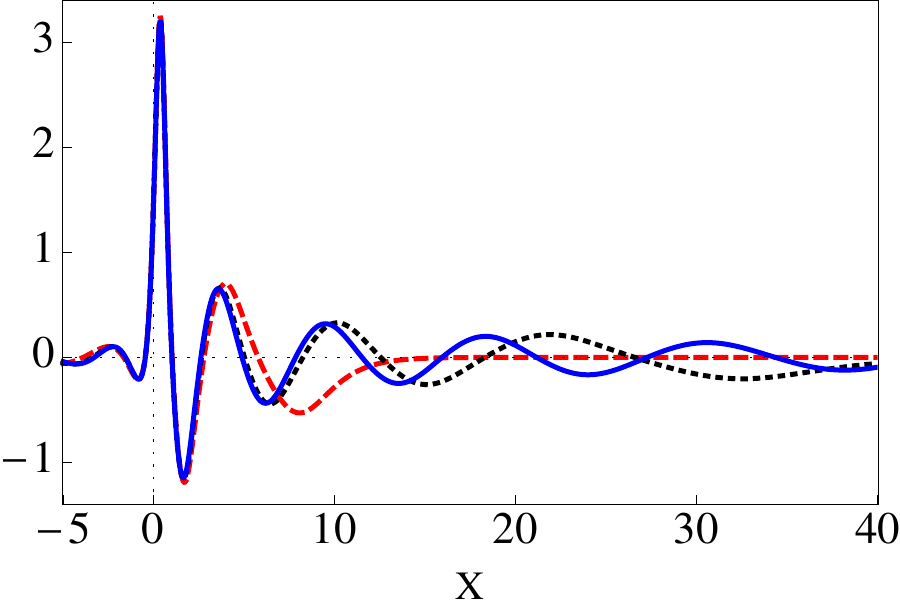}
\includegraphics[height=1.4in]{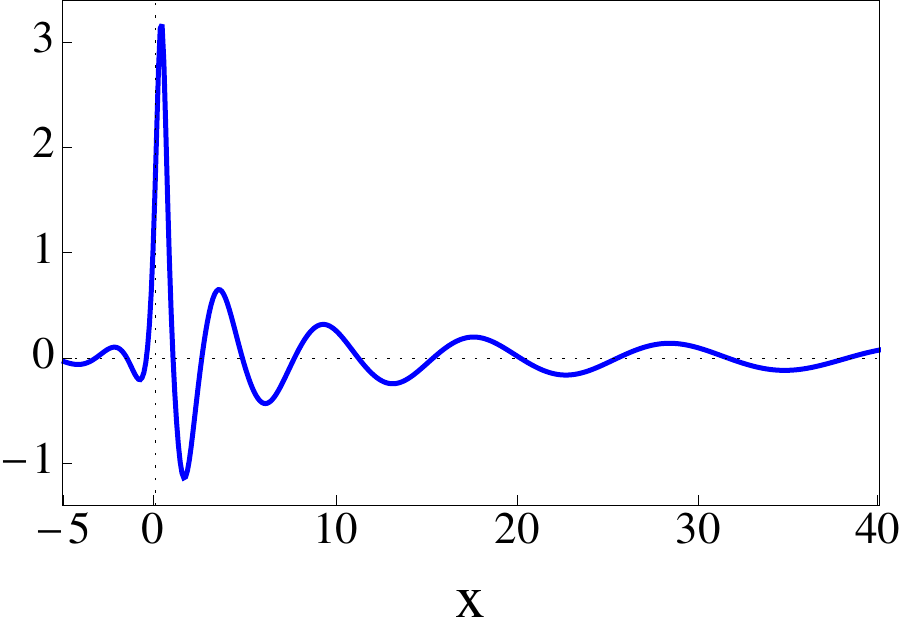}
\caption{\emph{Left:}  Scaled multiple-pole soliton solutions $\frac{1}{n}\psi^{[2n]}(\frac{X}{n},0;(1,5))$ of the nonlinear Schr\"odinger equation \eqref{nls} for $T=0$ and $-5\leq X\leq 40$ (i.e. $t=0$ and $-\frac{5}{n}\leq x\leq\frac{40}{n}$) for $n=2$ (red and dashed), $n=4$ (black and dotted), and $n=8$ (blue and solid) with ${\bf c}=(1,5)$.  \emph{Right:}  The limiting function $\Psi(X,0;(1,5))$ with ${\bf c}=(1,5)$.}
\label{fig-origin-c5}
\end{center}
\end{figure}
We now state our near-field results.  
\begin{theorem}
\label{thm-near-field}
Fix ${\bf c}\equiv(c_1,c_2)\in(\mathbb{C}^*)^2$ and set $\xi=i$.  Then the 
following hold. 
\begin{enumerate}
\item[(a)] There is an $n$-independent function $\Psi(X,T;{\bf c})$ 
such that, as $n\to\infty$,
\eq
\label{psi-Psi-approx}
\frac{1}{n}\psi^{[2n]}\left(\frac{X}{n},\frac{T}{n^2};{\bf c}\right) = \Psi(X,T;{\bf c}) + \mathcal{O}\left(\frac{1}{n}\right)
\endeq
uniformly for compact subsets of the $XT$-plane. 
The function $\Psi(X,T;{\bf c})$ is a solution of \eqref{nls} in the variables 
$X$ and $T$, i.e. 
\eq
\label{nls-XT}
i\Psi_T + \frac{1}{2}\Psi_{XX} + |\Psi|^2\Psi = 0, \quad X,T\in\mathbb{R}.
\endeq
\item[(b)] For fixed $T\in\mathbb{R}$, the function $\Psi(X,T;{\bf c})$ 
satisfies the fourth-order ordinary differential equation 
\eq
\label{fourth-order-ode}
\begin{split}
X\Psi\Psi_{XXX}+3\Psi\Psi_{XX}-X\Psi_X\Psi_{XX}-2(\Psi_X)^2+4\Psi^3\Psi^*+2X\Psi^2\Psi^*\Psi_X+2X\Psi^3\Psi_X^* & \\
+iT(\Psi\Psi_{XXXX}-\Psi_X\Psi_{XXX}+6\Psi^2\Psi_X\Psi_X^*+6\Psi^2\Psi^*\Psi_{XX}) & = 0,
\end{split}
\endeq
which is the second member of Sakka's Painlev\'e-III hierarchy 
\cite{Sakka:2009} with certain parameters.  Also, $\Psi(X,T;{\bf c})$ 
satisfies the initial conditions 
\eq
\label{Psi-values}
\Psi(0,0;{\bf c}) = 8\frac{c_1c_2^*}{|{\bf c}|^2}, \quad \Psi_X(0,0;{\bf c}) = 32\frac{c_1c_2^*}{|{\bf c}|^4}(c_2c_2^*-c_1c_1^*),
\endeq
as well as the symmetries
\eq
\label{Psi-symmetry-X}
\Psi(-X,T;{\bf c}^*\sigma_1) = \Psi(X,T;{\bf c}) 
\endeq
and 
\eq
\label{Psi-symmetry-T}
\Psi(X,-T;{\bf c})^* = \Psi(X,T;{\bf c}) \text{ if }{\bf c}\in\mathbb{R}^2.
\endeq
In particular, if ${\bf c}\in\mathbb{R}^2$ then $\Psi(X,0;{\bf c})$ is 
real-valued.
\item[(c)]  If ${\bf c}\in\mathbb{R}^2$, then 
\eq
\label{u-def}
u(s;{\bf c}):=\frac{2s^2\Psi(-\frac{1}{8}s^2,0;{\bf c})}{(s^2\Psi(-\frac{1}{8}s^2,0;{\bf c}))_s}
\endeq
satisfies the standard Painlev\'e-III equation 
\eq
\label{PIII}
u_{ss} = \frac{1}{u}(u_s)^2 - \frac{1}{s}u_s + \frac{4\Theta_0 u^2+4(1-\Theta_\infty)}{s} + 4u^3 - \frac{4}{u}
\endeq
with parameters $\Theta_\infty=\Theta_0=0$.
The odd function $u(s;{\bf c})$ satisfies
\eq
\label{u-values}
u_s(0;{\bf c}) = 1, \quad u_{sss}(0;{\bf c}) = \frac{3}{|{\bf c}|^2}(c_2^2 - c_1^2).
\endeq
\end{enumerate}
\end{theorem}

\subsection{Outline and notation}

In \S\ref{sec-Darboux} we define the Darboux transformation that generates 
$\psi^{[2n]}(x,t)$ from $\psi^{[2n-2]}(x,t)$ and use the robust 
inverse-scattering transform recently introduced by Bilman and Miller 
\cite{BilmanM:2017} to derive Riemann-Hilbert Problem \ref{rhp:psi-n} 
encoding $\psi^{[2n]}(x,t)$.  In \S\ref{sec-zero-reg} we analyze this 
Riemann-Hilbert problem in the far-field scaling \eqref{chi-tau} and prove 
Theorem \ref{thm-zero-reg} concerning the zero region.  In 
\S\ref{sec-near-field} we consider the near-field scaling 
\eqref{near-field-scaling} and prove Theorem \ref{thm-near-field} showing the 
local behavior near the origin is described by the solution $\Psi(X,T)$ of 
Sakka's Painlev\'e-III hierarchy.  In Appendix \ref{sec-app-linear} we show 
how the basic Riemann-Hilbert problem can be reformulated as a linear system, 
a fact used to plot $\psi^{[2n]}(x,t)$ in Figures 
\ref{fig-far-field-2d}--\ref{fig-origin-c5}.  In Appendix \ref{sec-app-Psi} 
we describe how \texttt{RHPackage} is used to compute $\Psi(X,T)$ for Figures 
\ref{fig-origin} and \ref{fig-origin-c5}.

\

\noindent
{\bf Notation.}
We define
\eq
\mathbb{I}:=\bpm 1 & 0 \\ 0 & 1 \epm, \quad \sigma_1 := \bpm 0 & 1 \\ 1 & 0 \epm, \quad \sigma_2:= \bpm 0 & -i \\ i & 0 \epm, \quad \sigma_3:=\bpm 1 & 0 \\ 0 & -1 \epm.
\endeq
With the exception of the identity matrix and the Pauli matrices, we denote 
$2\times2$ matrices by bold upper-case letters and 2-vectors by bold 
lower-case letters.  If ${\bf M}$ is a matrix, the $jk$-entry of ${\bf M}$ is 
denoted by $[{\bf M}]_{jk}$.  The complex conjugate of a number $a$ is 
denoted $a^*$, while the conjugate-transpose of a vector ${\bf v}$ is denoted 
${\bf v}^\dagger$.  By $\mathbb{C}^*$ we mean $\mathbb{C}\backslash\{0\}$, 
while by $\mathbb{C}^+$ we mean $\{z\in\mathbb{C}:\Im(z)>0\}$.  
The boundary of a domain $D\subset\mathbb{C}$ is denoted 
$\partial D$.  When we write  $f(x;s)$ we mean that $f$ is a function of $x$ 
with parametric dependence on $s$, and we may suppress the explicit 
dependence on parameters for brevity.

\section{The Darboux transformation and initial Riemann-Hilbert problem}
\label{sec-Darboux}

\subsection{Summary of the robust inverse scattering transform}
\label{subsec-robust}
We begin with a brief overview of the robust inverse-scattering transform 
introduced recently in \cite{BilmanM:2017} and how it differs from the standard 
inverse-scattering transform.  The basis for any inverse-scattering trasform for 
\eqref{nls} is the Lax pair \cite{ZakharovS:1972}
\eq
\label{lax-pair-xt}
\begin{split}
\frac{\partial{\bf u}}{\partial x}(\lambda;x,t) & = \bbm -i\lambda & \psi(x,t) \\ -\psi(x,t) & i\lambda \ebm {\bf u}(\lambda;x,t), \\ 
\frac{\partial{\bf u}}{\partial t}(\lambda;x,t) & = \bbm -i\lambda^2+\frac{i}{2}|\psi(x,t)|^2 & \lambda\psi(x,t) + \frac{i}{2}\psi_x(x,t) \\ -\lambda\psi(x,t)^* + \frac{i}{2}\psi_x(x,t)^* & i\lambda^2 - \frac{i}{2}|\psi(x,t)|^2 \ebm {\bf u}(\lambda;x,t).
\end{split}
\endeq
The standard inverse-scattering transform \cite{BealsC:1984} begins by defining the 
Jost matrices ${\bf J}^\pm(\lambda;x,t)$ to be the unique fundamental matrix of 
simultaneous solutions to \eqref{lax-pair-xt} satisfying the boundary conditions 
${\bf J}^\pm(\lambda;x,t)e^{i\lambda x\sigma_3} = \mathbb{I}+o(1)$, $x\to\pm\infty$.
The Jost functions also satisfy the scattering relation 
\eq
{\bf J}^+(\lambda;x,t)={\bf J}^-(\lambda;x,t)\bbm a(\lambda^*)^* & b(\lambda^*;t)^* \\ -b(\lambda;t) & a(\lambda) \ebm.
\endeq
If we denote the first and second columns of ${\bf J}^\pm$ by ${\bf j}^{\pm,1}$ and 
${\bf j}^{\pm,2}$, respectively, 
then we can define the Beals-Coifman simultaneous solution to \eqref{lax-pair-xt} 
as 
\eq
{\bf U}^\text{BC}(\lambda;x,t) := \begin{cases} \left[\frac{1}{a(\lambda)}{\bf j}^{-,1}(\lambda;x,t)e^{-i\lambda^2 t}, \hspace{.2in} {\bf j}^{+,2}(\lambda;x,t)e^{i\lambda^2 t}\right], & \lambda\in\mathbb{C}^+, \vspace{.05in}\\  
\left[{\bf j}^{+,1}(\lambda;x,t)e^{-i\lambda^2 t}, \hspace{.2in} \frac{1}{a(\lambda^*)^*}{\bf j}^{-,2}(\lambda;x,t)e^{i\lambda^2 t}\right], & \lambda\in\mathbb{C}^-. 
\end{cases}
\endeq
Then the function 
\eq
{\bf M}^\text{BC}(\lambda;x,t) := {\bf U}^\text{BC}(\lambda;x,t)e^{i(\lambda x+\lambda^2 t)\sigma_3}
\endeq
satisfies the normalization condition 
\eq
\label{MBC-normalization}
\lim_{\lambda\to\infty}{\bf M}^\text{BC}(\lambda;x,t) = \mathbb{I},
\endeq
the Schwarz-symmetry condition 
\eq
\label{Schwarz-symmetry}
{\bf M}^\text{BC}(\lambda;x,t) = \sigma_2{\bf M}^\text{BC}(\lambda^*;x,t)^*\sigma_2, \quad \lambda\in\mathbb{C}\backslash\mathbb{R},
\endeq
and the jump condition
\eq
\label{M-real-jump}
{\bf M}_+^\text{BC}(\lambda;x,t) = {\bf M}_-^\text{BC}(\lambda;x,t) \bbm 1+|R(\lambda)|^2 & R(\lambda)^*e^{-2i(\lambda x+\lambda^2 t)} \\ R(\lambda)e^{2i(\lambda x+\lambda^2 t)} & 1 \ebm, \quad \lambda\in\mathbb{R}, 
\endeq
where $R(\lambda):= b(\lambda;t)/a(\lambda)$.  
These properties allow ${\bf M}^\text{BC}(\lambda;x,t)$ to be obtained as the 
solution of a Riemann-Hilbert problem.  

The function ${\bf M}^\text{BC}(\lambda;x,t)$ has nice properties as 
$\lambda\to\infty$ (i.e. \eqref{MBC-normalization}).  However, for general 
$\lambda$ ${\bf M}^\text{BC}(\lambda;x,t)$ is only sectionally 
\emph{meromorphic}, with poles arising from zeros of $a(\lambda)$ 
corresponding to solitons.  In general, these poles can be handled in the 
Riemann-Hilbert formalism either by solving the problem exactly (as we do in 
Appendix \ref{sec-app-linear}) or by interpolation (see, for instance, 
\cite{KamvissisMM:2003}).  There are technical 
issues with unique solvability of the Riemann-Hilbert problem with spectral 
singularities, i.e.~points in the continuous spectrum for which 
${\bf M}^\text{BC}(\lambda;x,t)$ fails to have a boundary value 
\cite[\S1.1.2]{BilmanM:2017}.  These spectral singularies, which arise in 
particular in the case of Peregrine breathers, can be handled by a limiting 
procedure.  The robust inverse-scattering transform bypasses the limiting 
procedure and leads directly to a sectionally analytic Riemann-Hilbert 
problem, even for solutions whose scattering data under the standard 
inverse-scattering transform consist of spectral singularities of high order 
(i.e higher-order Peregrine breathers) \cite{BilmanLM:2018}.  While we are 
not concerned with spectral singularities in the current work, we take 
advantage of the robust inverse-scattering transform's ability to handle 
higher-order poles.

The key observation of the robust inverse-scattering transform is that 
different fundamental solutions of \eqref{lax-pair-xt} have desirable 
properties in different sections of the $\lambda$-plane.  The Beals-Coifman 
solution ${\bf U}^\text{BC}(\lambda;x,t)$ is well-behaved for $|\lambda|$ 
sufficiently large.  On the other hand, there are other solutions that are 
bounded in the regions where ${\bf U}^\text{BC}(\lambda;x,t)$ has poles.  
The following key proposition is proved in 
\cite[Proposition 2.1]{BilmanM:2017} for \eqref{nls2};  nevertheless the 
proof goes through verbatim for \eqref{nls}.  
\begin{proposition}
Suppose $\psi(x,t)$ is a bounded classical solution of \eqref{nls} defined 
for $(x,t)$ in a simply connected domain $\Omega\subset\mathbb{R}^2$ 
containing $(0,0)$.  Then, for each $\lambda\in\mathbb{C}$, there exists a 
unique simultaneous fundamental solution matrix 
${\bf U}^{\rm in}(\lambda;x,t)$, $(x,t)\in\Omega$, of the Lax pair equations 
\eqref{lax-pair-xt} together with the initial condition 
${\bf U}^{\rm in}(\lambda;0,0)=\mathbb{I}$.  Furthermore, 
${\bf U}^{\rm in}(\lambda;x,t)$ is an entire function of $\lambda$ for each 
$(x,t)\in\Omega$, $\det{\bf U}^{\rm in}(\lambda;x,t)\equiv 1$, and 
${\bf U}^{\rm in}(\lambda;x,t)=\sigma_2{\bf U}^{\rm in}(\lambda^*;x,t)^*\sigma_2$.
\label{p:U-in}
\end{proposition}
Now define $D_0\subset\mathbb{C}$ to be an open disk centered at the origin 
of sufficiently large radius to contain all the singularities of 
${\bf U}^\text{BC}(\lambda;x,t)$.  Set 
\eq
\label{U-def}
{\bf U}(\lambda;x,t) := \begin{cases} {\bf U}^{\rm in}(\lambda;x,t), & \lambda\in D_0, \\ {\bf U}^\text{BC}(\lambda;x,t), & \text{otherwise} \end{cases}
\endeq
and define $\Sigma_L:=(-\infty,-r)$, $\Sigma_R:=(r,\infty)$, 
$\Sigma_+:=\partial D_0\cap\mathbb{C}^+$, $\Sigma_-:=\partial D_0\cap\mathbb{C}^-$ 
($\Sigma_L$ and $\Sigma_R$ are oriented left-to-right while $\Sigma_+$ and $\Sigma_-$ 
are oriented clockwise).  Then the function 
\eq
{\bf M}(\lambda;x,t) := {\bf U}(\lambda;x,t)e^{i(\lambda x+\lambda^2t)\sigma_3}, \quad \lambda\in\mathbb{C}\backslash\mathbb{R}, \quad (x,t)\in\mathbb{R}^2
\endeq
is analytic for $\lambda\notin\Sigma_L\cup\Sigma_R\cup\Sigma_+\cup\Sigma_-$, 
satisfies the jump \eqref{M-real-jump} on $\Sigma_L\cup\Sigma_R$, and has the jump 
\eq
{\bf M}_+(\lambda;x,t) = \begin{cases} \vspace{.02in} {\bf M}_-(\lambda;x,t) \left[\frac{1}{a(\lambda)}{\bf j}^{-,1}(\lambda;0,0), \hspace{.2in} {\bf j}^{+,2}(\lambda;0,0) \right], & \lambda\in\Sigma^+, \\ {\bf M}_-(\lambda;x,t) \left[{\bf j}^{+,1}(\lambda;0,0), \hspace{.2in} \frac{1}{a(\lambda^*)^*}{\bf j}^{-,2}(\lambda;0,0) \right], & \lambda\in\Sigma^- \end{cases}
\endeq
on the remaining contours.  Thus, in addition to having identity asymptotics at 
infinity (see \eqref{MBC-normalization}), ${\bf M}(\lambda;x,t)$ satisfies a jump 
condition in a form amenable to analysis via the nonlinear steepest-descent method 
of Deift and Zhou \cite{DeiftZ:1993}.  Once ${\bf M}(\lambda;x,t)$ is 
known, the solution to \eqref{nls} can be found by 
\eq
\psi(x,t) = 2i\lim_{\lambda\to\infty}\lambda[{\bf M}(\lambda;x,t)]_{12}.
\endeq

\subsection{Definition of the Darboux transformation}
\label{subsec-Darboux}
We now define the specific Darboux transformations we will use and show how 
to formulate them in terms of the robust inverse-scattering transform.  The 
basic idea of a Darboux transformation is to take a solution of \eqref{nls} 
and find the solution having the same Beals-Coifman scattering data with the 
exception of one or more additional poles (see \cite{MatveevS:1991} for further 
background).  

We begin by defining certain Darboux transformations depending on an 
eigenvalue $\xi=\alpha+i\beta\in\mathbb{C}^+$ and a row vector of connection 
coefficients ${\bf d}:=(d_1,d_2)\in\mathbb{C}^2$.  The Darboux 
transformations we study will be derived from these by taking a certain 
limit of ${\bf d}$.  
Given the matrix ${\bf U}(\lambda;x,t)$ defined in \eqref{U-def} and 
associated to a solution $\psi(x,t)$ of \eqref{nls}, we introduce 
\eq
\label{basic-Darboux}
\dot{\bf U}(\lambda;x,t) := \left(\mathbb{I} + \frac{{\bf R}(x,t)}{\lambda-\xi}\right){\bf U}(\lambda;x,t)
\endeq
for a to-be-determined matrix ${\bf R}(x,t)$.  The transformation 
\eqref{basic-Darboux} is constructed so that $\dot{\bf U}(\lambda;x,t)$ has 
the same jump conditions and normalization as $\lambda\to\infty$ as 
${\bf U}(\lambda;x,t)$.  If we further assume that ${\bf R}(x,t)^2={\bf 0}$, 
then we also have $\det\dot{\bf U}(\lambda;x,t)=1$.  The remaining freedom in 
determining ${\bf R}(x,t)$ can be used to ensure 
\eq
\mathop{\text{Res}}_{\lambda=\xi}\dot{\bf U}(\lambda;x,t) = \lim_{\lambda\to\xi}\dot{\bf U}(\lambda;x,t)\bbm d_1 \\ d_2 \ebm \bbm id_2 & -id_1 \ebm = \lim_{\lambda\to\xi}\dot{\bf U}(\lambda;x,t)\bbm id_1d_2 & -id_1^2 \\ id_2^2 & -id_1d_2 \ebm.
\endeq
This condition completely specifies ${\bf R}(x,t)$ (see 
\cite[\S 3.1]{BilmanM:2017} for complete details) as 
\eq
{\bf R}(x,t) = \frac{{\bf U}(\xi;x,t){\bf d}^\mathsf{T}{\bf d}{\bf U}(\xi;x,t)^\mathsf{T}\sigma_2}{1-{\bf d}{\bf U}(\xi;x,t)^\mathsf{T}\sigma_2{\bf U}'(\xi;x,t){\bf d}^\mathsf{T}}.
\endeq

At this point we would like to define 
$\dot{\bf M}(\lambda;x,t):=\dot{\bf U}(\lambda;x,t)e^{i(\lambda x+\lambda^2t)\sigma_3}$ 
and set up the associated Riemann-Hilbert problem.  The problem is that 
$\dot{\bf M}(\lambda;x,t)$ is not Schwarz-symmetric (i.e. does not satisfy an 
analogue of \eqref{Schwarz-symmetry}), and thus cannot generate a solution of 
\eqref{nls}.  This can be remedied by first performing a 
Darboux transformation ${\bf U}(\lambda;x,t)\to\dot{\bf U}(\lambda;x,t)$ with 
data $(\xi,{\bf d})$, and then performing a second Darboux transformation 
$\dot{\bf U}(\lambda;x,t)\to\ddot{\bf U}(\lambda;x,t)$ with data 
$(\xi^*,({\bf d}\sigma_2)^*)$.  We now write the composition of these two 
Darboux transformations explicitly (the interested reader can find full 
details of the straightforward calculation in \cite[\S3.2]{BilmanM:2017}).  
Define 
\eq
\begin{split}
{\bf s}_\text{f}(x,t):={\bf U}(\xi;x,t){\bf d}^\mathsf{T}, \quad N_\text{f}(x,t):={\bf s}_\text{f}(x,t)^\dagger{\bf s}_\text{f}(x,t), \\
w_\text{f}(x,t):={\bf d}{\bf U}(\xi;x,t)^\mathsf{T}\sigma_2{\bf U}^\prime(\xi;x,t){\bf d}^\mathsf{T}\phantom{nnnnn}
\end{split}
\endeq
(here the subscript f stands for ``finite'') and use these to define
\eq
\begin{split}
{\bf Y}_\text{f}(x,t):= & \frac{4\beta^2(1-w_\text{f}(x,t)^*)}{4\beta^2|1-w_\text{f}(x,t)|^2 + N_\text{f}(x,t)^2}{\bf s}_\text{f}(x,t){\bf s}_\text{f}(x,t)^\mathsf{T}\sigma_2 \\ 
& + \frac{2i\beta N_\text{f}(x,t)}{4\beta^2|1-w_\text{f}(x,t)|^2+N_\text{f}(x,t)^2}\sigma_2{\bf s}_\text{f}(x,t)^*{\bf s}_\text{f}(x,t)^\mathsf{T}\sigma_2, \\
{\bf Z}_\text{f}(x,t) := & \sigma_2{\bf Y}_\text{f}(x,t)^*\sigma_2.
\end{split}
\endeq
Then 
\eq
\label{U-double-dot}
\ddot{\bf U}(\lambda;x,t) = {\bf G}_\text{f}(\lambda;x,t){\bf U}(\lambda;x,t),
\endeq
where
\eq
{\bf G}_\text{f}(\lambda;x,t) := \mathbb{I} + \frac{{\bf Y}_\text{f}(x,t)}{\lambda-\xi} + \frac{{\bf Z}_\text{f}(x,t)}{\lambda-\xi^*}.
\endeq

If we apply the general Darboux transformation defined by 
\eqref{U-double-dot} to the trivial background solution $\psi(x,t)\equiv 0$, 
the resulting position of the generated solution can be shifted by replacing 
$(d_1,d_2)$ with $(\epsilon^{-1}d_1,\epsilon^{-1}d_2)$ for some fixed 
constant $\epsilon\in\mathbb{C}^*$.  We 
choose to fix two complex constants $c_1$ and $c_2$ and study the Darboux 
transformation with connection data $(\epsilon^{-1}c_1,\epsilon^{-1}c_2)$ in 
the limit $\epsilon\to 0$.  This will have in particular the effect of 
ensuring $\psi(x,t;(1,\pm 1))$ achieves its maximum value at the origin 
$(x,t)=(0,0)$.  The formulas used to construct the Darboux transformation all 
have well-defined limits as $\epsilon\to 0$ that are given explicitly in 
\eqref{s-N-w}--\eqref{Yn-Zn} below.  We note that those formulas are 
unchanged if both $c_1$ and $c_2$ are multiplied by the same nonzero complex 
number.  Because of this, we can think of ${\bf c}$ as an element of the 
complex projective space $\mathbb{CP}^1$.  This means that, although we 
write both $c_1$ and $c_2$, there is actually only \emph{one} complex degree 
of freedom in the Darboux transformation.  Furthermore, if either $c_1=0$ or 
$c_2=0$ then the limiting Darboux transformation defined by 
\eqref{s-N-w}--\eqref{Yn-Zn} is trivial (i.e. takes the input solution 
$\psi(x,t)$ to itself).  Thus we restrict ourselves to 
${\bf c}\in(\mathbb{C}^*)^2$.

\subsection{Iteration of the Darboux transformation}
\label{subsec-iteration}
If we apply the Darboux transformation with data $\{\xi,{\bf c}\}$ to the 
trivial solution 
\eq
\label{psi0}
\psi^{[0]}(x,t)\equiv 0
\endeq
then we obtain a second-order pole soliton $\psi^{[2]}(x,t;{\bf c})$.  In the 
same way, applying the Darboux transformation $n$ times in succession (using 
the robust inverse-scattering transformation each time to sweep the spectral 
poles to $\partial D_0$) will generate a $2n^\text{th}$-order pole soliton 
$\psi^{[2n]}(x,t;{\bf c})$.  In principal there is no need to fix the data 
$\{\xi,{\bf c}\}$ between iterations, although we will do so in order to 
obtain a well-defined limit as $n\to\infty$.  We now explain the construction 
in detail.  

Fix $\xi\in\mathbb{C}^+$ and ${\bf c} = (c_1,c_2)\in(\mathbb{C}^*)^2$.  We 
begin with the background eigenvector matrix
\eq
\label{background-U}
{\bf U}^{[0]}(\lambda;x,t) := e^{-i(\lambda x+\lambda^2 t)\sigma_3}
\endeq
corresponding to the trivial solution \eqref{psi0}.  
Let $D_0$ be a circular disc centered at the origin that is large enough to 
contain $\xi$.  Suppose that ${\bf U}^{[n]}(\lambda;x,t)$ is known and is 
analytic for $\lambda\notin\partial D_0$.  Set 
\eq
\label{s-N-w}
\begin{split}
{\bf s}^{[n]}(x,t):={\bf U}^{[n]}(\xi;x,t){\bf c}^\mathsf{T}, \quad N^{[n]}(x,t):={\bf s}^{[n]}(x,t)^\dagger{\bf s}^{[n]}(x,t), \\
w^{[n]}(x,t):={\bf c}{\bf U}^{[n]}(\xi;x,t)^\mathsf{T}\sigma_2{\bf U}^{[n]\prime}(\xi;x,t){\bf c}^\mathsf{T}.\phantom{nnnnnn}
\end{split}
\endeq
Then define 
\eq
\label{Gn-ito-Y-and-Z}
{\bf G}^{[n]}(\lambda;x,t) := \mathbb{I} + \frac{{\bf Y}^{[n]}(x,t)}{\lambda-\xi} + \frac{{\bf Z}^{[n]}(x,t)}{\lambda-\xi^*},
\endeq
wherein
\eq
\label{Yn-Zn}
\begin{split}
{\bf Y}^{[n]}(x,t):= & \frac{-4\beta^2 w^{[n]}(x,t)^*}{4\beta^2|w^{[n]}(x,t)|^2 + N^{[n]}(x,t)^2}{\bf s}^{[n]}(x,t){\bf s}^{[n]}(x,t)^\mathsf{T}\sigma_2 \\ 
& + \frac{2i\beta N^{[n]}(x,t)}{4\beta^2|w^{[n]}(x,t)|^2+N^{[n]}(x,t)^2}\sigma_2{\bf s}^{[n]}(x,t)^*{\bf s}^{[n]}(x,t)^\mathsf{T}\sigma_2, \\
{\bf Z}^{[n]}(x,t) := & \sigma_2{\bf Y}^{[n]}(x,t)^*\sigma_2.
\end{split}
\endeq
Define
\eq
\ddot{\bf U}^{[n+1]}(\lambda;x,t) := {\bf G}^{[n]}(\lambda;x,t){\bf U}^{[n]}(\lambda;x,t).
\endeq
Now $\ddot{\bf U}^{[n+1]}(\lambda;x,t)$ has simple poles at $\xi$ and $\xi^*$.  We 
apply the idea of the robust inverse-scattering transform and sweep the poles to 
$\partial D_0$ by defining 
$\ddot{\bf U}^{[n,\text{in}]}(\lambda;x,t) := \ddot{\bf U}^{[n]}(\lambda;x,t){\bf G}^{[n]}(\lambda;0,0)^{-1}$
for $\lambda\in D_0$.  Then the matrix 
\eq
\label{Unp1-ito-Gn-Un}
{\bf U}^{[n+1]}(\lambda;x,t) := \begin{cases} {\bf G}^{[n]}(\lambda;x,t){\bf U}^{[n]}(\lambda;x,t), & \lambda\notin D_0, \\ {\bf G}^{[n]}(\lambda;x,t){\bf U}^{[n]}(\lambda;x,t){\bf G}^{[n]}(\lambda;0,0)^{-1}, & \lambda\in D_0 \end{cases}
\endeq
is analytic for $\lambda\notin\partial D_0$.  
In terms of the matrix ${\bf Y}^{[n]}$, the solution $\psi^{[2n+2]}(x,t)$ to 
\eqref{nls} is obtained from the solution $\psi^{[2n]}(x,t)$ by 
\eq
\label{psi-n}
\psi^{[2n+2]}(x,t) = \psi^{[2n]}(x,t) + 2i([{\bf Y}^{[n]}(x,t)]_{12}-[{\bf Y}^{[n]}(x,t)^*]_{21}).
\endeq

For reference we perform the first Darboux transformation 
explicitly.  Writing 
${\bf s}^{[0]}=\left(s_1^{[0]},s_2^{[0]}\right)^\mathsf{T}$, from 
\eqref{psi-n}, \eqref{psi0}, and \eqref{Yn-Zn} we see
\eq
\label{first-Darboux}
\psi^{[2]} = \frac{-8\beta^2w^{[0]*}(s_1^{[0]})^2 + 8\beta^2w^{[0]}(s_2^{[0]*})^2 + 8\beta N^{[0]}s_1^{[0]}s_2^{[0]*}}{4\beta^2|w^{[0]}|^2+(N^{[0]})^2}.
\endeq
We also have (writing $\xi=\alpha+i\beta$)
\eq
\begin{split}
s_1^{[0]}(x,t) = c_1 e^{-i(\xi x+\xi^2 t)},& \quad s_2^{[0]}(x,t) = c_2 e^{i(\xi x+\xi^2 t)}, \\
w^{[0]}(x,t) = 2c_1c_2(x+2\xi t), \quad N^{[0]}&(x,t) = |c_1|^2e^{2\beta x + 4\alpha\beta t} + |c_2|^2e^{-2\beta x - 4\alpha\beta t}.
\end{split}
\endeq
Combining the previous two equations gives an explicit formula for 
$\psi^{[2]}(x,t)$ for any choice of $c_1,c_2\in\mathbb{C}^*$ and 
$\xi\in\mathbb{C}^+$.

\begin{rmk}
\label{rmk-pole-order}
Looking at Figures \ref{fig-far-field-1d} and \ref{fig-far-field-2d-c5}, it 
appears that $\psi^{[2n]}(x,t)$ 
is a coalescence of $2n$ single-pole solitons (and this is indeed the case).  
Yet the observant reader may have noticed that $\psi^{[2n]}(x,t)$ is 
generated from the trivial background by only $n$ applications of the Darboux 
transformation \eqref{Unp1-ito-Gn-Un}, each of which only involves a single 
pole.  How do the $n$ extra poles arise?  To understand this, note that the 
right (i.e normalized as $x\to+\infty$) Beals-Coifman matrix 
${\bf U}^\text{BC}(\lambda;x,t)$ associated to $\psi^{[2n]}(x,t)$ with poles 
of order $2n$ at $\lambda=\xi$ and $\lambda=\xi^*$ has the asymptotic 
behavior 
\eq
\begin{cases}
\displaystyle\lim_{x\to+\infty} {\bf U}^\text{BC}(\lambda;x,t)e^{i(\lambda x+\lambda^2 t)\sigma_3} = \mathbb{I}, \\ 
{\bf U}^\text{BC}(\lambda;x,t)e^{i(\lambda x+\lambda^2 t)\sigma_3} \text{ is bounded as }x\to-\infty, 
\end{cases}
\quad \lambda\notin D_0, \, \Im(\lambda)>0.
\endeq
However, the matrix ${\bf U}^{[2n]}(\lambda;x,t)$ has different asymptotics 
as $x\to +\infty$ for $\Im(\lambda)>0$, and it is necessary to renormalize to obtain the associated 
Beals-Coifman matrix, which introduces the additional pole at each iteration.
\end{rmk}

\begin{rmk}
If one wanted to study the sequence of odd-order pole solitons 
$\{\psi^{[2n+1]}(x,t)\}$, then one could start with the standard single-pole 
soliton $\psi^{[1]}(x,t)$ and apply the Darboux transformation 
\eqref{Unp1-ito-Gn-Un} $n$ times.  We anticipate that the large-$n$ behavior 
of the odd sequence is the same as that of the even sequence, and so we 
restrict our attention to $\psi^{[2n]}(x,t)$.  
\end{rmk}

\subsection{The Riemann-Hilbert problem}
Given ${\bf U}^{[n]}(\lambda;x,t)$, we define 
\eq
\label{Mn-ito-Un}
{\bf M}^{[n]}(\lambda;x,t):={\bf U}^{[n]}(\lambda;x,t)e^{i(\lambda x+\lambda^2 t)\sigma_3}.
\endeq
We now pose the Riemann-Hilbert problem satisfied by ${\bf M}^{[n]}(\lambda;x,t)$.  
Orient $\partial D_0$ clockwise.  From \eqref{Unp1-ito-Gn-Un} and \eqref{Mn-ito-Un}, 
we see the jump for ${\bf M}^{[n]}(\lambda;x,t)$ is
\eq
\label{Mn-jump}
{\bf M}_+^{[n]}(\lambda;x,t) = {\bf M}_-^{[n]}(\lambda;x,t){\bf V}_{\bf M}^{[n]}(\lambda;x,t), \quad \lambda\in\partial D_0,
\endeq
where
\eq
{\bf V}_{\bf M}^{[n]}(\lambda;x,t) = e^{-i(\lambda x+\lambda^2 t)\sigma_3}{\bf G}^{[n-1]}(\lambda;0,0)\cdots{\bf G}^{[1]}(\lambda;0,0){\bf G}^{[0]}(\lambda;0,0)e^{i(\lambda x+\lambda^2 t)\sigma_3}.
\endeq
Explicit evaluation of ${\bf Y}^{[n]}(0,0)$ shows it is actually independent of 
$n$, and so ${\bf G}^{[n]}(\lambda;0,0)={\bf G}^{[0]}(\lambda;0,0)$ for all 
$n$. The reason for this is the point $(x,t)=(0,0)$ is the normalization point 
of $\mathbf{U}^\mathrm{in}(\lambda;x,t)$ in Proposition~\ref{p:U-in}, and hence 
the values of the quantities \eqref{s-N-w} coincide at $(0,0)$ for each $n$.  
Thus 
\eq
{\bf V}_{\bf M}^{[n]}(\lambda;x,t):= e^{-i(\lambda x+\lambda^2 t)\sigma_3}{\bf G}^{[0]}(\lambda;0,0)^n e^{i(\lambda x+\lambda^2 t)\sigma_3}.
\endeq
From \eqref{Gn-ito-Y-and-Z} we have 
\eq
{\bf G}^{[0]}(\lambda;0,0) = \mathbb{I} + \frac{{\bf Y}^{[0]}(0,0)}{\lambda-\xi} + \frac{{\bf Z}^{[0]}(0,0)}{\lambda-\xi^*}.
\endeq
By direct calculation we have 
\eq
{\bf Y}^{[0]}(0,0) = \frac{2i\beta}{|{\bf c}|^2}\sigma_2({\bf c}^*)^\mathsf{T}{\bf c}\sigma_2 = \frac{2i\beta}{|{\bf c}|^2}\bpm c_2c_2^* & -c_1c_2^* \\ -c_1^*c_2 & c_1c_1^* \epm
\endeq
and
\eq
{\bf Z}^{[0]}(0,0) = \sigma_2{\bf Y}^{[0]}(0,0)^*\sigma_2 = \frac{-2i\beta}{|{\bf c}|^2}\bpm c_1c_1^* & c_1c_2^* \\ c_1^*c_2 & c_2c_2^* \epm.
\endeq
The eigenvalues of ${\bf G}^{[0]}(\lambda;0,0)$ are 
$\frac{\lambda-\xi}{\lambda-\xi^*}$ and $\frac{\lambda-\xi^*}{\lambda-\xi}$.  
Recall the eigenvector matrix ${\bf S}$ defined by \eqref{S-def}.  We define a 
second eigenvector matrix by 
\eq
\label{Stilde-def}
\widetilde{\bf S}:=\frac{1}{|{\bf c}|}\bpm c_2^* & c_1 \\ -c_1^* & c_2 \epm.
\endeq
Then we have the following two useful representations of the jump matrix for 
${\bf M}^{[n]}(\lambda;x,t)$:
\eq
\label{V-representations}
\begin{split}
{\bf V}_{\bf M}^{[n]}(\lambda;x,t) & = e^{-i(\lambda x+\lambda^2 t)\sigma_3}{\bf S}\bpm \left(\frac{\lambda-\xi}{\lambda-\xi^*}\right)^n & 0 \\ 0 & \left(\frac{\lambda-\xi^*}{\lambda-\xi}\right)^n\epm{\bf S}^{-1}e^{i(\lambda x+\lambda^2 t)\sigma_3} \\
& = e^{-i(\lambda x+\lambda^2 t)\sigma_3}\widetilde{\bf S}\bpm \left(\frac{\lambda-\xi^*}{\lambda-\xi}\right)^n & 0 \\ 0 & \left(\frac{\lambda-\xi}{\lambda-\xi^*}\right)^n\epm\widetilde{\bf S}^{-1}e^{i(\lambda x+\lambda^2 t)\sigma_3}.
\end{split}
\endeq
We therefore have the basic Riemann-Hilbert Problem \ref{rhp:psi-n} for 
${\bf M}^{[n]}(\lambda;x,t)$.  The $2n^\text{th}$-order pole soliton 
$\psi^{[2n]}(x,t;{\bf c})$ is obtained from ${\bf M}^{[n]}(\lambda;x,t)$ via 
\eqref{psi-from-M}.

\section{Analysis in the zero region}
\label{sec-zero-reg}
We now prove Theorem \ref{thm-zero-reg}.  
Starting with ${\bf M}^{[n]}(\lambda;x,t)$, we perform a series 
of invertible transformations to analyze Riemann-Hilbert Problem \ref{rhp:psi-n}
asymptotically.  Recall $\chi:=x/n$ and $\tau:=t/n$ as introduced in 
\eqref{chi-tau}.
If $\chi>0$, define
\eq
{\bf N}^{[n]}(\lambda;\chi,\tau) := \begin{cases} {\bf M}^{[n]}(\lambda;n\chi,n\tau)e^{-in(\lambda \chi+\lambda^2\tau)}{\bf S}e^{in(\lambda \chi+\lambda^2\tau)}, & \lambda\in D_0, \\ {\bf M}^{[n]}(\lambda;n\chi,n\tau)\left(\frac{\lambda-\xi^*}{\lambda-\xi}\right)^{n\sigma_3}, & \lambda\notin D_0 \end{cases} \quad (\chi>0).
\endeq
If $\chi<0$, define
\eq
{\bf N}^{[n]}(\lambda;\chi,\tau) := \begin{cases} {\bf M}^{[n]}(\lambda;n\chi,n\tau)e^{-in(\lambda \chi+\lambda^2\tau)}\widetilde{\bf S}e^{in(\lambda \chi+\lambda^2\tau)}, & \lambda\in D_0, \\ {\bf M}^{[n]}(\lambda;n\chi,n\tau)\left(\frac{\lambda-\xi}{\lambda-\xi^*}\right)^{n\sigma_3}, & \lambda\notin D_0 \end{cases} \quad (\chi<0).
\endeq
The normalization for ${\bf N}^{[n]}(\lambda;\chi,\tau)$ as $\lambda\to\infty$ is 
unchanged from that of ${\bf M}^{[n]}(\lambda;\chi,\tau)$.
Introducing the phase functions $\varphi(\lambda;\chi,\tau)$ by \eqref{phi} and 
$\widetilde\varphi(\lambda;\chi,\tau)$ by 
\eq
\label{phitilde}
\widetilde\varphi(\lambda;\chi,\tau) :=i(\lambda \chi+\lambda^2 \tau) + \log\left(\frac{\lambda-\xi}{\lambda-\xi^*}\right),
\endeq
the jump matrices for ${\bf N}^{[n]}(\lambda;\chi,\tau)$ can be written as 
\eq
{\bf V}_{\bf N}^{[n]}(\lambda;\chi,\tau) := \begin{cases} e^{-n\varphi(\lambda;\chi,\tau)\sigma_3}{\bf S}^{-1}e^{n\varphi(\lambda;x,t)\sigma_3}, & \chi>0, \\ e^{-n\widetilde\varphi(\lambda;\chi,\tau)\sigma_3}\widetilde{\bf S}^{-1}e^{n\widetilde\varphi(\lambda;x,t)\sigma_3}, & \chi<0 \end{cases}
\endeq
for $\lambda\in\partial D_0$ (oriented clockwise).  Note that 
$\varphi(\lambda;\chi,\tau)$ and $\widetilde\varphi(\lambda;\chi,\tau)$ are 
independent of ${\bf c}$.

Our immediate goal is to understand the topology of the level curves 
$\Re(\varphi(\lambda;\chi,\tau))=0$ in the complex $\lambda$-plane as $\chi$ and $\tau$ 
vary.  Note that $\Re(\varphi(\lambda;\chi,\tau))$ is zero for all 
$\lambda\in\mathbb{R}$.  Also, observe the critical points are those $\lambda$ 
values satisfying
\eq
\label{crit-point-cubic}
2\tau(\lambda-\alpha)^3+(\chi+2\alpha \tau)(\lambda-\alpha)^2+2\beta^2\tau(\lambda-\alpha)+(\beta^2\chi-2\beta+2\alpha\beta^2\tau)=0.
\endeq  
This cubic has real coefficients.  

Assume first that $\tau=0$.  Then there are two critical points
\eq
\lambda_\pm(\chi,0) = \alpha \pm \left(\frac{2\beta-\beta^2\chi}{\chi}\right)^{1/2}
\endeq
(we write $(\cdot)^{1/2}$ for the principal branch of the square root).  
If $\chi=\frac{2}{\beta}$, then the two critical points coincide at $\alpha$.  
If $\chi>\frac{2}{\beta}$, then $\lambda_\pm(\chi,0)$ are complex conjugates.  
As we will see shortly, the ray $\{(\chi,\tau):\chi>\frac{2}{\beta},\,\tau=0\}$ 
is in the zero region, so we assume $\chi\geq\frac{2}{\beta}$.  As 
$\lambda\to\infty$, $i\lambda \chi$ is the dominant term in 
$\varphi(\lambda;\chi,0)$, so $\Re(\varphi(\lambda;\chi,0))<0$ for 
$\Im(\lambda)>0$ and $\Re(\varphi(\lambda;\chi,0))>0$ for $\Im(\lambda)<0$ if 
$|\lambda|$ is sufficiently large.  Looking at the logarithm term, 
$\Re(\varphi(\lambda;\chi,0))>0$ for $|\lambda-\xi|$ sufficiently small, and 
$\Re(\varphi(\lambda;\chi,0))<0$ for $|\lambda-\xi^*|$ sufficiently small.  
Therefore there must be a level line curve $\Re(\varphi(\lambda;\chi,0))=0$ 
completely enclosing $\lambda=\xi$ (and another around $\lambda=\xi^*$).  
Since $\Re(\varphi(\lambda;\chi,0))$ is harmonic (and not constant) away 
from $\xi$ and $\xi^*$, there can be at most one closed loop of the zero 
level lines in the upper half-plane and at most one in the lower half-plane.  
Therefore, the zero level curves must be exactly the real axis along with 
two simple loops enclosing $\xi$ and $\xi^*$ that intersect the real axis at a 
single (shared) point if $\chi=\frac{2}{\beta}$ (as shown in the left panel in 
Figure \ref{fig-re-phi1}) or that are entirely in their respective half-planes 
if $\chi>\frac{2}{\beta}$ (as shown in the center panel in Figure 
\ref{fig-re-phi1}).  By explicit calculation, we check that $\lambda_\pm(\chi,0)$ 
lie outside the two closed loops if $\chi>\frac{2}{\beta}$.  

Now fix $\chi>\frac{2}{\beta}$.  For $\tau>0$ sufficiently small, 
$\varphi(\lambda;\chi,\tau)$ has three distinct critical points.  Two of these form a 
complex conjugate pair with real part approximately $\alpha$.  We define 
$\lambda_\pm(\chi,\tau)$ to be the analytic continuation (in $\tau$) of 
$\lambda_\pm(\chi,0)$.  This analytic continuation is well defined if $\tau$ is 
sufficiently small such that $\lambda_+(\chi,S)\neq\lambda_-(\chi,S)$ for 
$0\leq S\leq \tau$.  The third critical point is negative real (for $\tau$ 
sufficiently small).  We denote this critical point by $\lambda_0(\chi,\tau)$ for 
any values of $\chi$ and $\tau$ for which $\lambda_\pm(\chi,\tau)$ are defined.  Since $\tau$ 
is nonzero, as $\lambda\to\infty$, 
$\Re(\varphi(\lambda;\chi,\tau))$ is dominated by $\Re(i\lambda^2\tau)$.  As 
the local behavior near $\xi$ and $\xi^*$ is topologically unchanged, the zero 
level curves must be topologically the same as in the case $\tau=0$ with the 
addition of an unbounded curve that, for large $|\lambda|$, is approximately 
parallel to the imaginary axis.  See the right panel in Figure 
\ref{fig-re-phi1}.  
\begin{figure}[H]
\begin{center}
\includegraphics[height=1.95in]{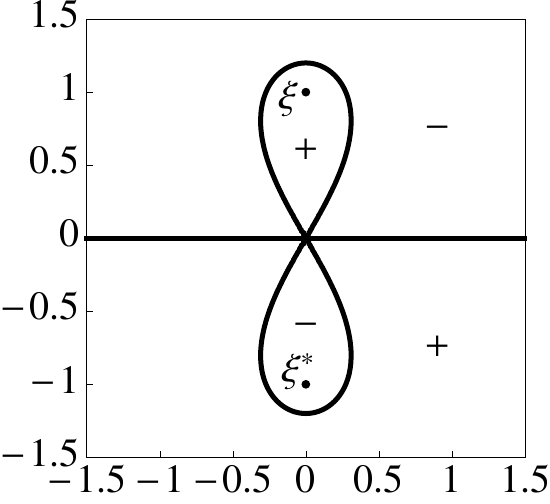}
\includegraphics[height=1.95in]{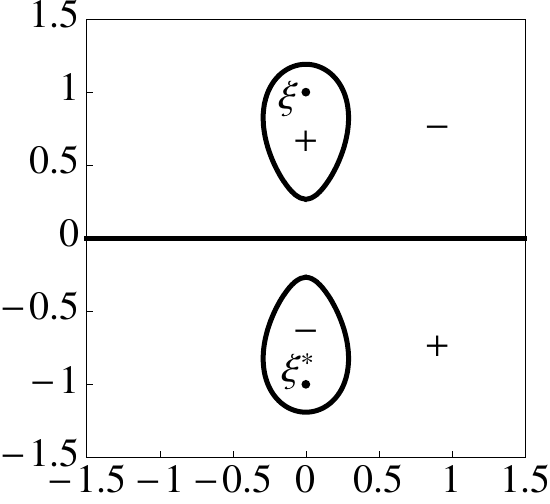}
\includegraphics[height=1.95in]{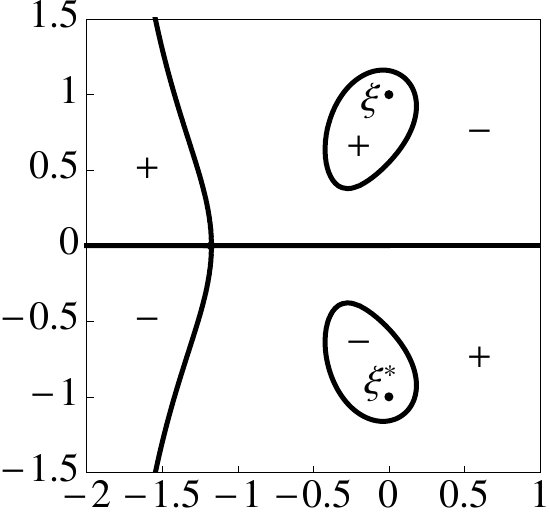}
\caption{Signature charts of $\Re(\varphi(\lambda;\chi,\tau))$ in the complex 
$\lambda$-plane for $\xi=i$.  \emph{Left:} $(\chi,\tau)=(2,0)$.  \emph{Center:} 
$(\chi,\tau)=(2.05,0)$.  \emph{Right:} $(\chi,\tau)=(2.25,0.6)$.  The left 
panel illustrates the topology of the zero level lines for general 
$\xi=\alpha+i\beta$ if 
$\chi=\frac{2}{\beta}$ and $\tau=0$, the center panel illustrates the general 
topology for $\chi>\frac{2}{\beta}$ and $\tau=0$, and the right panel 
illustrates the general topology for $\chi>\frac{2}{\beta}$ and $\tau>0$ with 
$|\tau|$ sufficiently small.}
\label{fig-re-phi1}
\end{center}
\end{figure}

For fixed $\chi>\frac{2}{\beta}$, let $\tau\to+\infty$ (the argument as 
$\tau\to-\infty$ is analogous).  Then, excluding shrinking 
neighborhoods of $\xi$ and $\xi^*$, $\Re(\varphi(\lambda;\chi,\tau))$ is well 
approximated everywhere by $\Re(i\lambda^2\tau)$.  This means that, with the 
possible exception of shrinking loops around $\xi$ and $\xi^*$, the zero 
level curves of $\Re(\varphi(\lambda;\chi,\tau))$ are the real axis and an 
unbounded curve with real part approximately $\alpha$.  Furthermore, explicit 
calculation shows that $\Re(\varphi(a+i\beta;\chi,\tau))>0$ for $a<\alpha$ (for 
$\tau$ sufficiently large), and thus there are no closed loops on which 
$\Re(\varphi(\lambda;\chi,\tau))>0$ around $\xi$ (or, by symmetry, around 
$\xi^*$).  

Therefore, for fixed $\chi>\frac{2}{\beta}$, there is (at least one) 
topological change in the zero level lines of $\Re(\varphi(\lambda;\chi,\tau))$ 
as $\tau$ changes from zero to infinity.  We show below that any $(\chi,\tau)$ 
values before the first 
transition (starting from $\tau=0$) are in the zero region.  The topological 
change can happen in one of three ways, all of which occur for certain values 
of $\chi$, $\tau$, and $\xi$.  The first way is for $\lambda_+(\chi,\tau)$ and 
$\lambda_-(\chi,\tau)$ to coincide on the real axis at a point distinct from 
$\lambda_0(\chi,t)$.  This transition is illustrated in the first panel in 
Figure \ref{fig-re-phi2} (as well as the first panel in Figure 
\ref{fig-re-phi1} in the special case $\tau=0$).  We conjecture that, as $\tau$ 
increases from this 
configuration, a single band will open in the model Riemann-Hilbert problem.  
This suggests a transition from the zero region to a region in which the 
solution to \eqref{nls} is nonzero and non-oscillatory.  A necessary algebraic 
condition for this to occur is for the discriminant of \eqref{crit-point-cubic} 
to be zero, which is equivalent to 
\eq
\label{zero-boundary1}
\begin{split}
(16\alpha^4\beta^2+32\alpha^2\beta^4+16\beta^6)\tau^4 + (32\alpha^3\beta^2 \chi - 16\alpha^3\beta + 32\alpha\beta^4 \chi - 144\alpha\beta^3)\tau^3 & \\ 
+ (24\alpha^2\beta^2\chi^2 - 24\alpha^2\beta \chi + 8\beta^4 \chi^2 - 72\beta^3 \chi + 108\beta^2)\tau^2 & \\ 
+ (8\alpha\beta^2 \chi^3 - 12\alpha\beta \chi^2)\tau + (\beta^2\chi^4 - 2\beta \chi^3) & = 0.
\end{split}
\endeq
If $\xi=i$, this simplifies to \eqref{xi-i-boundary-quadratic}, a quadratic 
relation in $\tau^2$.
\begin{figure}[H]
\begin{center}
\includegraphics[height=2in]{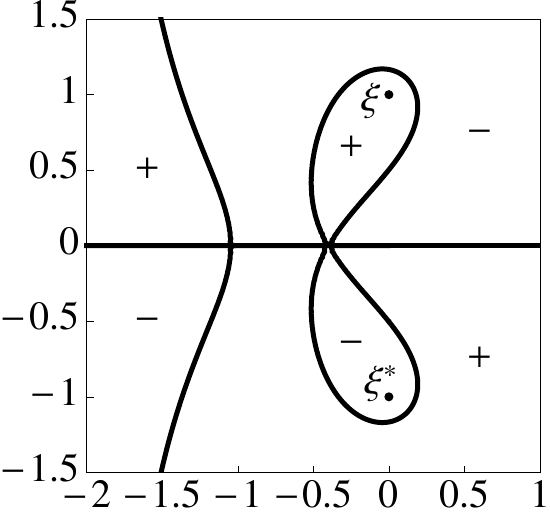}
\includegraphics[height=2in]{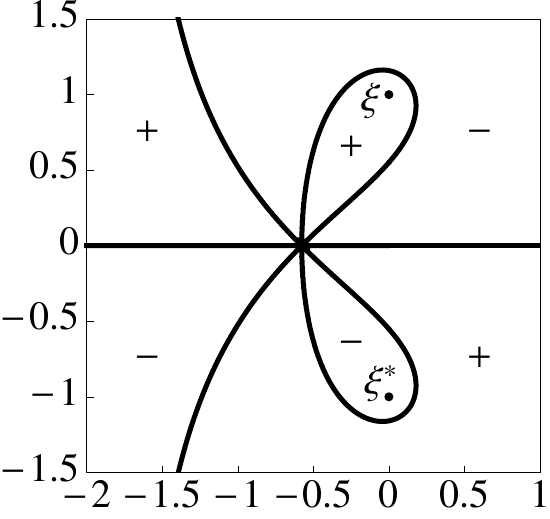}
\includegraphics[height=2in]{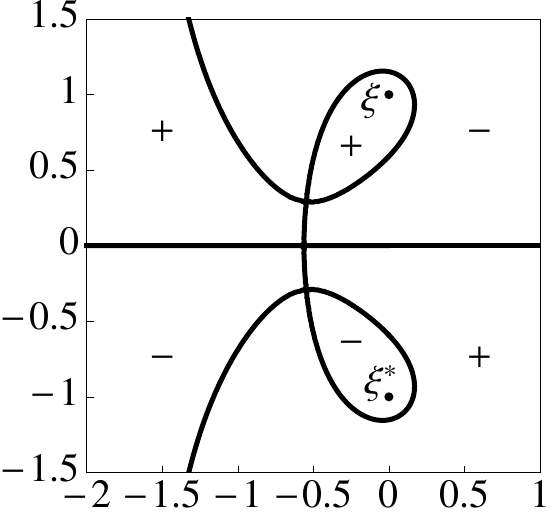}
\caption{Signature charts of $\Re(\varphi(\lambda;\chi,\tau))$ in the complex 
$\lambda$-plane for $\xi=i$.  \emph{Left:} $(\chi,\tau)\approx(2.2,0.595)$.  
\emph{Center:}  $(\chi,\tau)=(\frac{9}{4},\frac{3\sqrt{3}}{8})$.  \emph{Right:}  
$(\chi,\tau)\approx(2.3,0.649)$.  The left panel illustrates the boundary 
between the zero region with $\chi>0$ and what we conjecture is a nonzero 
non-oscillatory region.  The right panel illustrates the boundary between the 
zero region with $\chi>0$ and what we conjecture is a nonzero single-phase 
region.  The center panel illustrates the single point lying at the corner of 
the three different regions.}
\label{fig-re-phi2}
\end{center}
\end{figure}

The second way the 
topological change can occur is if $\lambda_+(\chi,\tau)$, $\lambda_-(\chi,\tau)$, and 
$\lambda_0(\chi,\tau)$ all coincide.  This is illustrated in the second panel in 
Figure \ref{fig-re-phi2}.  This double-critical behavior appears to 
correspond with a point at the corner of three different regions.  If 
$\alpha=0$, a necessary condition for this triple critical point is the 
discriminant of \eqref{zero-boundary1} must be zero, which 
occurs exactly (for $\tau>0$) at the critical point
\eq
(\chi_c,\tau_c)=\left(\frac{9}{4\beta},\frac{3\sqrt{3}}{8\beta^2}\right).
\endeq

The third 
way for the topological change to occur, illustrated in the third panel in 
Figure \ref{fig-re-phi2}, is when $\lambda_+(\chi,\tau)$ and $\lambda_-(\chi,\tau)$ 
simultaneously intersect a zero level line of $\Re(\varphi(\lambda;\chi,\tau))$ off 
the real axis.  We expect this to correspond to an opening of two bands in the 
model Riemann-Hilbert problem and a transition between the zero region and a 
region in which the solution to \eqref{nls} has single-phase oscillations with  
period of order $n^{-1}$.  A necessary algebraic condition for this 
transition is the explicit (although transcendental) algebro-logarithmic 
relation \eqref{zero-boundary2}.
In Figure \ref{fig-zero-region-boundary} we illustrate the boundary of the 
zero region computed using \eqref{zero-boundary1} and \eqref{zero-boundary2} 
for $\xi=i$ and $\xi=\frac{1}{2}+2i$.  

At this point we have proven the existence of a connected open region in the 
$\chi \tau$-plane containing the ray $\{(\chi,\tau):\chi>\frac{2}{\beta},\tau=0\}$ such that
\begin{itemize}
\item The point $\lambda=\xi$ is enclosed by a simple loop on which 
$\Re(\varphi(\lambda;\chi,\tau))=0$.  For $\lambda$ immediately outside this loop 
$\Re(\varphi(\lambda;\chi,\tau))<0$, and for $\lambda$ immediately inside this loop 
$\Re(\varphi(\lambda;\chi,\tau))>0$.
\item The point $\lambda=\xi^*$ is enclosed by a simple loop on which 
$\Re(\varphi(\lambda;\chi,\tau))=0$.  For $\lambda$ immediately outside this 
loop $\Re(\varphi(\lambda;\chi,\tau))>0$, and for $\lambda$ immediately inside 
this loop $\Re(\varphi(\lambda;\chi,\tau))<0$.
\end{itemize}
We denote the maximal region satisfying these conditions $\mathcal{Z}_+$.  
Furthermore, if $(\chi,\tau)\in\mathcal{Z}_+$, then it is immediate from 
\eqref{phi} and \eqref{phitilde} that $(-\chi,-\tau)$ has complementary 
properties:
\begin{itemize}
\item The point $\lambda=\xi$ is enclosed by a simple loop on which 
$\Re(\widetilde\varphi(\lambda;-\chi,-\tau))=0$.  For $\lambda$ immediately 
outside this loop $\Re(\widetilde\varphi(\lambda;-\chi,-\tau))>0$, and for 
$\lambda$ immediately inside this loop 
$\Re(\widetilde\varphi(\lambda;-\chi,-\tau))<0$.
\item The point $\lambda=\xi^*$ is enclosed by a simple loop on which 
$\Re(\widetilde\varphi(\lambda;-\chi,-\tau))=0$.  For $\lambda$ immediately 
outside this loop $\Re(\widetilde\varphi(\lambda;-\chi,-\tau))<0$, and for 
$\lambda$ immediately inside this loop 
$\Re(\widetilde\varphi(\lambda;-\chi,-\tau))>0$.
\end{itemize}
We denote this complementary region by $\mathcal{Z}_-$, and we call 
$\mathcal{Z}:=\mathcal{Z}_+\cup\mathcal{Z}_-$ the \emph{zero region}.  We now 
use nonlinear steepest-descent analysis to show that $\psi^{[2n]}(x,t)$ is 
exponentially close to zero in the zero region as $n\to\infty$.

If $\chi>0$, we denote the bounded region in the $\lambda$-plane containing 
$\xi$ in which 
$\Re(\varphi(\lambda;\chi,\tau))>0$ by $D_\xi$, and the bounded region 
containing $\xi^*$ in which $\Re(\varphi(\lambda;\chi,\tau))<0$ by $D_{\xi^*}$. 
Similarly, if $\chi<0$, we denote the bounded region containing $\xi$ in which 
$\Re(\varphi(\lambda;\chi,\tau))<0$ by $D_\xi$, and the bounded region 
containing $\xi^*$ in which $\Re(\varphi(\lambda;\chi,\tau))>0$ by $D_{\xi^*}$. 
Recall that the jump for ${\bf N}^{[n]}(\lambda;\chi,\tau)$ is defined on the 
loop $\partial D_0$ enclosing both $\xi$ and $\xi^*$.  The next step in the 
analysis is to deform the jump contour from $\partial D_0$ to 
$\partial D_\xi\cup\partial D_{\xi^*}$.  Define 
\eq
{\bf O}^{[n]}(\lambda;\chi,\tau):= \begin{cases} {\bf N}^{[n]}(\lambda;\chi,\tau){\bf V}_{\bf N}^{[n]}(\lambda;\chi,\tau), & \lambda\in D_0\cap(D_\xi\cup D_{\xi^*})^\mathsf{c}, \\ {\bf N}^{[n]}(\lambda;\chi,\tau){\bf V}_{\bf N}^{[n]}(\lambda;\chi,\tau)^{-1}, & \lambda\in (D_\xi\cup D_{\xi^*})\cap D_0^\mathsf{c}, \\ {\bf N}^{[n]}(\lambda;\chi,\tau), & \text{otherwise}.
\end{cases}
\endeq
Then, orienting $\partial D_\xi$ and $\partial D_{\xi^*}$ clockwise, the 
function ${\bf O}^{[n]}(\lambda;\chi,\tau)$ is analytic for 
$\lambda\notin \partial D_\xi\cup \partial D_{\xi^*}$, satisfies 
${\bf O}^{[n]}_+(\lambda;\chi,\tau) = {\bf O}^{[n]}_-(\lambda;\chi,\tau) {\bf V}_{\bf O}^{[n]}(\lambda;\chi,\tau)$
for $\lambda\in\partial D_\xi\cup\partial D_{\xi^*}$, where 
${\bf V}_{\bf O}^{[n]}(\lambda;\chi,\tau) = {\bf V}_{\bf N}^{[n]}(\lambda;\chi,\tau)$, and 
${\bf O}^{[n]}(\lambda;\chi,\tau)\to\mathbb{I}$ as $\lambda\to\infty$.
Now observe we have the following four factorizations:
\begin{eqnarray}
\label{Dxi-Xpos}
{\bf S}^{-1} & = & \bpm 1 & \frac{c_2^*}{c_1} \\ 0 & 1 \epm \bpm \frac{|{\bf c}|}{c_1} & 0 \\ 0 & \frac{c_1}{|{\bf c}|} \epm \bpm 1 & 0 \\ -\frac{c_2}{c_1} & 1 \epm \quad\quad (\text{use for }\lambda\in\partial D_\xi,\,\,\chi>0), \\
\label{Dxic-Xpos}
{\bf S}^{-1} & = & \bpm 1 & 0 \\ -\frac{c_2}{c_1^*} & 1 \epm \bpm \frac{c_1^*}{|{\bf c}|} & 0 \\ 0 & \frac{|{\bf c}|}{c_1^*} \epm \bpm 1 & \frac{c_2^*}{c_1^*} \\ 0 & 1 \epm \quad\quad (\text{use for }\lambda\in\partial D_{\xi^*},\,\,\chi>0), \\
\label{Dxi-Xneg}
\widetilde{\bf S}^{-1} & = & \bpm 1 & 0 \\ \frac{c_1^*}{c_2} & 1 \epm \bpm \frac{c_2}{|{\bf c}|} & 0 \\ 0 & \frac{|{\bf c}|}{c_2} \epm \bpm 1 & -\frac{c_1}{c_2} \\ 0 & 1 \epm \quad\quad (\text{use for }\lambda\in\partial D_\xi,\,\,\chi<0), \\
\label{Dxic-Xneg}
\widetilde{\bf S}^{-1} & = & \bpm 1 & -\frac{c_1}{c_2^*} \\ 0 & 1 \epm \bpm \frac{|{\bf c}|}{c_2^*} & 0 \\ 0 & \frac{c_2^*}{|{\bf c}|} \epm \bpm 1 & 0 \\ \frac{c_1^*}{c_2^*} & 1 \epm \quad\quad (\text{use for }\lambda\in\partial D_{\xi^*},\,\,\chi<0). 
\end{eqnarray}
For convenience we indicate when we will use each factorization;  of course 
each relation is an algebraic identity that holds independent of $\lambda$ 
or $\chi$.  First, suppose $\chi>0$.  We define four simple clockwise-oriented loops 
$\Sigma_\xi^\text{(out)}$, $\Sigma_\xi^\text{(in)}$, 
$\Sigma_{\xi^*}^\text{(out)}$, and $\Sigma_{\xi^*}^\text{(in)}$ such that:
\begin{itemize}
\item $\Sigma_\xi^\text{(out)}$ encloses $D_\xi$ and lies entirely in the 
region in which $\Re(\varphi(\lambda;\chi,\tau))<0$.
\item $\Sigma_\xi^\text{(in)}$ encloses $\xi$ and lies entirely in $D_\xi$ (so 
that $\Re(\varphi(\lambda;\chi,\tau))>0$).
\item $\Sigma_{\xi^*}^\text{(out)}$ encloses $D_{\xi^*}$ and lies entirely in 
the region in which $\Re(\varphi(\lambda;\chi,\tau))>0$.
\item $\Sigma_{\xi^*}^\text{(in)}$ encloses $\xi^*$ and lies entirely in 
$D_{\xi^*}$ (so that $\Re(\varphi(\lambda;\chi,\tau))<0$).
\end{itemize}
Also define the following four annular regions:
\begin{itemize}
\item $L_\xi^\text{(out)}$ is bounded by $\partial D_\xi$ and $\Sigma_\xi^\text{(out)}$.
\item $L_\xi^\text{(in)}$ is bounded by $\partial D_\xi$ and $\Sigma_\xi^\text{(in)}$.
\item $L_{\xi^*}^\text{(out)}$ is bounded by $\partial D_{\xi^*}$ and $\Sigma_{\xi^*}^\text{(out)}$.
\item $L_{\xi^*}^\text{(in)}$ is bounded by $\partial D_{\xi^*}$ and $\Sigma_{\xi^*}^\text{(in)}$.
\end{itemize}
For $(\chi,\tau)\in\mathcal{Z}$ and $\chi>0$, define 
\eq
{\bf P}^{[n]}(\lambda;\chi,\tau) := \begin{cases} 
{\bf O}^{[n]}(\lambda;\chi,\tau)\bpm 1 & \frac{c_2^*}{c_1}e^{-2n\varphi(\lambda;\chi,\tau)} \\ 0 & 1 \epm, & \lambda\in L_\xi^\text{(in)}, \\ 
{\bf O}^{[n]}(\lambda;\chi,\tau)\bpm 1 & 0 \\ -\frac{c_2}{c_1}e^{2n\varphi(\lambda;\chi,\tau)} & 1 \epm^{-1}, & \lambda\in L_\xi^\text{(out)}, \\ 
{\bf O}^{[n]}(\lambda;\chi,\tau)\bpm 1 & 0 \\ -\frac{c_2}{c_1^*}e^{2n\varphi(\lambda;\chi,\tau)} & 1 \epm, & \lambda\in L_{\xi^*}^\text{(in)}, \\ 
{\bf O}^{[n]}(\lambda;\chi,\tau)\bpm 1 & \frac{c_2^*}{c_1^*}e^{-2n\varphi(\lambda;\chi,\tau)} \\ 0 & 1 \epm^{-1}, & \lambda\in L_{\xi^*}^\text{(out)}, \\ 
{\bf O}^{[n]}(\lambda;\chi,\tau), & \text{otherwise}. \end{cases}
\endeq
Then ${\bf P}^{[n]}(\lambda;\chi,\tau)$ is analytic for 
$\lambda\notin\partial D_\xi\cup\partial D_{\xi^*}\cup\Sigma_\xi^\text{(out)}\cup\Sigma_\xi^\text{(in)}\cup\Sigma_{\xi^*}^\text{(out)}\cup\Sigma_{\xi^*}^\text{(in)}$, 
has the normalization ${\bf P}^{[n]}(\lambda;\chi,\tau)\to\mathbb{I}$ as $\lambda\to\infty$, and 
has the jumps ${\bf P}_+^{[n]}(\lambda;\chi,\tau) = {\bf P}_-^{[n]}(\lambda;\chi,\tau){\bf V}_{\bf P}^{[n]}(\lambda;\chi,\tau)$, where
\eq
{\bf V}_{\bf P}^{[n]}(\lambda;\chi,\tau) := \begin{cases} \bpm 1 & \frac{c_2^*}{c_1}e^{-2n\varphi(\lambda;\chi,\tau)} \\ 0 & 1 \epm, & \lambda\in\Sigma_\xi^\text{(in)}, \\ \bpm \frac{|{\bf c}|}{c_1} & 0 \\ 0 & \frac{c_1}{|{\bf c}|} \epm, & \lambda\in\partial D_\xi, \\ \bpm 1 & 0 \\ -\frac{c_2}{c_1}e^{2n\varphi(\lambda;\chi,\tau)} & 1 \epm, & \lambda\in\Sigma_\xi^\text{(out)}, \\ \bpm 1 & 0 \\ -\frac{c_2}{c_1^*}e^{2n\varphi(\lambda;\chi,\tau)} & 1 \epm, & \lambda\in\Sigma_{\xi^*}^\text{(in)}, \\ \bpm \frac{c_1^*}{|{\bf c}|} & 0 \\ 0 & \frac{|{\bf c}|}{c_1^*} \epm, & \lambda\in\partial D_{\xi^*}, \\ \bpm 1 & \frac{c_2^*}{c_1^*}e^{-2n\varphi(\lambda;\chi,\tau)} \\ 0 & 1 \epm, & \lambda\in\Sigma_{\xi^*}^\text{(out)}. \end{cases}
\endeq
The jumps for ${\bf P}^{[n]}(\lambda;\chi,\tau)$ on the contours $\Sigma_\xi^\text{(out)}$, 
$\Sigma_\xi^\text{(in)}$, $\Sigma_{\xi^*}^\text{(out)}$, and $\Sigma_{\xi^*}^\text{(in)}$ 
all decay exponentially to $\mathbb{I}$ as $n\to\infty$.  Thus we define 
${\bf Q}(\lambda;\chi,\tau)$ as the $n$-independent solution to the following Riemann-Hilbert 
problem:
\begin{rhp}[The model problem in the zero region with $\chi>0$]
Fix a pair of nonzero complex numbers $(c_1,c_2)$, along with a pair of real numbers 
$(\chi,\tau)\in\mathcal{Z}$ with $\chi>0$.  Determine the unique $2\times 2$ matrix 
${\bf Q}(\lambda;\chi,\tau)$ with the following properties:
\begin{itemize}
\item[]\textbf{Analyticity:}  ${\bf Q}(\lambda;\chi,\tau)$
is analytic for $\lambda\in\mathbb{C}$ except on $\partial D_\xi\cup\partial D_{\xi^*}$, 
where it achieves continuous boundary values.   
\item[]\textbf{Jump condition:}  The boundary values taken by ${\bf Q}(\lambda;\chi,\tau)$ are 
related by the jump conditions
${\bf Q}_+(\lambda;\chi,\tau)={\bf Q}_-(\lambda;\chi,\tau){\bf V}_{\bf Q}(\lambda;\chi,\tau)$, 
where
\eq
{\bf V}_{\bf Q}(\lambda;\chi,\tau) := \begin{cases} \bbm \frac{|{\bf c}|}{c_1} & 0 \\ 0 & \frac{c_1}{|{\bf c}|} \ebm, & \lambda\in\partial D_\xi, \vspace{.05in} \\ \bbm \frac{c_1^*}{|{\bf c}|} & 0 \\ 0 & \frac{|{\bf c}|}{c_1^*} \ebm, & \lambda\in\partial D_{\xi^*}. \end{cases}
\endeq 
\item[]\textbf{Normalization:}  As $\lambda\to\infty$, the matrix ${\bf Q}(\lambda;\chi,\tau)$ 
satisfies the condition
\begin{equation}
{\bf Q}(\lambda;\chi,\tau) = \mathbb{I}+\mathcal{O}(\lambda^{-1})
\end{equation}
with the limit being uniform with respect to direction.
\end{itemize}
\label{rhp:Q}
\end{rhp}
This Riemann-Hilbert problem reduces to two scalar problems, and as such can be solved 
explicitly using the Plemelj formula.  However, we will not need the exact formula, only 
that ${\bf Q}(\lambda;\chi,\tau)$ is diagonal:
\eq
{\bf Q}(\lambda;\chi,\tau) \equiv \bpm Q_{11}(\lambda;\chi,\tau) & 0 \\ 0 & Q_{22}(\lambda;\chi,\tau) \epm.
\endeq
Finally, we define the error function by the ratio 
\eq
{\bf R}^{[n]}(\lambda;\chi,\tau):={\bf P}^{[n]}(\lambda;\chi,\tau){\bf Q}(\lambda;\chi,\tau)^{-1}.
\endeq
The jumps across $\partial D_\xi$ and $\partial D_{\xi^*}$ cancel exactly.  Therefore, 
${\bf R}^{[n]}(\lambda;\chi,\tau)$ is analytic for 
$\lambda\notin\Sigma_\xi^\text{(in)}\cup\Sigma_\xi^\text{(out)}\cup\Sigma_{\xi^*}^\text{(in)}\cup\Sigma_{\xi^*}^\text{(out)}$, 
${\bf R}^{[n]}(\lambda;\chi,\tau)\to\mathbb{I}$ as $\lambda\to\infty$, and 
${\bf R}_+^{[n]}(\lambda;\chi,\tau)={\bf R}_-^{[n]}(\lambda;\chi,\tau){\bf V}_{\bf R}^{[n]}(\lambda;\chi,\tau)$, 
where 
\eq
{\bf V}_{\bf R}^{[n]}(\lambda;\chi,\tau)={\bf Q}_-(\lambda;\chi,\tau){\bf V}_{\bf P}^{[n]}(\lambda;\chi,\tau){\bf Q}_+(\lambda;\chi,\tau)^{-1}.
\endeq
Therefore, the jump matrices for ${\bf R}^{[n]}(\lambda;\chi,\tau)$ are 
exponentially close to the identity matrix as $n\to\infty$.  From standard 
nonlinear steepest-descent analysis (see, for instance, \cite{DeiftZ:1993} or 
\cite[Appendix B]{BuckinghamM:2013}), there is a constant $d>0$ such that 
${\bf R}^{[n]}(\lambda;\chi,\tau)=\mathbb{I}+\mathcal{O}(e^{-dn})$ uniformly 
in $\lambda$ and uniformly in $(\chi,\tau)$ bounded a fixed distance away 
from the edge of the zero region.  This implies that 
${\bf P}^{[n]}(\lambda;\chi,\tau)$ is exponentially close to 
${\bf Q}(\lambda;\chi,\tau)$ as $n\to\infty$.  Unwinding the transformations 
${\bf M}^{[n]}(\lambda;\chi,\tau)\to{\bf N}^{[n]}(\lambda;\chi,\tau)\to{\bf O}^{[n]}(\lambda;\chi,\tau)\to{\bf P}^{[n]}(\lambda;\chi,\tau)$, 
for fixed $(\chi,\tau)\in\mathcal{Z}$ with $\chi>0$, we have in particular 
\eq
[{\bf M}^{[n]}(\lambda;n\chi,n\tau)]_{12} = \mathcal{O}(e^{-dn})
\endeq
for some constant $d>0$ uniformly in $\lambda$.  Thus, from 
\eqref{psi-from-M}, for fixed $(\chi,\tau)\in\mathcal{Z}$ with $\chi>0$, 
\eqref{zero-region-result} holds for some constant $d>0$.  The analysis for 
$(\chi,\tau)\in\mathcal{Z}$ with $\chi<0$ follows exactly the same logic, 
only starting from the factorizations 
\eqref{Dxi-Xneg}--\eqref{Dxic-Xneg} instead of 
\eqref{Dxi-Xpos}--\eqref{Dxic-Xpos}.  This concludes the proof of Theorem 
\ref{thm-zero-reg} and the asymptotic description of the zero region.

\section{The near-field limit and the Painlev\'e-III hierarchy}
\label{sec-near-field}

We now prove Theorem \ref{thm-near-field}.  Our first move is to obtain an 
$n$-independent Riemann-Hilbert problem whose solution is a good 
approximation of ${\bf M}^{[n]}(\lambda;x,t)$ in a suitable rescaling near 
$(x,t)=(0,0)$.  Recall the near-field scalings $X:=nx$, $T:=n^2t$ (see 
\eqref{near-field-scaling}).  We also scale the spectral parameter $\lambda$:
\eq
\Lambda:=n^{-1}\lambda.
\endeq
With this scaling in mind we recall that the radius of the jump contour $D_0$ 
for ${\bf M}^{[n]}(\lambda;x,t)$ is arbitrary (as long as it encloses $\xi$). 
Thus we choose $\partial D_0$ to be a circle of radius $n$ centered at the 
origin (and hence $|\Lambda|=1$).  Applying these scalings for $x$, $t$, and 
$\lambda$ to the jump matrix for ${\bf M}^{[n]}(\lambda;x,t)$ in 
\eqref{eq:jump-rhp-M} gives
\eq
\label{M-jump-approx}
\begin{split}
e^{-i(\lambda x+\lambda^2 t)\sigma_3}\mathbf{S} \left( \frac{\lambda-\xi}{\lambda-\xi^*}\right)^{n\sigma_3} \mathbf{S}^{-1}e^{i(\lambda x+\lambda^2 t)\sigma_3}\Big\vert_{x=n^{-1}X,\, t=n^{-2}T,\,\lambda=n\Lambda}= & \\
\left(\mathbb{I} + \mathcal{O}(n^{-1}) \right)e^{-i(\Lambda X+\Lambda^2 T)\sigma_3}\mathbf{S}e^{-2i\beta\Lambda^{-1}\sigma_3}&\mathbf{S}^{-1}e^{i(\Lambda X+\Lambda^2 T)\sigma_3}.
\end{split}
\endeq
Neglecting the terms of $\mathcal{O}(n^{-1})$, we arrive (formally) at the 
near-field Riemann-Hilbert problem.
\begin{rhp}[The near-field problem]
Let $(X,T)\in\mathbb{R}^2$ be fixed but arbitrary parameters.  Find the 
unique $2\times 2$ matrix-valued function $\mathbf{A}(\Lambda;X,T)$ with the 
following properties:
\begin{itemize}
\item[]\textbf{Analyticity:} $\mathbf{A}(\Lambda;X,T)$ is analytic in 
$\Lambda$ for $|\Lambda|\neq 1$, and it takes continuous boundary values from 
the interior and exterior of $|\Lambda|=1$.
\item[]\textbf{Jump condition:} The boundary values on the jump contour 
(oriented clockwise) follow the relation
\begin{equation}
\mathbf{A}_+(\Lambda;X,T) = \mathbf{A}_-(\Lambda;X,T)e^{-i(\Lambda X+\Lambda^2 T)\sigma_3}\mathbf{S} e^{-2i\beta\Lambda^{-1}\sigma_3} \mathbf{S}^{-1}e^{i(\Lambda X+\Lambda^2 T)\sigma_3},\quad|\Lambda|=1.
\label{eq:jump-rhp-A}
\end{equation}
\item[]\textbf{Normalization:} $\mathbf{A}(\Lambda;X,T)\to\mathbb{I}$ as $\Lambda\to\infty$.
\end{itemize}
\label{rhp:limiting}
\end{rhp}
If $(c_1,c_2)=(1,\pm 1)$ and $\xi=i$, then Riemann-Hilbert Problem 
\ref{rhp:limiting} is exactly Riemann-Hilbert Problem 3 in \cite{BilmanLM:2018} 
used to define $\Psi(X,T;(1,\pm 1))$.  We now define 
\eq
\label{Psi-def}
\Psi(X,T;{\bf c}):=2i\lim_{\Lambda\to\infty}\Lambda[{\bf A}(\Lambda;X,T)]_{12}.
\endeq
As we show in Theorem \ref{thm-near-field}, this function 
$\Psi(X,T;{\bf c})$ is the scaled limit of $\psi^{[2n]}(x,t;{\bf c})$ in the 
near field.  From here on we assume $\xi=i$.

\subsection{The function $\Psi(X,T)$ and the NLS equation:  Proof of Theorem \ref{thm-near-field}(a).}
To prove \eqref{psi-Psi-approx} we follow the standard argument used in 
\cite[Theorem 1]{BilmanLM:2018}.  To measure the difference between Riemann-Hilbert 
Problem \ref{rhp:psi-n} (appropriately scaled) and Riemann-Hilbert Problem 
\ref{rhp:limiting}, define the ratio matrix 
\eq
{\bf F}(\Lambda;X,T;{\bf c}) := {\bf M}^{[n]}\left(n\Lambda;\frac{X}{n},\frac{T}{n^2};{\bf c}\right){\bf A}(\Lambda;X,T;{\bf c})^{-1}.
\endeq
Then ${\bf F}(\Lambda;X,T)$ is analytic for $|\Lambda|\neq 1$, whereas for 
$|\Lambda|=1$ we have 
\eq
\begin{split}
{\bf F}_+(\Lambda;X,T) & = {\bf F}_-(\Lambda;X,T){\bf A}_-(\Lambda;X,T)(\mathbb{I}+\mathcal{O}(n^{-1})){\bf A}_-(\Lambda;X,T)^{-1} \\ 
   & = {\bf F}_-(\Lambda;X,T)(\mathbb{I}+\mathcal{O}(n^{-1})).
\end{split}
\endeq
Here the first line follows from \eqref{M-jump-approx}, while the second line 
follows from the boundedness of ${\bf A}(\Lambda;X,T)$ and the fact that 
$\det{\bf A}(\Lambda;X,T)\equiv 1$.  Since we also have 
${\bf F}(\Lambda;X,T)\to\mathbb{I}$ as $\Lambda\to\infty$, the function 
${\bf F}(\Lambda;X,T)$ satisfies a small-norm Riemann-Hilbert problem, from 
which it follows \cite{DeiftZ:1993} that 
\eq
{\bf F}(\Lambda;X,T) = \mathbb{I} + \mathcal{O}(n^{-1})
\endeq
uniformly for compact regions in the $XT$-plane.  Starting from 
\eqref{psi-from-M}, 
\eq
\begin{split}
\frac{1}{n}\psi^{[2n]}\left(\frac{X}{n},\frac{T}{n^2}\right) & = \frac{2i}{n}\lim_{\lambda\to\infty}\lambda\left[{\bf M}^{[n]}\left(\lambda;\frac{X}{n},\frac{T}{n^2}\right)\right]_{12} \\
   & = \frac{2i}{n}\lim_{\Lambda\to\infty} n\Lambda\left([{\bf F}(\Lambda;X,T)]_{11}[{\bf A}(\Lambda;X,T)]_{12} + [{\bf F}(\Lambda;X,T)]_{12}[{\bf A}(\Lambda;X,T)]_{22} \right) \\
   & = 2i \lim_{\Lambda\to\infty}\Lambda[{\bf A}(\Lambda;X,T)]_{12} + \mathcal{O}(n^{-1}) \\ 
   & = \Psi(X,T) + \mathcal{O}(n^{-1})
\end{split}
\endeq
uniformly in $X$ and $T$ chosen from compact sets.

To prove that $\Psi(X,T;{\bf c})$ satisfies the nonlinear Schr\"odinger 
equation, define 
\eq
\label{K-ito-A}
{\bf K}(\Lambda;X,T):={\bf A}(\Lambda;X,T)e^{-i(\Lambda X+\Lambda^2T)\sigma_3}.
\endeq
Then, following the proof of \cite[Proposition 3]{BilmanLM:2018} with their 
matrix ${\bf Q}=\frac{1}{\sqrt{2}}\bbm 1 & -1 \\ 1 & 1 \ebm$ replaced with 
the more general matrix ${\bf S}$ (indeed, ${\bf S}={\bf Q}$ if 
${\bf c}=(1,1)$), the function ${\bf K}(\Lambda;X,T)$ satisfies the system of 
overdetermined ordinary differential equations 
\eq
\begin{split}
\frac{\partial{\bf K}}{\partial X}(\Lambda;X,T) & = \bbm -i\Lambda & \Psi(X,T) \\ -\Psi(X,T) & i\Lambda \ebm {\bf K}(\Lambda;X,T), \\ 
\frac{\partial{\bf K}}{\partial T}(\Lambda;X,T) & = \bbm -i\Lambda^2+\frac{i}{2}|\Psi(X,T)|^2 & \Lambda\Psi(X,T) + \frac{i}{2}\Psi_X(X,T) \\ -\Lambda\Psi(X,T)^* + \frac{i}{2}\Psi_X(X,T)^* & i\Lambda^2 - \frac{i}{2}|\Psi(X,T)|^2 \ebm {\bf K}(\Lambda;X,T), 
\end{split}
\endeq
which is simply the Lax pair \eqref{lax-pair-xt} with $\psi$, $x$, $t$, and 
$\lambda$ replaced by $\Psi$, $X$, $T$, and $\Lambda$, respectively.  This means 
\eqref{nls-XT} is equivalent to the condition 
${\bf K}_{XT}(\Lambda;X,T)={\bf K}_{TX}(\Lambda;X,T)$, and so $\Psi(X,T)$ 
satisfies \eqref{nls-XT}.  This completes the proof of 
Theorem \ref{thm-near-field}(a).

\subsection{The function $\Psi(X,T)$ and the $P_{III}$ hierarchy:  Proof of Theorem \ref{thm-near-field}(b).}
In \cite[\S 3.2.1]{BilmanLM:2018} it is shown that $\Psi(X,T;(1,\pm 1))$ 
satisfies \eqref{fourth-order-ode} by deriving a Lax pair in $\Lambda$ and 
$X$ for the function 
\eq
{\bf B}(\Lambda;X,T)e^{-i(\Lambda X+\Lambda^2 T+2\Lambda^{-1})\sigma_3},
\endeq  
where 
\begin{equation}
\label{B-ito-A}
\mathbf{B}(\Lambda;X,T):= \begin{cases}
\mathbf{A}(\Lambda;X,T)e^{-i(\Lambda X+\Lambda^2 T)\sigma_3} \mathbf{S}^{-1} e^{i(\Lambda X+\Lambda^2 T)\sigma_3},&|\Lambda|<1, \\
\mathbf{A}(\Lambda;X,T)e^{2i\Lambda^{-1}\sigma_3},&|\Lambda|>1.
\end{cases}
\end{equation}
The derivation depends on the fact that the jump for this function across the 
unit circle is constant, but not on the particular constant jump matrix.  
Since we have 
\eq
({\bf B}(\Lambda;X,T)e^{-i(\Lambda X+\Lambda^2 T+2\Lambda^{-1})\sigma_3})_+ = ({\bf B}(\Lambda;X,T)e^{-i(\Lambda X+\Lambda^2 T+2\Lambda^{-1})\sigma_3})_-{\bf S}^{-1}, \quad |\Lambda|=1,
\endeq
the derivation in \cite{BilmanLM:2018} goes through unchanged, and 
$\Psi(X,T;{\bf c})$ satisfies \eqref{fourth-order-ode} for general 
${\bf c}\in(\mathbb{C}^*)^2$. 
We now calculate $\Psi(0,0;{\bf c})$ and $\Psi_X(0,0;{\bf c})$.  
\begin{lemma} 
\label{Psi-lemma}
For any ${\bf c}=(c_1,c_2)\in(\mathbb{C}^*)^2$ and $\xi=i$, 
$\Psi(0,0;{\bf c})=8\frac{ c_1 c_2^*}{|\mathbf{c}|^2}$.
\end{lemma}
\begin{proof}
Using the (Riemann-Hilbert) properties of $\mathbf{A}(\Lambda;X,T)$, we see that 
$\mathbf{B}(\Lambda;X,T)$ defined in \eqref{B-ito-A} is unimodular and analytic for $\Lambda\neq 0$ away from 
$|\Lambda|=1$, and has the property $\mathbf{B}(\Lambda;X,T)\to\mathbb{I}$ as 
$\Lambda\to\infty$. The jump condition satisfied by $\mathbf{B}(\Lambda;X,T)$ is
\begin{equation}
\mathbf{B}_+(\Lambda;X,T) = \mathbf{B}_-(\Lambda;X,T) e^{-i(\Lambda X + \Lambda^2 T + 2 \Lambda^{-1})\sigma_3} \mathbf{S}^{-1}e^{i(\Lambda X + \Lambda^2 T + 2 \Lambda^{-1})\sigma_3},\quad |\Lambda|=1.
\end{equation}
Moreover, since 
$e^{2i\Lambda^{-1}\sigma_3}=\mathbb{I} + 2i \sigma_3 \Lambda^{-1} + \mathcal{O}(\Lambda^{-2})$ 
as $\Lambda\to \infty$, the identity \eqref{Psi-def} implies that we recover 
$\Psi(X,T;{\bf c})$ via the limit
\begin{equation}
\label{Psi-ito-B}
\Psi(X,T;{\bf c})=2i\lim_{\Lambda\to\infty}\Lambda [{\bf B}(\Lambda;X,T)]_{12}.
\end{equation}
When $(X,T)=(0,0)$, we have the explicit formula 
\begin{equation}
\label{B-at-00}
\mathbf{B}(\Lambda;0,0)=\begin{cases}
\mathbf{S}, & |\Lambda|<1, \\
\mathbf{S}e^{-2i \Lambda^{-1}\sigma_3}\mathbf{S}^{-1}e^{2i \Lambda^{-1}\sigma_3}, & |\Lambda|>1.
\end{cases}
\end{equation}
Then, as $\Lambda\to\infty$,
\begin{equation}
\begin{aligned}
\mathbf{S}e^{-2i \Lambda^{-1}\sigma_3}\mathbf{S}^{-1}e^{2i \Lambda^{-1}\sigma_3}&=
\mathbf{S} \left(\mathbb{I} - 2i\Lambda^{-1}\sigma_3 + \mathcal{O}(\Lambda^{-2}) \right)\mathbf{S}^{-1} \left(\mathbb{I} + 2i\Lambda^{-1}\sigma_3 + \mathcal{O}(\Lambda^{-2}) \right)\\
&=\mathbb{I} + 2i (\sigma_3 - \mathbf{S}\sigma_3\mathbf{S}^{-1})\Lambda^{-1} + \mathcal{O}(\Lambda^{-2}),
\end{aligned}
\end{equation}
and hence
\begin{equation}
\Psi(0,0;{\bf c})=-4\left[\sigma_3 - \mathbf{S}\sigma_3\mathbf{S}^{-1}\right]_{12} = 8\frac{c_1 c_2^*}{|\mathbf{c}|^2},
\end{equation}
as asserted.
\end{proof}

\begin{lemma}
\label{PsiX-lemma}
For any ${\bf c}=(c_1,c_2)\in(\mathbb{C}^*)^2$ and $\xi=i$, 
$\Psi_X(0,0;{\bf c}) = 32\frac{c_1c_2^*}{|{\bf c}|^4}(c_2c_2^*-c_1c_1^*)$.
\end{lemma}
\begin{proof}
We expand ${\bf B}(\Lambda;X,T)$ as 
\eq
{\bf B}(\Lambda;X,T) = \mathbb{I} + \frac{1}{\Lambda}{\bf B}_{-1}(X,T) + \frac{1}{\Lambda^2}{\bf B}_{-2}(X,T) + \mathcal{O}\left(\frac{1}{\Lambda^3}\right), \quad \Lambda\to\infty.
\endeq
From \eqref{K-ito-A} and \eqref{B-ito-A},
\eq
{\bf K}(\Lambda;X,T) = {\bf B}(\Lambda;X,T)e^{-i(\Lambda x+\Lambda^2 T + 2i\Lambda^{-1})\sigma_3}, \quad |\Lambda|>1.
\endeq 
Inserting this expression into 
\eq
{\bf K}_T(\Lambda;X,T) {\bf K}(\Lambda;X,T)^{-1} = \bbm -i\Lambda^2+\frac{i}{2}|\Psi(X,T)|^2 & \Lambda\Psi(X,T) + \frac{i}{2}\Psi_X(X,T) \\ -\Lambda\Psi(X,T)^* + \frac{i}{2}\Psi_X(X,T)^* & i\Lambda^2 - \frac{i}{2}|\Psi(X,T)|^2 \ebm 
\endeq
and expanding as $\Lambda\to\infty$ shows 
\eq
\label{PsiX-ito-B}
\Psi_X(0,0;{\bf c}) = 4[{\bf B}_{-2}(0,0)]_{12} - 4[{\bf B}_{-1}(0,0)]_{12}[{\bf B}_{-1}(0,0)]_{22}.
\endeq
From the formula \eqref{B-at-00} for ${\bf B}(\Lambda;0,0)$,
\eq
\label{B-entries}
[{\bf B}_{-2}(0,0)]_{12} = -\frac{8}{|{\bf c}|^2}c_1c_2^*, \quad [{\bf B}_{-1}(0,0)]_{12} = -\frac{4i}{|{\bf c}|^2}c_1c_2^*, \quad [{\bf B}_{-1}(0,0)]_{22} = -\frac{4i}{|{\bf c}|^2}c_2c_2^*.
\endeq
Combining \eqref{PsiX-ito-B} and \eqref{B-entries} gives the desired result.
\end{proof}
Now that $\Psi(0,0;{\bf c})$ and $\Psi_X(0,0;{\bf c})$ are known, higher 
derivatives $\frac{\partial^j\Psi}{\partial X^j}(0,0;{\bf c})$ can be found 
using \eqref{fourth-order-ode}.  Our next step is to prove 
\eqref{Psi-symmetry-X}.
\begin{lemma}
\label{Psi-symmetry-X-lemma}
$\Psi(-X,T;{\bf c}^*\sigma_1)=\Psi(X,T;\mathbf{c})$ for any $\mathbf{c}\in(\mathbb{C}^*)^2$.
\end{lemma}
\begin{proof}
Let ${\bf A}(\Lambda;X,T)$ be the unique solution of Riemann-Hilbert Problem 
\ref{rhp:limiting} with given $\mathbf{c}=(c_1,c_2)\in(\mathbb{C}^*)^2$ and 
$(X,T)\in\mathbb{R}^2$, and hence
\begin{equation}
2i \left[ \mathbf{A}(\Lambda;X,T;\mathbf{c}) \right]_{12}= \frac{\Psi(X,T;\mathbf{c})}{\Lambda} + \mathcal{O}\left(\frac{1}{\Lambda^2}\right),\quad \Lambda\to\infty.
\end{equation}
From the different representations given in \eqref{V-representations} with  
$\mathbf{S}=\mathbf{S}(\mathbf{c})$ and 
$\widetilde{\mathbf{S}}=\widetilde{\mathbf{S}}(\mathbf{c})$ defined as in 
\eqref{S-def} and \eqref{Stilde-def}, it follows that the jump matrix
\begin{equation}
\label{VA-def}
\mathbf{V}_{\bf A}(\Lambda;X,T;\mathbf{c}):=e^{-i(\Lambda X+\Lambda^2 T)\sigma_3}\mathbf{S}(\mathbf{c}) e^{-2i\Lambda^{-1}\sigma_3} \mathbf{S}(\mathbf{c})^{-1}e^{i(\Lambda X+\Lambda^2 T)\sigma_3}
\end{equation}
in \eqref{eq:jump-rhp-A} also has the representation
\begin{equation}
\mathbf{V}_{\bf A}(\Lambda;X,T;\mathbf{c}) = e^{-i(\Lambda X+\Lambda^2 T)\sigma_3}\widetilde{\mathbf{S}}(\mathbf{c}) e^{2i\Lambda^{-1}\sigma_3} \widetilde{\mathbf{S}}(\mathbf{c})^{-1}e^{i(\Lambda X+\Lambda^2 T)\sigma_3}.
\end{equation}
Using the identity
\begin{equation}
\mathbf{S}(\sigma_1 \mathbf{c}^{*})=\frac{1}{|\mathbf{c}^*|}\begin{pmatrix} c_2^{*} & -c_1 \\ c_1^* & c_2  \end{pmatrix} =\sigma_3 \frac{1}{|\mathbf{c}|}\begin{pmatrix} c_2^{*} & c_1 \\ -c_1^* & c_2  \end{pmatrix}\sigma_3 = \sigma_3 \mathbf{\widetilde{\mathbf{S}}(\mathbf{c})}\sigma_3,
\end{equation}
we see that
\begin{equation}
\begin{aligned}
\mathbf{V}_{\bf A}(-\Lambda;-X,T;{\bf c}^*\sigma_1)&=e^{-i(\Lambda X+\Lambda^2 T)\sigma_3}\mathbf{S}({\bf c}^*\sigma_1) e^{2i\Lambda^{-1}\sigma_3} \mathbf{S}({\bf c}^*\sigma_1)^{-1}e^{i(\Lambda X+\Lambda^2 T)\sigma_3}\\
&=\sigma_3e^{-i(\Lambda X+\Lambda^2 T)\sigma_3}\widetilde{\mathbf{S}}(\mathbf{c})e^{2i\Lambda^{-1}\sigma_3} \widetilde{\mathbf{S}}(\mathbf{c})^{-1}e^{i(\Lambda X+\Lambda^2 T)\sigma_3}\sigma_3\\
&=\sigma_3\mathbf{V}_{\bf A}(\Lambda;X,T; \mathbf{c})\sigma_3.
\end{aligned}
\end{equation}
Thus, the matrix function $\widehat{\mathbf{A}}(\Lambda;X,T;\mathbf{c}):= \sigma_3{\mathbf{A}}(-\Lambda;-X,T;{\bf c}^*\sigma_1)\sigma_3$ satisfies the jump condition
\begin{equation}
\widehat{\mathbf{A}}_+(\Lambda;X,T;\mathbf{c})=\widehat{\mathbf{A}}_-(\Lambda;X,T;\mathbf{c})\mathbf{V}_{\bf A}(\Lambda;X,T;\mathbf{c}),\quad |\Lambda|=1,
\end{equation}
which is exactly the jump condition \eqref{eq:jump-rhp-A} in Riemann-Hilbert Problem 
\ref{rhp:limiting}. Moreover, $\widehat{\mathbf{A}}(\Lambda;X,T;\mathbf{c})$ is 
well-defined and analytic away from $|\Lambda|= 1$, 
$\det\widehat{\mathbf{A}}(\Lambda;X,T;\mathbf{c})=1$ on its domain of definition, 
and $\widehat{\mathbf{A}}(\Lambda;X,T;\mathbf{c})\to \mathbb{I}$ as 
$\Lambda\to\infty$. Therefore, 
$\widehat{\mathbf{A}}(\Lambda;X,T;\mathbf{c})=\sigma_3 \mathbf{A}(-\Lambda;-X,T;{\bf c}^*\sigma_1)\sigma_3$ 
also solves Riemann-Hilbert Problem~\ref{rhp:limiting}.  The identity
\begin{equation}
\mathbf{A}(\Lambda;X,T;\mathbf{c})=\sigma_3 \mathbf{A}(-\Lambda;-X,T;{\bf c}^*\sigma_1)\sigma_3
\label{eq:A-symmetry1}
\end{equation}
follows from uniqueness of the solution to the Riemann-Hilbert problem.
Using \eqref{eq:A-symmetry1}, we write
\begin{equation}
\Psi(-X,T;{\bf c}^*\sigma_1) = 2i \lim_{\Lambda\to\infty} \Lambda \left[\mathbf{A}(\Lambda;-X,T;{\bf c}^*\sigma_1) \right]_{12}=2i \lim_{\Lambda\to\infty} \Lambda\left[\sigma_3\mathbf{A}(-\Lambda;X,T;\mathbf{c})\sigma_3 \right]_{12}.
\end{equation}
Now replacing $\Lambda$ with $-\Lambda$ gives
\begin{equation}
\Psi(-X,T;{\bf c}^*\sigma_1)  = - 2i\lim_{\Lambda\to\infty} \Lambda\left[\sigma_3\mathbf{A}(\Lambda;X,T;\mathbf{c})\sigma_3 \right]_{12} = \Psi(X,T;\mathbf{c}),
\end{equation}
as desired.
\end{proof}
Now we prove \eqref{Psi-symmetry-T}.
\begin{lemma}
\label{Psi-symmetry-T-lemma}
$\Psi(X,-T;\mathbf{c})^*=\Psi(X,T;\mathbf{c})$ if ${\bf c}\in\mathbb{R}^2$. In 
particular, $\Psi(X,0;\mathbf{c})$ is real-valued when 
$\mathbf{c}\in\mathbb{R}^2$.
\end{lemma}
\begin{proof}
Let ${\bf A}(\Lambda;X,T)$ be the unique solution of Riemann-Hilbert 
Problem~\ref{rhp:limiting} with given $\mathbf{c}\in\mathbb{R}^2$ and 
$(X,T)\in\mathbb{R}^2$. Note that $\mathbf{S}\in\mathbb{R}^{2\times 2}$ and 
that the jump matrix$\mathbf{V}_{\bf A}(\Lambda;X,T;\mathbf{c})$
defined in \eqref{VA-def} admits the symmetry
\begin{equation}
\begin{aligned}
\mathbf{V}_{\bf A}(-\Lambda^*;X,-T;\mathbf{c})^*&=e^{-i(\Lambda X+\Lambda^2 T)\sigma_3}\mathbf{S}^* e^{-2i\Lambda^{-1}\sigma_3} (\mathbf{S}^*)^{-1}e^{i(\Lambda X+\Lambda^2 T)\sigma_3}\\
&=e^{-i(\Lambda X+\Lambda^2 T)\sigma_3}\mathbf{S} e^{-2i\Lambda^{-1}\sigma_3} \mathbf{S}^{-1}e^{i(\Lambda X+\Lambda^2 T)\sigma_3}\\
&=\mathbf{V}_{\bf A}(\Lambda;X,T; \mathbf{c}).
\end{aligned}
\end{equation}
Thus, the matrix function $\widehat{\mathbf{A}}(\Lambda;X,T;\mathbf{c}):= \mathbf{A}(-\Lambda^*;X,-T;\mathbf{c})^*$ satisfies the jump condition
\begin{equation}
\widehat{\mathbf{A}}_+(\Lambda;X,T;\mathbf{c})=\widehat{\mathbf{A}}_-(\Lambda;X,T;\mathbf{c})\mathbf{V}_{\bf A}(\Lambda;X,T;\mathbf{c}),\quad |\Lambda|=1,
\end{equation}
which is exactly the jump condition \eqref{eq:jump-rhp-A} in Riemann-Hilbert 
Problem~\ref{rhp:limiting}. Moreover, $\widehat{\mathbf{A}}(\Lambda;X,T;\mathbf{c})$ 
is well-defined and analytic away from $|\Lambda|= 1$, 
$\det\widehat{\mathbf{A}}(\Lambda;X,T;\mathbf{c})=1$ on its domain of 
definition, and $\widehat{\mathbf{A}}(\Lambda;X,T;\mathbf{c})\to \mathbb{I}$ as 
$\Lambda\to\infty$. Therefore, 
$\widehat{\mathbf{A}}(\Lambda;X,T;\mathbf{c})=\mathbf{A}(-\Lambda^*;X,-T;\mathbf{c})^*$ 
also solves Riemann-Hilbert Problem~\ref{rhp:limiting}.  Uniqueness implies 
\begin{equation}
\mathbf{A}(\Lambda;X,T;\mathbf{c})=\mathbf{A}(-\Lambda^*;X,-T;\mathbf{c})^*.
\label{eq:A-symmetry2}
\end{equation}
Using \eqref{eq:A-symmetry2},
\begin{equation}
\begin{aligned}
\Psi(X,-T)&=2i\lim_{\Lambda\to\infty}\Lambda \left[ \mathbf{A}(\Lambda; X,-T)\right]_{12}=2i\lim_{\Lambda\to\infty}\Lambda \left[ \mathbf{A}(-\Lambda^*; X,T)^*\right]_{12}\\
&=-\left(2i\lim_{\Lambda\to\infty}\Lambda^* \left[ \mathbf{A}(-\Lambda^*; X,T)\right]_{12}\right)^*.
\end{aligned}
\end{equation}
Now replacing $\Lambda$ with $-\Lambda^{*}$ gives
\begin{equation}
\Psi(X,-T)=\left(2i\lim_{\Lambda\to\infty}\Lambda \left[ \mathbf{A}(\Lambda; X,T)\right]_{12}\right)^*=\Psi(X,T)^*,
\end{equation}
which completes the proof.
\end{proof}
Lemmas \ref{Psi-lemma}, \ref{PsiX-lemma}, \ref{Psi-symmetry-X-lemma}, and 
\ref{Psi-symmetry-T-lemma}, combined with the fact that 
$\Psi(X,T;{\bf c})$ satisfies the second member of the Painlev\'e-III hierarchy 
given in \eqref{fourth-order-ode}, prove Theorem \ref{thm-near-field}(b).

\subsection{The Painlev\'e-III function u(x):  Proof of Theorem \ref{thm-near-field}(c).}

Inserting the definition \eqref{u-def} into the Painlev\'e-III equation 
\eqref{PIII} and using the chain rule to convert $s$-derivatives to $X$-derivatives 
shows that \eqref{PIII} is satisfied as long as 
\eq
X\Psi\Psi_{XXX}+3\Psi\Psi_{XX}-X\Psi_X\Psi_{XX}-2(\Psi_X)^2+4\Psi^4+2X\Psi^3\Psi_X+2X\Psi^3\Psi_X = 0,
\endeq
which holds from \eqref{fourth-order-ode} since $T=0$ and $\Psi(X,0;{\bf c})$ is a 
real-valued function (since ${\bf c}\in\mathbb{R}^2$).  Plugging the expansion 
$\Psi(X,0;{\bf c}) = \Psi(0,0;{\bf c}) + \Psi_X(0,0;{\bf c})X + \mathcal{O}(X^2)$ 
as $X\to 0$ into the definition \eqref{u-def} along with $X=-\frac{1}{8}s^2$ 
immediately gives the expansion 
\eq
u(s;{\bf c}) = s + \frac{\Psi_X(0,0;{\bf c})}{8\Psi(0,0;{\bf c})}s^3 + \mathcal{O}(s^5), \quad x\to 0.
\endeq
Using the explicit forms for $\Psi(0,0;{\bf c})$ and $\Psi_X(0,0;{\bf c})$ given 
in \eqref{Psi-values} shows \eqref{u-values}.  Higher derivatives 
$u^{(j)}(0;{\bf c})$, $j\geq 5$ can now be computed using the governing 
Painlev\'e-III equation \eqref{PIII}.  This completes the proof of Theorem 
\ref{thm-near-field}.

\appendix

\section{Reformulation of the Riemann-Hilbert problem as a linear system}
\label{sec-app-linear}
We now show how to rewrite Riemann-Hilbert Problem \ref{rhp:psi-n} for 
${\bf M}^{[n]}$ as 
the linear system \eqref{linear-system} of $4n$ equations in $4n$ unknowns.  
In principal, this linear system can be solved explicitly 
for any fixed $n$, thus allowing the determination of $\psi^{[2n]}(x,t)$ via 
\eqref{recover-psi}.  In practice, the entries of the resulting coefficient 
matrix are increasingly complicated functions of $x$ and $t$ as $n$ 
increases, and the system can only be feasibly solved for at most a few 
values of $n$.  However, picking specific values of $x$ and $t$ reduces the 
problem to the inversion of a $4n\times 4n$ matrix with numerical entries, 
which can be done rapidly to any desired precision using standard numerical 
linear algebra packages for moderately large values of $n$.  This procedure 
was used to create Figures \ref{fig-far-field-2d}--\ref{fig-origin-c5}.  
Analogous methods have been used previously to study semiclassical behavior 
of the nonlinear Schr\"odinger equation \cite{LyngM:2007}, the sine-Gordon 
equation \cite{BuckinghamM:2008}, and the three-wave resonant interaction 
equations \cite{BuckinghamJM:2017}.

From \eqref{Mn-ito-Un} and \eqref{Unp1-ito-Gn-Un} we have, 
using ${\bf M}^{[0]}(\lambda;x,t)\equiv\mathbb{I}$, 
\eq
{\bf M}^{[n]}(\lambda;x,t) = \begin{cases} {\bf G}^{[n-1]}(\lambda;x,t)\cdots{\bf G}^{[0]}(\lambda;x,t), & \lambda\notin D_0,\\ {\bf G}^{[n-1]}(\lambda;x,t)\cdots{\bf G}^{[0]}(\lambda;x,t)e^{-i(\lambda x+\lambda^2 t)\sigma_3}{\bf G}^{[0]}(\lambda;0,0)^{-n}e^{i(\lambda x+\lambda^2 t)}, & \lambda\in D_0.\end{cases}
\endeq
For succinctness we define
\eq
{\bf \Pi}^{[n]}(\lambda;x,t):={\bf G}^{[n-1]}(\lambda;x,t)\cdots{\bf G}^{[0]}(\lambda;x,t).
\endeq
From the jump condition \eqref{Mn-jump}, we have 
${\bf M}_-^{[n]}(\lambda;x,t)={\bf M}_+^{[n]}(\lambda;x,t){\bf V}_{\bf M}^{[n]}(\lambda;x,t)^{-1}$. 
Since the left-hand side extends analytically into $D_0$, the right-hand side 
must as well.  Our conditions for the linear system will arise from the 
fact that 
\eq
\label{M-Vinv}
{\bf M}_+^{[n]}(\lambda;x,t){\bf V}_{\bf M}^{[n]}(\lambda;x,t)^{-1} = {\bf\Pi}^{[n]}(\lambda;x,t)e^{-i(\lambda x+\lambda^2 t)\sigma_3}{\bf G}^{[0]}(\lambda;0,0)^{-n}e^{i(\lambda x+\lambda^2t)\sigma_3}
\endeq
is analytic at $\xi$ and $\xi^*$.  We can write
\eq
{\bf G}^{[0]}(\lambda;0,0)^{-1} = \mathbb{I} + \frac{{\bf W}}{\lambda-\xi} + \frac{{\bf X}}{\lambda-\xi^*},
\endeq
where (recall $\xi=\alpha+i\beta$)
\eq
\label{W-X}
\begin{split}
{\bf W} & :=\frac{2i\beta}{|{\bf c}|^2}\bpm c_1c_1^* & c_1c_2^* \\ c_1^*c_2 & c_2c_2^* \epm = \frac{2i\beta}{|{\bf c}|^2}\bpm c_1 & c_1 \\ c_2 & c_2 \epm \bpm c_1^* & 0 \\ 0 & c_2^* \epm, \\
{\bf X} & :=\frac{-2i\beta}{|{\bf c}|^2}\bpm c_2c_2^* & -c_1c_2^* \\ -c_1^*c_2 & c_1c_1^* \epm = \frac{-2i\beta}{|{\bf c}|^2}\bpm c_2^* & c_2^* \\ -c_1^* & -c_1^* \epm \bpm c_2 & 0 \\ 0 & -c_1 \epm.
\end{split}
\endeq
By direct calculation, we have the relations
\eq
\label{WX-relations}
{\bf W}{\bf X} = {\bf X}{\bf W} = {\bf 0}, \quad {\bf X}^2=-2i\beta{\bf X}, \quad {\bf W}^2 = 2i\beta{\bf W}.
\endeq
Using these, we have
\eq
\label{G-expansion}
\begin{split}
{\bf G}^{[0]}(\lambda;0,0)^{-n} & = \mathbb{I} + \sum_{k=1}^n\bpm n \\ k \epm \left(\frac{{\bf W}^k}{(\lambda-\xi)^k} + \frac{{\bf X}^k}{(\lambda-\xi^*)^k}\right) \\ 
  & = \mathbb{I} + \sum_{k=1}^n\bpm n \\ k \epm (2i\beta)^{k-1} \left(\frac{{\bf W}}{(\lambda-\xi)^k} + (-1)^k\frac{{\bf X}}{(\lambda-\xi^*)^k}\right).
\end{split}
\endeq
Dropping the explicit dependence on $n$, we now set
\eq
{\bf L}(\lambda;x,t):={\bf\Pi}^{[n]}(\lambda;x,t)e^{-i(\lambda x+\lambda^2 t)\sigma_3}.
\endeq
Therefore, the condition that the quantity in \eqref{M-Vinv} is analytic at 
$\lambda=\xi$ and $\lambda=\xi^*$ can be reformulated as the fact that 
${\bf L}(\lambda;x,t){\bf G}^{[0]}(\lambda;0,0)^{-n}$ is analytic at 
$\lambda=\xi$ and $\lambda=\xi^*$.  We expand ${\bf L}$ (which has poles of 
order $n$ at $\xi$ and $\xi^*$) about $\lambda=\xi$ and $\lambda=\xi^*$:
\eq
\label{L-expansions}
{\bf L}(\lambda;x,t) = \sum_{j=-n}^\infty{\bf L}_j^+(x,t)(\lambda-\xi)^j, \quad {\bf L}(\lambda;x,t) = \sum_{j=-n}^\infty{\bf L}_j^-(x,t)(\lambda-\xi^*)^j.
\endeq
Here the unknown $2\times 2$ matrices ${\bf L}_j^\pm$ are independent of 
$\lambda$.  Using the expansions \eqref{G-expansion} and \eqref{L-expansions}, along 
with \eqref{WX-relations}, the analyticity conditions become 
\eq
\label{LWX-at-xi}
\left(\sum_{j=-n}^\infty{\bf L}_j^+(\lambda-\xi)^j\right)\left( \mathbb{I} + \sum_{k=1}^n\bpm n \\ k \epm (2i\beta)^{k-1} \left(\frac{{\bf W}}{(\lambda-\xi)^k} + (-1)^k\frac{{\bf X}}{(\lambda-\xi^*)^k}\right) \right)=\mathcal{O}(1), \quad \lambda\to\xi
\endeq
and 
\eq
\label{LWX-at-xi*}
\left(\sum_{j=-n}^\infty{\bf L}_j^-(\lambda-\xi^*)^j\right)\left( \mathbb{I} + \sum_{k=1}^n\bpm n \\ k \epm (2i\beta)^{k-1} \left(\frac{{\bf W}}{(\lambda-\xi)^k} + (-1)^k\frac{{\bf X}}{(\lambda-\xi^*)^k}\right) \right)=\mathcal{O}(1), \quad \lambda\to\xi^*.
\endeq
Expanding \eqref{LWX-at-xi} and collecting powers of $\lambda-\xi$ gives 
$2n$ equations for ${\bf L}_{-n}^+,\dots,{\bf L}_{n-1}$.  For example, for 
$n=1$ we obtain the two equations 
\eq
{\bf L}_{-1}^+{\bf W}={\bf 0}, \quad {\bf L}_{-1}^+ - \frac{1}{{\xi-\xi^*}}{\bf L}_{-1}^+{\bf X} + {\bf L}_0^+{\bf W} = {\bf 0}.
\endeq
Multiplying the second equation by ${\bf W}$ on the right and then using the 
first equation and the relations \eqref{WX-relations} yields the 
simplified equation ${\bf L}_0^+{\bf W}={\bf 0}$.  Indeed, using the same 
procedure of right-multiplying by ${\bf W}$ and using forward substitution 
and \eqref{WX-relations} works for general $n$ to deliver the equations 
${\bf L}_j^+{\bf W}={\bf 0}$, $-n\leq j\leq n-1$.  Similarly, expanding 
\eqref{LWX-at-xi*} and collecting powers of $\lambda-\xi^*$ gives, after 
analogous manipulations, the equations ${\bf L}_j^-{\bf X}={\bf 0}$, 
$-n\leq j\leq n-1$.  From the explicit forms \eqref{W-X} for ${\bf W}$ and 
${\bf X}$, we see these matrix equations are equivalent to the vector 
equations
\eq
\label{L-c-eqs}
{\bf L}_j^+\bpm c_1 \\ c_2 \epm = {\bf 0}, \quad {\bf L}_j^-\bpm c_2^* \\ -c_1^* \epm = {\bf 0}, \quad -n\leq j\leq n-1.
\endeq
Next, recalling 
${\bf L} = {\bf \Pi}^{[n]}e^{-i(\lambda x+\lambda^2 t)\sigma_3}$, we expand 
\eq
e^{-i(\lambda x+\lambda^2 t)\sigma_3} = \sum_{j=0}^\infty{\bf D}_j^+(x,t)(\lambda-\xi)^j = \sum_{j=0}^\infty{\bf D}_j^-(x,t)(\lambda-\xi^*)^j
\endeq
(here the coefficient matrices ${\bf D}_j^\pm$ are known, or at least can 
be computed) and 
\eq
{\bf \Pi}^{[n]}(\lambda;x,t) = \mathbb{I} + \sum_{j=1}^n\left(\frac{{\bf A}_{-j}^+(x,t)}{(\lambda-\xi)^j} + \frac{{\bf A}_{-j}^-(x,t)}{(\lambda-\xi^*)^j}\right) 
\endeq
(here the coefficient matrices ${\bf A}_{-j}^\pm$ are unknown).
If we write 
\eq
{\bf A}_{-j}^+=:\bpm r_{-j} & u_{-j} \\ * & * \epm, \quad {\bf A}_{-j}^-=:\bpm s_{-j} & v_{-j} \\ * & * \epm, 
\endeq  
(here * denotes an entry we will not need), then 
\eq
\label{recover-psi}
\psi^{[2n]}(x,t) = 2i(u_{-1}(x,t)+v_{-1}(x,t)).
\endeq
Since the equations for the entries in the top and bottom rows in 
\eqref{L-c-eqs} decouple, we only need to calculate the first rows of 
the matrices ${\bf A}_{-j}^\pm$ to reconstruct $\psi^{[2n]}(x,t)$.  This 
gives $4n$ linear equations in $4n$ unknowns, which we now express in a form 
suitable for numerical computations.

Direct calculation gives
\eq
\label{L-list1}
{\bf L}_j^+ = \sum_{k=0}^{j+n}{\bf A}_{j-k}^+{\bf D}_k^+, \quad {\bf L}_j^- = \sum_{k=0}^{j+n}{\bf A}_{j-k}^-{\bf D}_k^-, \quad -n\leq j\leq -1.
\endeq
Furthermore, if we define the constants $\gamma_{k,m}^\pm$ by the 
expansions 
\eq
\frac{1}{(\lambda-\xi^*)^m} = \sum_{k=0}^\infty\gamma_{k,m}^+(\lambda-\xi)^k, \quad \frac{1}{(\lambda-\xi)^m} = \sum_{k=0}^\infty\gamma_{k,m}^-(\lambda-\xi^*)^k,
\endeq
then we also have 
\eq
\label{L-list2}
\hspace*{-.2in}
\begin{matrix}
{\bf L}_j^+ = {\bf D}_j^+ + \sum_{k=1}^n{\bf A}_{-k}^+{\bf D}_{j+k}^+ + \sum_{\ell=0}^j\left(\sum_{m=1}^n\gamma_{j-\ell,m}^+{\bf A}_{-m}^-\right){\bf D}_\ell^+,
\\
{\bf L}_j^- = {\bf D}_j^- + \sum_{k=1}^n{\bf A}_{-k}^-{\bf D}_{j+k}^- + \sum_{\ell=0}^j\left(\sum_{m=1}^n\gamma_{j-\ell,m}^-{\bf A}_{-m}^+\right){\bf D}_\ell^-,
\end{matrix} 
\quad 0\leq j\leq n-1.
\endeq
Comparing \eqref{L-c-eqs} with \eqref{L-list1} and \eqref{L-list2}, we see 
that ${\bf D}_j^+$ only occurs multiplied by $(c_1,c_2)^\mathsf{T}$, and 
${\bf D}_j^-$ only occurs multiplied by $(c_2^*,-c_1^*)^\mathsf{T}$, so we 
define 
\eq
\bpm f_j^+(x,t) \\ g_j^+(x,t) \epm :={\bf D}_j^+(x,t)\bpm c_1 \\ c_2 \epm, \quad \bpm f_j^-(x,t) \\ g_j^-(x,t) \epm:={\bf D}_j^-(x,t)\bpm c_2^* \\ -c_1^* \epm, \quad 0\leq j\leq 2n-1.
\endeq
We also define 
\eq
\begin{split}
{\bf F}_j:=\bpm f_j^+ & 0 \\ 0 & f_j^- \epm, & \quad {\bf G}_j:=\bpm g_j^+ & 0 \\ 0 & g_j^- \epm, \\
{\bf H}_{jk}:=\bpm 0 & \displaystyle\sum_{\ell=0}^j\gamma_{\ell k}^+ f_{j-\ell}^+ \\ \displaystyle\sum_{\ell=0}^j\gamma_{\ell k}^-f_{j-\ell}^- & 0 \epm, & \quad {\bf I}_{jk}:=\bpm 0 & \displaystyle\sum_{\ell=0}^j\gamma_{\ell k}^+ g_{j-\ell}^+ \\ \displaystyle\sum_{\ell=0}^j\gamma_{\ell k}^-g_{j-\ell}^- & 0 \epm.
\end{split}
\endeq
Using these, we define the $4n\times 4n$ coefficient matrix 
\eq
{\bf T}:=\scalebox{0.8}{
$\bpm {\bf F}_0 & {\bf G}_0 & {\bf 0} & {\bf 0} & \cdots & {\bf 0} & {\bf 0} \\ {\bf F}_1 & {\bf G}_1 & {\bf F}_0 & {\bf G}_0 & \cdots & {\bf 0} & {\bf 0} \\ \vdots & \vdots & \vdots & \vdots & & \vdots & \vdots \\ {\bf F}_{n-1} & {\bf G}_{n-1} & {\bf F}_{n-2} & {\bf G}_{n-2} & \cdots & {\bf F}_0 & {\bf G}_0 \\ 
{\bf F}_n+{\bf H}_{0,n} & {\bf G}_n+{\bf I}_{0,n} & {\bf F}_{n-1}+{\bf H}_{0,n-1} & {\bf G}_{n-1}+{\bf I}_{0,n-1} & \cdots & {\bf F}_1+{\bf H}_{0,1} & {\bf G}_1+{\bf I}_{0,1} \\
{\bf F}_{n+1}+{\bf H}_{1,n} & {\bf G}_{n+1}+{\bf I}_{1,n} & {\bf F}_n+{\bf H}_{1,n-1} & {\bf G}_n+{\bf I}_{1,n-1} & \cdots & {\bf F}_2+{\bf H}_{1,1} & {\bf G}_2+{\bf I}_{1,1} \\
\vdots & \vdots & \vdots & \vdots & & \vdots & \vdots \\
{\bf F}_{2n-1}+{\bf H}_{n-1,n} & {\bf G}_{2n-1}+{\bf I}_{n-1,n} & {\bf F}_{2n-2}+{\bf H}_{n-1,n-1} & {\bf G}_{2n-2}+{\bf I}_{n-1,n-1} & \cdots & {\bf F}_n+{\bf H}_{n-1,1} & {\bf G}_n+{\bf I}_{n-1,1} 
\epm$,
}
\endeq
the $4n$-vector of unknowns
\eq
{\bf y}:=\bpm r_{-n} & s_{-n} & u_{-n} & v_{-n} & \cdots & r_{-1} & s_{-1} & u_{-1} & v_{-1} \epm,
\endeq
and the $4n$-vector of inhomogeneous terms
\eq
{\bf f}:= ( \underbrace{0 \hspace{.08in} 0 \hspace{.08in} \cdots \hspace{.08in} 0}_{2n\text{ terms}} \hspace{.08in} -f_0^+ \hspace{.08in} -f_0^- \hspace{.08in} -f_1^+ \hspace{.08in} -f_1^- \hspace{.08in} \cdots \hspace{.08in} -f_{n-1}^+ \hspace{.08in} -f_{n-1}^-).
\endeq
Here ${\bf T}$, ${\bf y}$, and ${\bf f}$ depend on $x$, $t$, $n$, ${\bf c}$, 
and $\xi$.  At last, the equations for the top rows in \eqref{L-c-eqs} can be 
recast as 
\eq
\label{linear-system}
{\bf T}{\bf y} = {\bf f},
\endeq
a form amenable to numerical computation for moderately large values of $n$.
Once ${\bf y}$ is obtained from solving this equation, the solution 
$\psi^{[2n]}(x,t)$ to the nonlinear Schr\"odinger equation is immediately 
recovered from \eqref{recover-psi}.

\section{Numerical computation of $\Psi(X,T;{\bf c})$ for arbitrary ${\bf c}\in(\mathbb{C}^*)^2$}
\label{sec-app-Psi}
A numerical procedure was developed in \cite[\S 5]{BilmanLM:2018} to compute 
the special functions $\Psi(X,T;(1,\pm 1))$ for the first time, with the aid 
of \texttt{RHPackage} \cite{Olver:website} in context of the numerical 
framework introduced in \cite{TrogdonO:2016}.  While Riemann-Hilbert 
Problem~\ref{rhp:limiting} can be solved numerically using \texttt{RHPackage} 
without contour deformations for $(X,T)$ lying in a small rectangle 
containing the origin, e.g.\@ $|X|+|T|<2$, for large values of $X$ one needs 
to deform the jump contours of this Riemann-Hilbert Problem by introducing 
lens-shaped domains to use the Deift-Zhou method of nonlinear steepest 
descent. In this section we briefly describe the deformations necessary to 
compute $\Psi(X,0;\mathbf{c})$ for arbitrary $\mathbf{c}\in(\mathbb{C}^*)^2$ 
and large values of $X$. These deformations are a generalization of what was 
employed in \cite[\S 4.1]{BilmanLM:2018} to arbitrary 
$\mathbf{c}\in(\mathbb{C}^*)^2$. Before we begin, we note it is sufficient to 
consider the case $X\geq 0$ by \eqref{Psi-symmetry-X} and $T\geq 0$ by 
\eqref{Psi-symmetry-T}.

The function $\mathbf{B}(\Lambda;X,T)$ defined in \eqref{B-ito-A} is 
unimodular and satisfies the following Riemann-Hilbert problem.
\begin{rhp}[Reformulated near-field problem]
Let $(X,T)\in\mathbb{R}^2$ be fixed but arbitrary parameters. Find the unique $2\times 2$ matrix-valued function $\mathbf{B}(\Lambda;X,T)$ with the following properties:
\begin{itemize}
\item[]\textbf{Analyticity:} $\mathbf{B}(\Lambda;X,T)$ is analytic for $|\Lambda|\neq 1$ and takes continuous boundary values from the interior and exterior of the jump contour.
\item[]\textbf{Jump condition:} The boundary values on the jump contour (oriented clockwise) follow the relation
\begin{equation}
\mathbf{B}_+(\Lambda;X,T) = \mathbf{B}_-(\Lambda;X,T) e^{-i(\Lambda X + \Lambda^2 T + 2 \Lambda^{-1})\sigma_3} \mathbf{S}^{-1}e^{i(\Lambda X + \Lambda^2 T + 2 \Lambda^{-1})\sigma_3},\quad |\Lambda|=1.
\label{eq:jump-rhp-B}
\end{equation}
\item[]\textbf{Normalization:} $\mathbf{B}(\Lambda;X,T)\to\mathbb{I}$ as 
$\Lambda\to\infty$.
\end{itemize}
\label{rhp:limiting-reformulated}
\end{rhp}

To solve Riemann-Hilbert Problem~\ref{rhp:limiting-reformulated} for $X>0$ 
large, we introduce 
\eq
v:=X^{-3/2}T, \quad z:=X^{1/2}\Lambda 
\endeq
and define $\mathbf{C}(z;X,v):=\mathbf{B}(X^{-1/2}z; X,X^{3/2}v)$.  Recall that 
$\Psi(X,T;{\bf c})$ is recovered from ${\bf B}(\Lambda;X,T)$ via 
\eqref{Psi-ito-B}, which implies
\begin{equation}
\Psi(X,X^{3/2}v)=2i X^{-1/2}\lim_{z\to\infty}z \left[\mathbf{C}(z;X,v)\right]_{12},\quad X>0,
\end{equation}
and the phase in the diagonal matrices conjugating $\mathbf{S}^{-1}$ in the 
jump condition \eqref{eq:jump-rhp-B} now has the form
\begin{equation}
\Lambda X + \Lambda^2 T + 2 \Lambda^{-1} = X^{1/2}(z+vz^2+2z^{-1})=:X^{1/2} \phi(z;v).
\end{equation}
It is clear that for each $X>0$ and $v\geq 0$, $\mathbf{C}(z;X,v)$ satisfies 
the jump condition
\begin{equation}
\mathbf{C}_+(z;X,v) = \mathbf{C}_-(z;X,v) e^{-i X^{1/2}\phi(z;v)\sigma_3} \mathbf{S}^{-1}e^{i X^{1/2}\phi(z;v)\sigma_3}, \quad z\in\Gamma,
\label{eq:jump-C}
\end{equation}
where $\Gamma$ is a Jordan curve (depending on $X$) surrounding $z=0$ with 
clockwise orientation, and $\mathbf{C}(z;X,v)$ is analytic in the complement of 
$\Gamma$. The matrix $\mathbf{C}(z;X,v)$ is unimodular and has the same 
normalization as $\mathbf{B}(\Lambda;X,T)$: $\mathbf{C}(z;X,v)\to \mathbb{I}$ 
as $z\to\infty$ for each fixed $X>0$ and $v\geq 0$.

We now proceed with introducing lens-shaped regions and deforming the jump 
contour to control the exponential factors in the jump matrix \eqref{eq:jump-C} 
as in \cite{BilmanLM:2018}. For given $v\geq 0$, the critical points of 
the phase $\phi(z;v)$ are roots of the real cubic equation
\begin{equation}
2v z^3 + z^2 -2 =0,
\label{eq:phi-cubic}
\end{equation}
which are all real and distinct if $0\leq v<54^{-1/2}$. If $v>54^{-1/2}$, however, there is a complex conjugate pair of roots and a real root. In the former case, the level curve $\Im(\phi(z;v))=0$ along which the exponential factors $ e^{\pm i X^{1/2}\phi(z;v)\sigma_3}$ are purely oscillatory has a component that is a Jordan curve enclosing the origin in the $z$-plane, and that passes through two critical points, with the remaining critical point (if $v>0$) in the exterior domain. We take this curve as the jump contour $\Gamma$ in \eqref{eq:jump-C} for $\mathbf{C}(z;X,v)$ and denote the relevant two critical points of $\phi(z;v)$ by $a<b$, where $a$ and $b$ depend on $v$. Note that when $v=0$, \eqref{eq:phi-cubic} is a real quadratic with the roots $a=-\sqrt{2}$ and $b=\sqrt{2}$. See Figure~\ref{fig:large-X-signs} for representative level curves $\Im(\phi(z;v))=0$ and the roots of \eqref{eq:phi-cubic} for different values of $v\in\mathbb{R}$.
 
\begin{figure}
\includegraphics[height=2in]{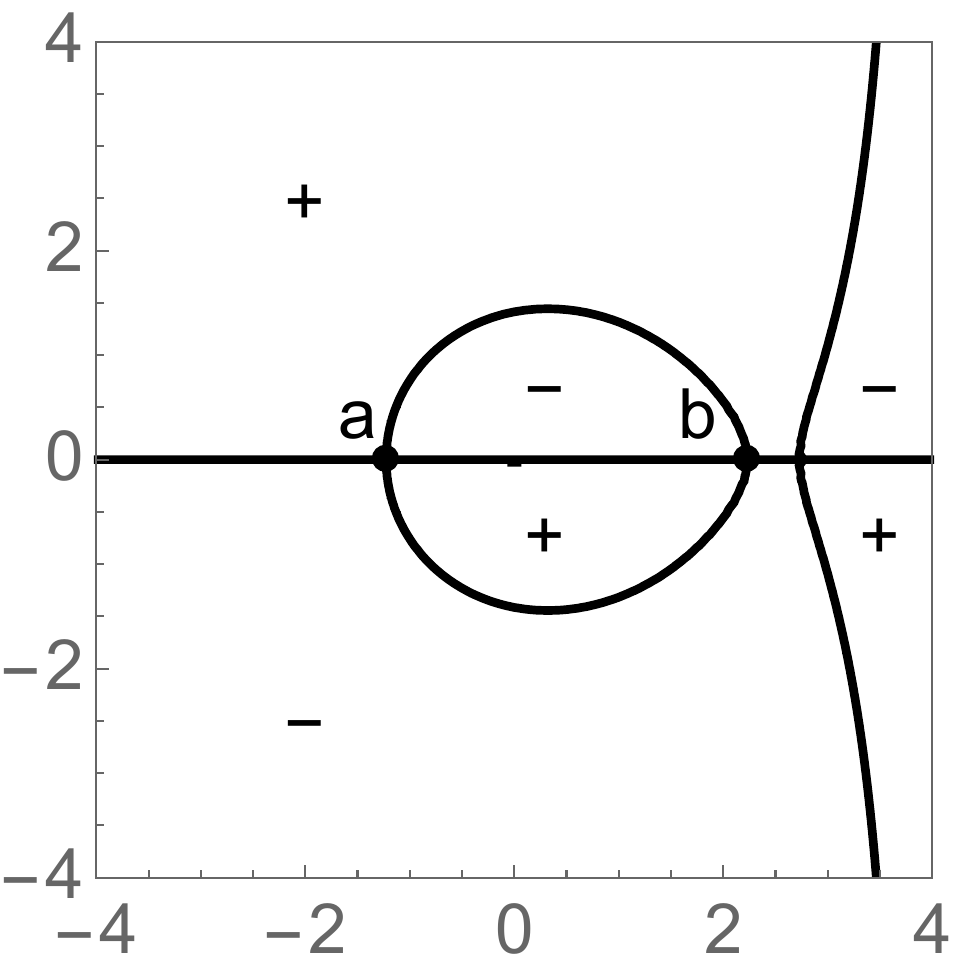}
\includegraphics[height=2in]{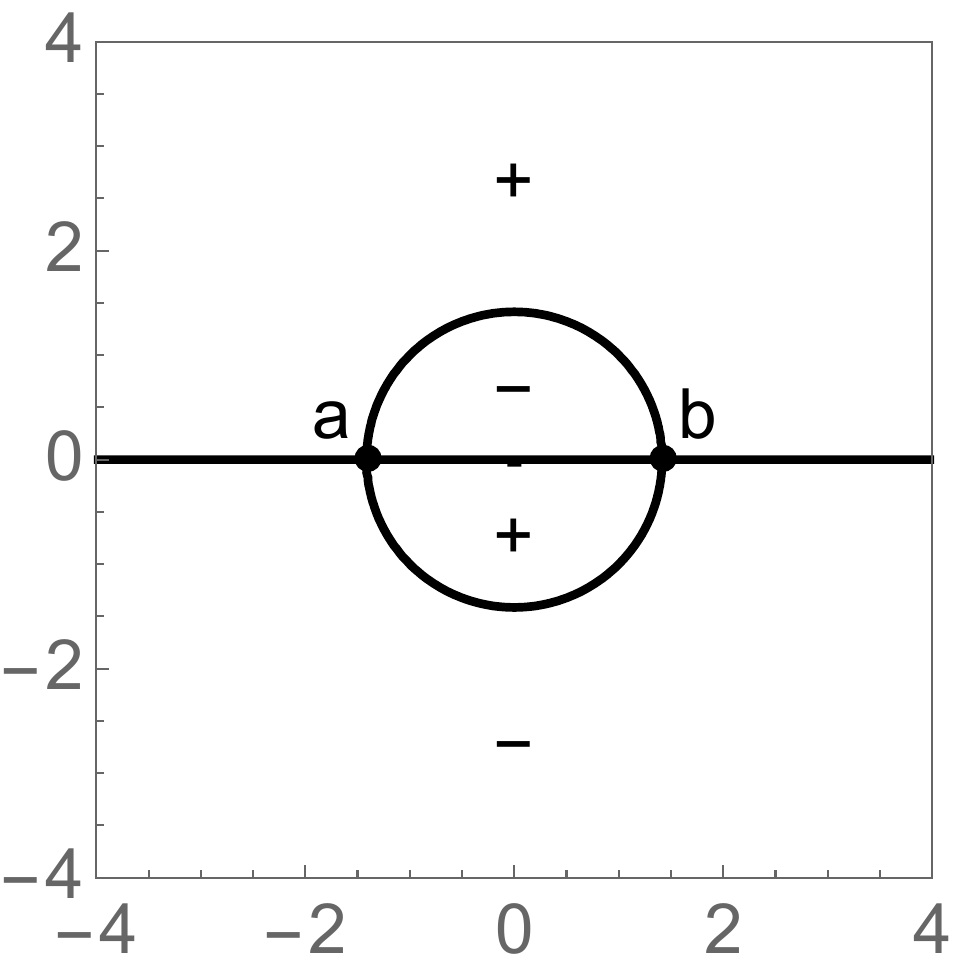}
\includegraphics[height=2in]{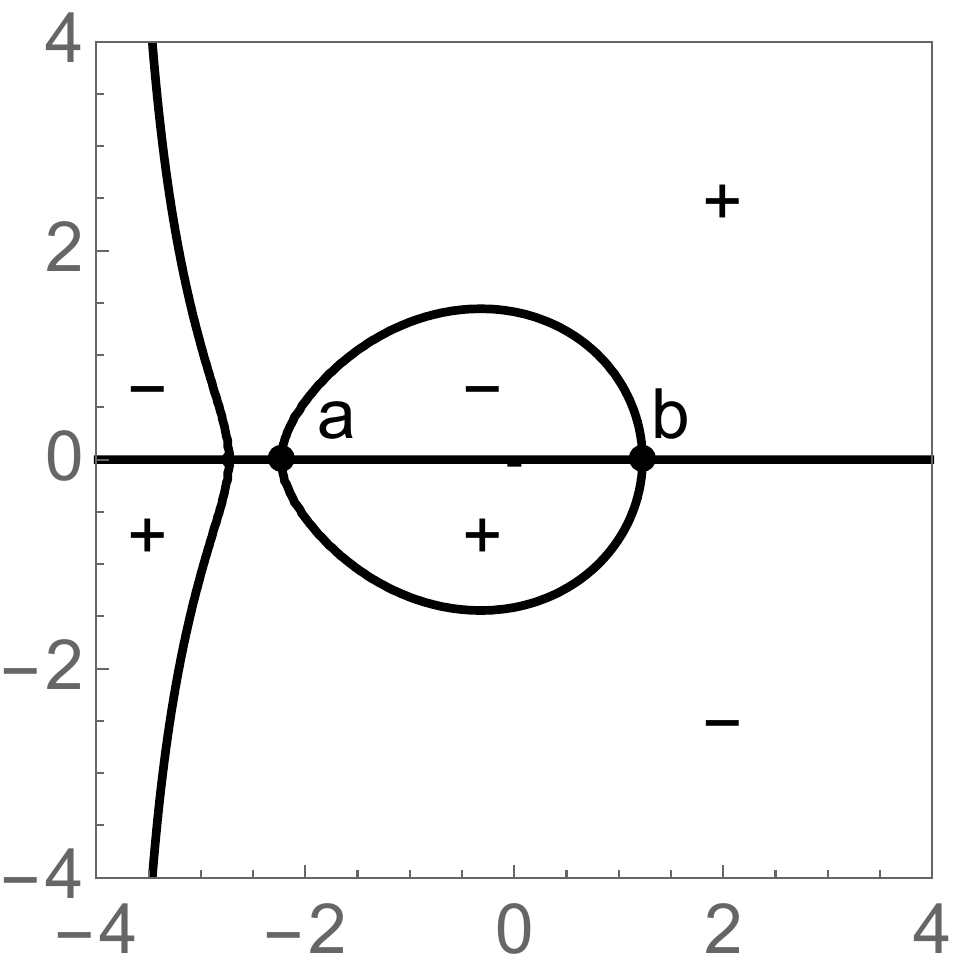}
\caption{Signature charts of $\Im(\phi(z; v))$ in the complex $z$-plane.  \emph{Left:} $v=-0.134$ (representative for $-54^{-1/2}<v<0$).  \emph{Center:} $v=0$.  \emph{Right:} $v=0.134$ (representative for $0<v<54^{-1/2}$).
}
\label{fig:large-X-signs}
\end{figure}
The real axis splits $\Gamma$ into two arcs, $\Gamma^+$ lying in the upper half 
plane and $\Gamma^{-}$ lying in the lower half plane. We deform $\Gamma^{\pm}$ 
by opening lens-shaped domains $L^{\pm}$ and $R^{\pm}$ on the left and right 
sides of $\Gamma^{\pm}$. The outer boundary arcs $C_{L}^{\pm}$ and 
$C_{R}^{\pm}$ of these regions meet the real axis at $45^\circ$ angles as shown 
in the left-hand panel of Figure~\ref{fig:large-X-deform} and on each of these 
arcs $\Im(\phi(z;v))$ has a definite sign. The line segment from $a$ to $b$ is 
denoted $I$. We label the region between $C_R^{+}$ and the real axis by 
$\Omega^{+}$, the region between $C_R^{-}$ and the real axis by $\Omega^{-}$, 
the region between $C_R^{\pm}$ and $\Gamma^{\pm}$ by $R^{\pm}$, and the region 
between $C_L^{\pm}$ and $\Gamma^{\pm}$ by $L^{\pm}$. See 
Figure~\ref{fig:large-X-deform} for illustrations of these domains and contour 
arcs.

To separate the exponential factors $e^{\pm i X^{1/2}\phi(z;v)}$ in the jump condition \eqref{eq:jump-C}, we make the following substitutions:
\begin{equation}
\mathbf{E}(z;X,v):=
\begin{cases}
\mathbf{C}(z;X,v)\begin{bmatrix}1 & 0\\\tfrac{c_2}{c_1}e^{2i X^{1/2}\phi(z;v)} & 1\end{bmatrix}, & z\in L^+, \\
\mathbf{C}(z;X,v)\left(\frac{|\mathbf{c}|}{c_1}\right)^{\sigma_3}\begin{bmatrix}1 & \tfrac{c_1 c_2^*}{|\mathbf{c}|^2}e^{-2i X^{1/2}\phi(z;v)}\\0 & 1\end{bmatrix}, & z\in R^+, \\
\mathbf{C}(z;X,v)\left(\frac{|\mathbf{c}|}{c_1}\right)^{\sigma_3}, & z\in\Omega^+, \\
\mathbf{C}(z;X,v)\left(\frac{|\mathbf{c}|}{c_1^*}\right)^{-\sigma_3}, & z\in\Omega^-,\\
\mathbf{C}(z;X,v)\left(\frac{|\mathbf{c}|}{c_1^*}\right)^{-\sigma_3}\begin{bmatrix}1 & 0\\-\tfrac{c_2 c_1^*}{|\mathbf{c}|^2}e^{2i X^{1/2}\phi(z;v)} & 1\end{bmatrix}, & z\in R^-,\\
\mathbf{C}(z;X,v)\begin{bmatrix}1 & -\frac{c_2^*}{c_1^*}e^{-2i X^{1/2}\phi(z;v)}\\0 & 1\end{bmatrix}, & z\in L^-,\\
\mathbf{C}(z;X,v), & \text{otherwise}.
\end{cases}
\end{equation}
Now $\mathbf{C}_+(z;X,v)=\mathbf{C}_-(z;X,v)$ for $z\in\Gamma^{\pm}$, so this transformation removes the jump condition across $\Gamma^{+}$ and $\Gamma^{-}$ and $\mathbf{C}(z;X,v)$ can be considered to be a well-defined analytic function on $\Gamma^{+}$ and $\Gamma^{-}$. Moreover, $\mathbf{C}(z;X,v)$ is unimodular and has the normalization $\lim_{z\to\infty} \mathbf{C}(z;X,v) = \mathbb{I}$. It therefore is the solution of the following Riemann-Hilbert problem.
\begin{rhp}[Large-$X$ problem]
Let $(X,v)\in\mathbb{R}_{>0}\times\mathbb{R}_{\geq 0}$ be fixed but arbitrary parameters. Find the unique $2\times 2$ matrix-valued function $\mathbf{E}(z;X,v)$ with the following properties:
\begin{itemize}
\item[]\textbf{Analyticity:} $\mathbf{E}(z;X,v)$ is analytic in $z$ for $z\in\mathbb{C}\setminus (C_L^- \cup C_R^- \cup I \cup C_R^+ \cup C_L^+) $, and it takes continuous boundary values from the interior and exterior of the union of the five arcs $C_L^- \cup C_R^- \cup I \cup C_R^+ \cup C_L^+$.
\item[]\textbf{Jump condition:} The boundary values on the jump contour $C_L^- \cup C_R^- \cup I \cup C_R^+ \cup C_L^+$ follow the relations
\begin{eqnarray}
\mathbf{E}_+(z;X,v) & = & \mathbf{E}_-(z;X,v)\begin{bmatrix}1 & 0\\-\tfrac{c_2}{c_1}e^{2i X^{1/2}\phi(z;v)} & 1\end{bmatrix},\quad z\in C_L^+,
\label{eq:Ejump-1} \\
\mathbf{E}_+(z;X,v) & = & \mathbf{E}_-(z;X,v)\begin{bmatrix}1 & \tfrac{c_1 c_2^*}{|\mathbf{c}|^2}e^{-2i X^{1/2}\phi(z;v)}\\0 & 1\end{bmatrix},\quad z\in C_R^+,
\label{eq:Ejump-2} \\
\mathbf{E}_+(z;X,v) & = & \mathbf{E}_-(z;X,v)\left(\frac{|\mathbf{c}|^2}{|c_1|^2}\right)^{\sigma_3},\quad z\in I, 
\label{eq:E-diagonal} \\
\mathbf{E}_+(z;X,v) & = & \mathbf{E}_-(z;X,v)\begin{bmatrix}1 & 0\\-\tfrac{c_2 c_1^*}{|\mathbf{c}|^2}e^{2i X^{1/2}\phi(z;v)} & 1\end{bmatrix},\quad z\in C_R^-,
\label{eq:Ejump-4} \\
\mathbf{E}_+(z;X,v) & = & \mathbf{E}_-(z;X,v)\begin{bmatrix}1 & -\frac{c_2^*}{c_1^*}e^{-2i X^{1/2}\phi(z;v)}\\0 & 1\end{bmatrix},\quad z\in C_L^-.
\label{eq:Ejump-5}
\end{eqnarray}
\item[]\textbf{Normalization:} $\mathbf{E}(z;X,v)\to\mathbb{I}$ as $z\to\infty$.
\end{itemize}
\label{rhp:large-X}
\end{rhp}
Since $\Im(\phi(z;v))>0$ for $z\in C^{+}_L\cup C^{-}_R$ and $\Im(\phi(z;v))<0$ for $z\in C^{+}_R\cup C^{-}_L$, the jump matrices supported on these four contour arcs are exponentially small (as $X\to+\infty$) perturbations of the identity matrix uniformly except near the end points $a$ and $b$. For the jump contours in Riemann-Hilbert Problem~\ref{rhp:large-X} with $v=0$, see the right panel of Figure~\ref{fig:large-X-deform}. The picture is completely analogous for $0<v<54^{-1/2}$. 

\begin{figure}
\includegraphics[height=2in]{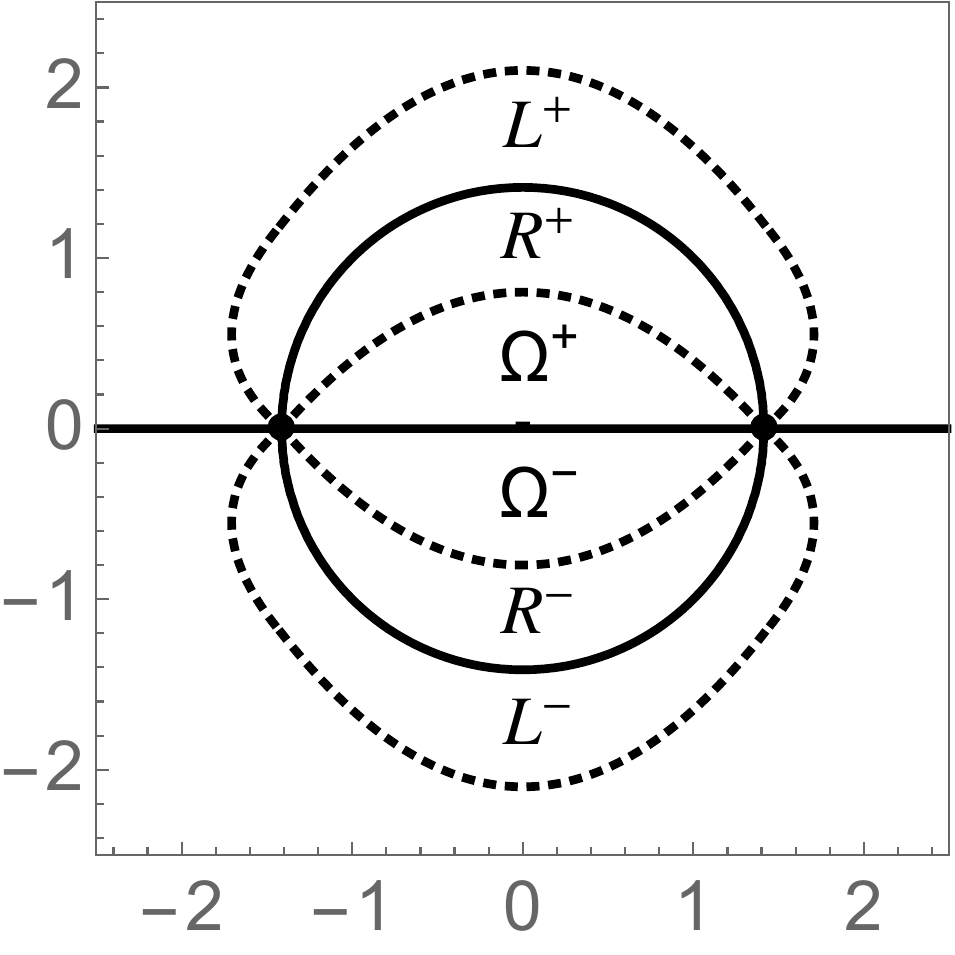}
\includegraphics[height=2in]{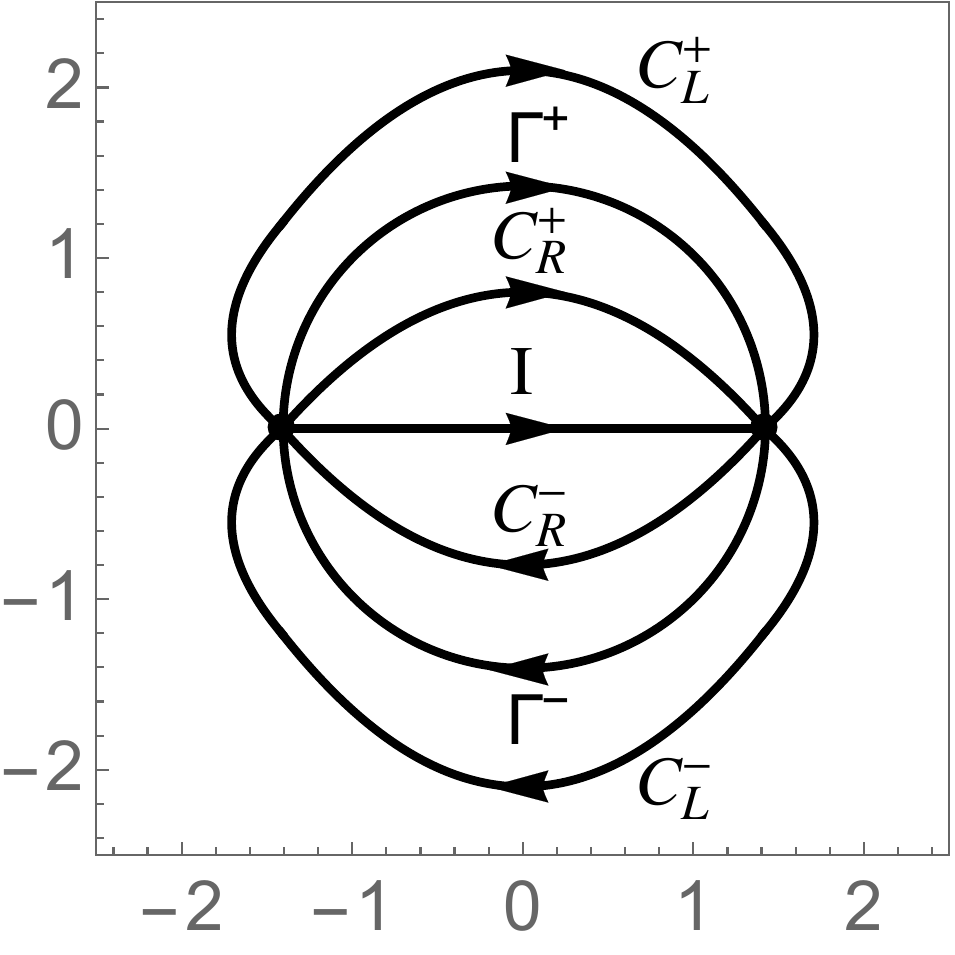}
\caption{\emph{Left:} The regions $\Omega^{\pm}$, $R^{\pm}$, and $L^{\pm}$ used in the definition of ${\bf E}(z;X,v)$.  \emph{Right:} The jump contours for Riemann-Hilbert Problem \ref{rhp:large-X}.  In both plots $v=0$.}
\label{fig:large-X-deform}
\end{figure}

Riemann-Hilbert Problem~\ref{rhp:large-X} can be used to numerically compute 
$\Psi(X,T)$ for values of $X> 54^{1/3}T^{2/3}$.  We now focus our
attention on the case $T=0$ (hence $v=0$) and $\mathbf{c}\in\mathbb{R}^2$.  For 
cross-validation of the numerical procedure described here along with further 
details, see \cite[\S 5]{BilmanLM:2018}.

As discussed in \cite[\S 5]{BilmanLM:2018}, although the jump matrices on the 
four arcs $C_L^- \cup C_R^- \cup  C_R^+ \cup C_L^+$ become exponentially small 
perturbations of the identity matrix as $X\to\infty$, their Sobolev norms 
(differentiation with respect to $z$) grow.  This presents a numerical 
challenge that is overcome in \texttt{RHPackage} by a rescaling algorithm (see 
\cite{TrogdonO:2016} for details). Thus, to have a procedure that is 
asymptotically robust as $X>0$ becomes large, one has to remove the so-called 
connecting jump condition \eqref{eq:E-diagonal} on the line segment $I$. To 
this end, we introduce the parametrix
\begin{equation}
\mathbf{\Delta}(z;v):=\left( \frac{z-a(v)}{z-b(v)} \right)^{i p \sigma_3},\quad p:=\frac{\ln\left(\frac{|\mathbf{c}|^2}{|c_1|^2}\right)}{2\pi}>0,\quad z\in\mathbb{C}\setminus I,
\end{equation}
which exactly satisfies the jump condition
\begin{equation}
\mathbf{\Delta}_+(z;v)=\mathbf{\Delta}_-(z;v)\left(\frac{|\mathbf{c}|^2}{|c_1|^2}\right)^{\sigma_3},\quad z\in I,
\end{equation}
and is normalized as $\mathbf{\Delta}(z;v)\to\mathbb{I}$ as $z\to\infty$.  
Thus, setting 
$\widehat{\mathbf{C}}(z;X,v):=\mathbf{C}(z;X,v)\mathbf{\Delta}(z;v)^{-1}$ for 
$z\in\mathbb{C}\setminus I$ removes the connecting jump condition across $I$ as 
desired and conjugates the existing other jump matrices given in 
\eqref{eq:Ejump-1}--\eqref{eq:Ejump-5} by $\mathbf{\Delta}(z;v)$. This comes 
with the cost of introducing bounded singularities in the jump matrices at 
$z=a$ and $z=b$ since $\mathbf{\Delta}(z;v)$ has bounded singularities at these 
points. To remedy this, we place small circles centered at $z=a$ and $z=b$ and 
transfer the jump conditions on the line segments inside these circles to jump 
conditions on arcs of these circles connecting the endpoints of these line 
segments. This successfully removes the aforementioned singular jump 
conditions, but introduces jump matrices on the small circles centered at 
$z=a(v)$ and $z=b(v)$ whose components now grow exponentially as $X\to+\infty$. 
Observe that for $p=a,b$,
\begin{equation}
\phi(z;v) - \phi(p;v) = \frac{\phi''(p;v)}{2}(z-p)^2 + \mathcal{O}((z-p)^3), \quad z\to p.
\end{equation}
Therefore,
\begin{equation}
e^{\pm i |X|^{1/2}\phi(z;v)}=\mathcal{O}(1),\quad X\to+\infty
\end{equation}
if $|z-p|=\mathcal{O}(|X|^{-1/4})$ as $X\to+\infty$ for both $p =a(v)$ and $p=b(v)$. In order to overcome the growth of these factors, we scale the common radius of these circles by $|X|^{-1/4}$ as $X$ becomes large. As noted in \cite[\S 5]{BilmanLM:2018}, while shrinking the circles at a faster rate ensures boundedness of the exponentials supported on them, it also moves the support of the jump matrices closer to singularities at a faster rate and hence should be avoided. The jump contours of the Riemann-Hilbert problem used to compute $\Psi(X,0)$ numerically for large values of $X>0$ are given in Figure~\ref{fig:large-X-truncated-contours}. In practice, the jump contours are truncated if the jump matrices supported on them differ from the identity matrix by at most machine precision. For more details see \cite[Chapter 2 and Chapter 7]{TrogdonO:2016} and \cite[\S 5]{BilmanLM:2018}.
\begin{figure}[ht]
\begin{center}
\includegraphics[height=2in]{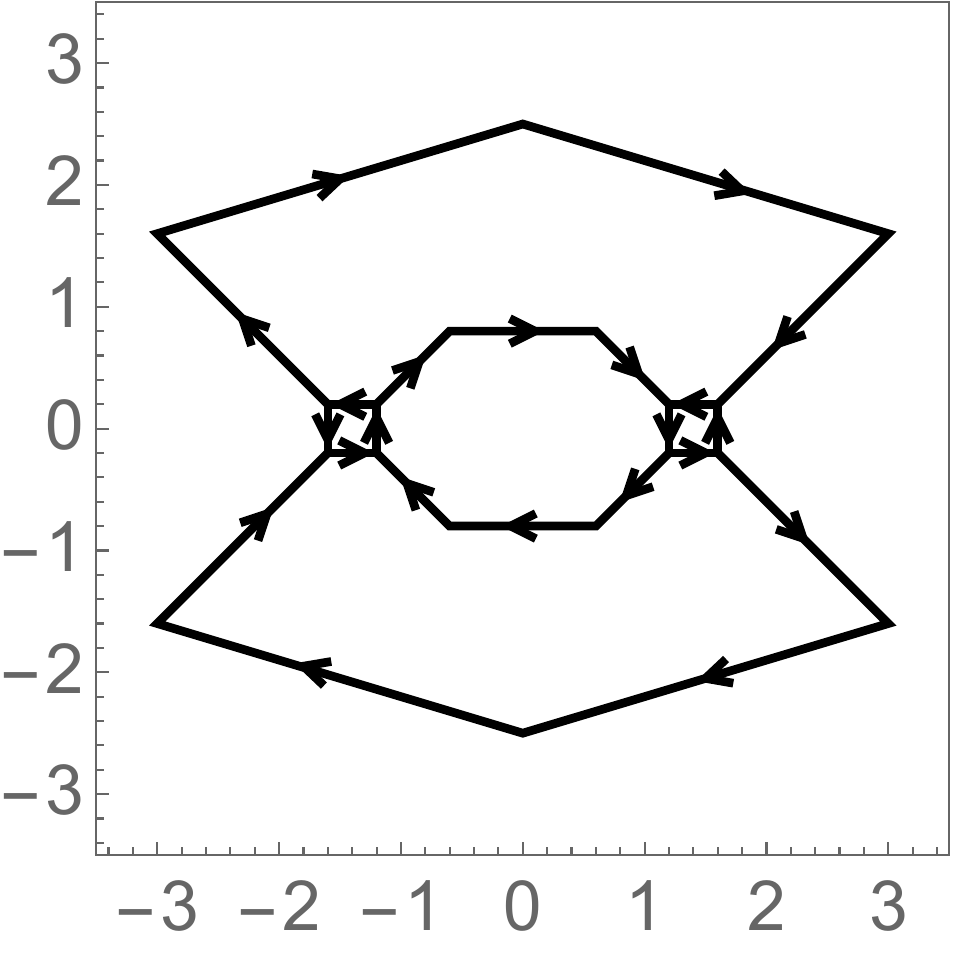}
\includegraphics[height=2in]{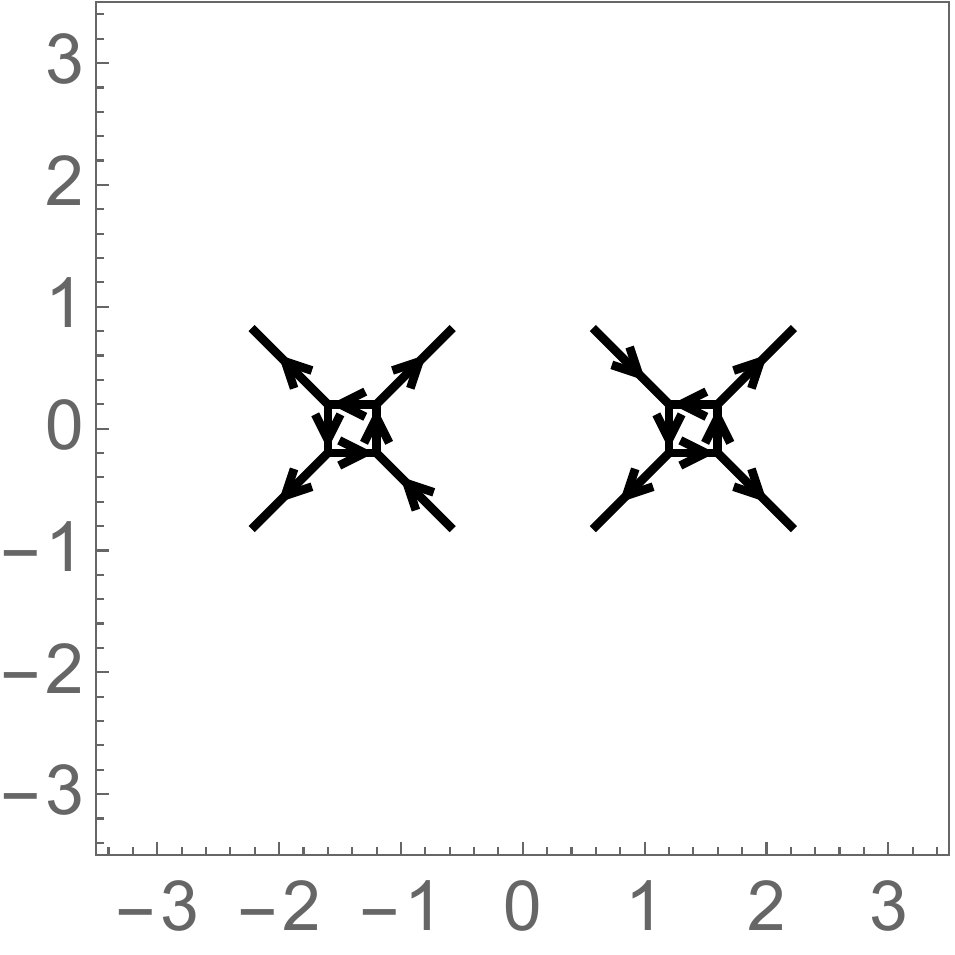}
\end{center}
\caption{\emph{Left:} Jump contours used in the numerical solution of the Riemann-Hilbert problem satisfied by $\widehat{\mathbf{C}}(z;X,v=0)$, which is asymptotically and numerically well-adapted for large $X>0$.  \emph{Right:}  Truncated jump contours that are used in practice if $X>0$ is large. For both plots, $X=2000$ and $v=0$.}
\label{fig:large-X-truncated-contours}
\end{figure}


\begin{thebibliography}{10}

\bibitem{AblowitzCTV:2000}
M. Ablowitz, S. Chakravarty, A. Trubatch, and J. Villarroel,
A novel class of solutions of the non-stationary Schr\"odinger and the Kadomtsev-Petviashvili I equations,
\textit{Phys. Lett. A}
{\bf 267},
132--146
(2000).

\bibitem{AktosunDv:2007}
T. Aktosun, F. Demontis, and C. van der Mee,
Exact solutions to the focusing nonlinear Schr\"odinger equation,
\textit{Inverse Problems}
{\bf 23},
2171--2195
(2007).

\bibitem{AktosunDv:2010}
T. Aktosun, F. Demontis, and C. van der Mee,
Exact solutions to the sine-Gordon equation,
\textit{J. Math. Phys.}
{\bf 51},
123521
(2010).

\bibitem{BealsC:1984}
R. Beals and R. Coifman,
Scattering and inverse scattering for first order systems,
\textit{Comm. Pure Appl. Math.}
{\bf 37},
39--90
(1984).

\bibitem{BenneyN:1967}
D. Benney and A. Newell,
The propagation of nonlinear wave envelopes,
\textit{J. Math. and Phys.}
{\bf 46},
133--139
(1967).

\bibitem{BertolaT:2013}
M. Bertola and A. Tovbis, 
Universality for the focusing nonlinear Schr\"odinger equation at the gradient catastrophe point: rational breathers and poles of the tritronqu\'ee solution to Painlev\'e I,
\textit{Comm. Pure Appl. Math.} 
{\bf 66}, 
678--752
(2013).

\bibitem{BilmanLM:2018}
D. Bilman, L. Ling, and P. Miller,
Extreme superposition:  rogue waves of infinite order and the Painlev\'e-III hierarchy,
arXiv:1806.00545
(2018).

\bibitem{BilmanM:2017}
D. Bilman and P. Miller,
A robust inverse scattering transform for the focusing nonlinear Schr\"odinger equation, 
arXiv:1710.06568 
(2017).  
To appear in \emph{Comm. Pure Appl. Math.}

\bibitem{BuckinghamJM:2017}
R. Buckingham, R. Jenkins, and P. Miller,
Semiclassical soliton ensembles for the three-wave resonant interaction equations,
\textit{Comm. Math. Phys.}
\textbf{354},
1015--1100
(2017).

\bibitem{BuckinghamM:2008}
R. Buckingham and P. Miller,
Exact solutions of semiclassical non-characteristic Cauchy problems for the sine-Gordon equation,
\textit{Physica D}
\textbf{237},
2296--2341
(2008).

\bibitem{BuckinghamM:2013}
R. Buckingham and P. Miller,
Large-degree asymptotics of rational Painlev\'e-II functions: noncritial behaviour,
\textit{Nonlinearity}
\textbf{27},
2489--2577
(2014).

\bibitem{ChiaoGT:1964}
R. Chiao, E. Garmire, and C. Townes, 
Self-trapping of optical beams,
\textit{Phys. Rev. Lett.} 
\textbf{13},
479--482
(1964).

\bibitem{DeiftZ:1993}
P. Deift and X. Zhou,
A steepest descent method for oscillatory Riemann-Hilbert problems. Asymptotics for the MKdV equation,
\textit{Ann. of Math. (2)}
\textbf{137},
295--368
(1993).

\bibitem{ElH:2016}
G. El and M. Hoefer,
Dispersive shock waves and modulation theory,
\textit{Physica D}
\textbf{333},
11--65
(2016).

\bibitem{FuchssteinerA:1987}
B. Fuchssteiner and R. Aiyer,
Multisolitons, or the discrete eigenfunctions of the recursion operator of 
non-linear evolution equations: II. Background,
\textit{J. Phys. A}
{\bf 20},
375--388
(1987).

\bibitem{GagnonS:1994}
L. Gagnon and N. Sti\'evenart,
$N$-soliton interaction in optical fibers: the multiple-pole case,
\textit{Opt. Lett.}
\textbf{19},
619--621
(1994).

\bibitem{KamvissisMM:2003}
S. Kamvissis, K. McLaughlin, and P. Miller, 
\textit{Semiclassical soliton ensembles for the focusing nonlinear Schr\"odinger equation},
Annals of Mathematics Studies 
\textbf{154},
Princeton University Press, 
Princeton, NJ, 
2003.

\bibitem{KuangZ:2017}
Y. Kuang and J. Zhu,
The higher-order soliton solutions for the coupled Sasa-Satsuma system via 
the $\overline{\partial}$-dressing method,
\textit{Appl. Math. Lett.}
{\bf 66},
47--53
(2017).

\bibitem{LingFZ:2016}
L. Ling, B. Feng, and Z. Zhu, 
Multi-soliton, multi-breather and higher order rogue wave solutions to the complex short pulse equation,
\textit{Phys. D} 
\textbf{327}, 
13--29,
(2016). 

\bibitem{LyngM:2007}
G. Lyng and P. Miller, 
The $N$-soliton of the focusing nonlinear Schr\"odinger equation for $N$ large,
\textit{Comm. Pure Appl. Math.} 
\textbf{60}, 
951--1026 
(2007).

\bibitem{Martines:2017}
T. Martines,
Generalized inverse scattering transform for the nonlinear Schr\"odinger 
equation for bound states with higher multiplicities,
\textit{Electron. J. Differential Equations}
\textbf{179},
1--15
(2017).

\bibitem{MatveevS:1991}
V. Matveev and M. Salle,
\textit{Darboux transformations and solitons},
Springer-Verlag,
Berlin,
1991.

\bibitem{Olmedilla:1987}
E. Olmedilla,
Multiple pole solutions of the non-linear Schr\"odinger equation,
\textit{Physica D}
{\bf 25},
330--346
(1987).

\bibitem{Olver:website}
S. Olver,
\texttt{RHPackage},
\texttt{http://www.maths.usyd.edu.au/u/olver/projects/RHPackage.html},
2011.

\bibitem{Poppe:1983}
C. P\"oppe,
Construction of solutions of the sine-Gordon equation by means of Fredholm determinants,
\textit{Physica D}
{\bf 9},
103--139
(1983).

\bibitem{Sakka:2009}
A. Sakka,
Linear problems and hierarchies of Painlev\'e equations,
\textit{J. Phys. A}
{\bf 42},
025210
(2009).

\bibitem{Schiebold:2017}
C. Schiebold,
Asymptotics for the multiple pole solutions of the nonlinear Schr\"odinger equation,
\textit{Nonlinearity}
{\bf 30},
2930--2981
(2017).

\bibitem{Shchesnovich:2003}
V. Shchesnovich and J. Yang, 
Higher-order solitons in the $N$-wave system,
\textit{Stud. Appl. Math.} 
{\bf 110},
297--332
(2003).

\bibitem{Suleimanov:2017}
B. Suleimanov,
Effect of a small dispersion on self-focusing in a spatially one-dimensional case,
\textit{JETP Lett.}
{\bf 106},
400--405
(2017).

\bibitem{Tanaka:1975}
S. Tanaka,
Non-linear Schr\"odinger equation and modified Korteweg-de Vries equation; construction of solutions in terms of scattering data,
\textit{Publ. RIMS, Kyoto Univ.}
{\bf 10},
329--357
(1975).

\bibitem{TrogdonO:2016}
T. Trogdon and S. Olver,
\textit{Riemann-Hilbert problems, their numerical solution, and the computation of nonlinear special functions},
SIAM,
Philadelphia,
2016.

\bibitem{TsuruW:1984}
H. Tsuru and M. Wadati,
The multiple pole solutions of the sine-Gordon equation,
\textit{J. Phys. Soc. Jpn.}
{\bf 53},
2908--2921
(1984).

\bibitem{VillarroelA:1999}
J. Villarroel and M. Ablowitz,
On the discrete spectrum of the nonstationary Schr\"odinger equation and multipole lumps of the Kadomtsev-Petviashvili I equation,
\textit{Comm. Math. Phys.}
{\bf 207},
1--42
(1999).

\bibitem{Vinh:2017}
N. Vinh,
Strongly interacting multi-solitons with logarithmic relative distance for the gKdV equation,
\textit{Nonlinearity}
{\bf 30},
4614--4648
(2017).

\bibitem{WadatiO:1982}
M. Wadati and K. Ohkuma,
Multiple-pole solutions of the modified Korteweg-de Vries equation,
\textit{J. Phys. Soc. Jpn.}
{\bf 51},
2029--2035
(1982).

\bibitem{Zakharov:1968}
V. Zakharov,
Stability of periodic waves of finite amplitude on the surface of a deep fluid,
\textit{J. Applied Mech. Tech. Phys.} 
{\bf 9},
190--194
(1968).

\bibitem{ZakharovS:1972}
V. Zakharov and A. Shabat,
Exact theory of two-dimensional self-focusing and one-dimensional self-modulation of waves in nonlinear media,
{\it Soviet Physics JETP}
{\bf 34},
62--69
(1972).
Translated from 
{\it Z. Eksper. Teoret. Fiz.}
{\bf 61},
118--134 
(1971).

\end{thebibliography}
\end{document}